\documentclass[,nonblindrev]{informs33}
\OneAndAHalfSpacedXII

\usepackage{natbib}
 \bibpunct[, ]{(}{)}{,}{a}{}{,}%
 \def\bibfont{\small}%

\usepackage{xspace}
\usepackage[T1]{fontenc}
\usepackage{graphicx}
\usepackage{caption}
\usepackage{subcaption}
\usepackage{diagbox}
\usepackage{xcolor}
\usepackage{xfrac}
\usepackage{type1cm}
\usepackage{booktabs}
\usepackage{bm}
\usepackage{dsfont}
\usepackage{algorithm}
\usepackage{algpseudocode}
\usepackage{thm-restate}
\newtheorem{theorem}{Theorem}[section]
\newtheorem{corollary}[theorem]{Corollary}
\newtheorem{lemma}[theorem]{Lemma}

\def\det{\mathop{\rm det}}

\def\proj{\mathop{\rm proj}}
\def\Re{\mathbb{R}}
\def\tr{^{\intercal}}
\def\mynorm#1{\ensuremath{\left|\left| #1 \right|\right|}}
\def\one{\ensuremath{\mathds{1}}}

\def\mycolora{black}%{red}
\def\mycolorb{black}%{orange}
\def\mycolorc{black}%{blue}
\def\mycolord{black}%{green}
\def\mycolore{black}%{cyan}
\def\mycolorf{black}%{magenta}
\def\mycolorg{black}%{purple}

\def\ie{\textit{i.e.},\xspace}
\newcommand{\gurobi}{{\color{\mycolora}\text{Gurobi}}}
\newcommand{\onlsup}{{\color{\mycolora}\text{Online Supplement}}}
\def\yes{{\color{\mycolorb}{\tt yes}}}
\def\no{{\color{\mycolorb}{\tt no}}}

\def\line{{\color{\mycolora}\ell}}
\def\tmin{{\color{\mycolora}l}}
\def\tmax{{\color{\mycolora}u}}
\def\start{{\color{\mycolora}x}}
\def\dir{{\color{\mycolora}d}}
\def\radius{{\color{\mycolora}r}}
\def\normp{{\color{\mycolora}p}}
\def\normq{{\color{\mycolora}q}}
\def\theneedle#1#2{{\color{\mycolorc}\mathcal{N}_{#1,#2}}}
\def\vachi{{\color{\mycolord}\chi}}
\def\vatau{{\color{\mycolord}\tau}}
\def\varho{{\color{\mycolord}\rho}}
\def\vaxi{{\color{\mycolord}\xi}}
\def\vasigma{{\color{\mycolord}\sigma}}

\def\Systwoneedles{{\color{\mycolore}FR}}
\def\dualvar{{\color{\mycolore}v}}
\def\dualvarzer{{\color{\mycolore}v^0}}
\def\dualvarzerbar{{\color{\mycolore}\bar{v}^0}}
\def\dualvarzerp{{\color{\mycolore}v^0_+}}
\def\dualvarone{{\color{\mycolore}v^1}}
\def\dualvartwo{{\color{\mycolore}v^2}}
\def\dualvarthr{{\color{\mycolore}v^3}}
\def\dualvarfou{{\color{\mycolore}v^4}}
\def\dualvarfiv{{\color{\mycolore}v^5}}
\def\dualvarsix{{\color{\mycolore}v^6}}
\def\dualvarsev{{\color{\mycolore}v^7}}
\def\nv#1{\ensuremath{{\color{\mycolora}n_v(#1)}}}

\def\SysneedleP{{\color{\mycolore}FP}}
\def\polylhs{{\color{\mycolorf}A}}
\def\polyrhs{{\color{\mycolorf}b}}
\def\dualvarw{{\color{\mycolore}w}}
\def\dualvarwone{{\color{\mycolore}w^1}}
\def\dualvarwtwo{{\color{\mycolore}w^2}}
\def\dualvarwtwop{{\color{\mycolore}w^2_+}}
\def\dualvarwthr{{\color{\mycolore}w^3}}
\def\dualvarwfou{{\color{\mycolore}w^4}}
\def\dualvarwfiv{{\color{\mycolore}w^5}}
\def\nw#1{\ensuremath{{\color{\mycolora}n_w(#1)}}}

\def\sn{\ensuremath{\mathcal{{\color{\mycolora}S}}}}
\def\vox{\ensuremath{\mathcal{{\color{\mycolora}V}}}}
\def\dwell{\ensuremath{\mathcal{{\color{\mycolora}D}}}}
\def\ints{\ensuremath{\mathcal{{\color{\mycolora}T}}}}
\def\tumor{{\color{\mycolora}c}}
\def\oars{{\color{\mycolora}o}}
\def\Tmax{\ensuremath{{\color{\mycolora}T_{\max}}}}
\def\Nmax{\ensuremath{{\color{\mycolora}N_{\max}}}}
\def\prescr{\ensuremath{{\color{\mycolorc}\delta}}}
\def\doselinear{{\color{\mycolora}\eta}}
\def\dwtimes{\ensuremath{{\color{\mycolord}t}}}
\def\dosevar{\ensuremath{{\color{\mycolord}y}}}
\def\pickneedle{\ensuremath{{\color{\mycolord}w}}}

\def\cand{\ensuremath{\mathcal{{\color{\mycolora}G}}}}
\def\oncand{\ensuremath{\mathcal{{\color{\mycolora}C}}}}
\def\Up{\ensuremath{{\color{\mycolora}\zeta}}}
\def\ndlvar{\ensuremath{{\color{\mycolord}w}}}
\def\Imax{\ensuremath{{\color{\mycolora}I_{\max}}}}
\def\dn{\ensuremath{\mathcal{{\color{\mycolorc}Q}}}}
\def\en{\ensuremath{\mathcal{{\color{\mycolorc}E}}}}
\def\points{\ensuremath{\mathcal{{\color{\mycolora}I}}}}
\def\needles{\ensuremath{\mathcal{{\color{\mycolora}J}}}}
\def\polyhyd{\ensuremath{\mathcal{{\color{\mycolora}M}}}}
\def\dwpt{{\color{\mycolora}\omega}}
\def\mdla{{\color{\mycolord}\square}}
\def\mdlb{{\color{\mycolord}\triangle}}
\def\asgt{\ensuremath{{\color{\mycolord}q}}}
\def\stept{\ensuremath{{\color{\mycolord}\iota}}}
\def\Dis{{\color{\mycolora}D}}
\def\Mmdl{\ensuremath{{\color{\mycolora}M(\mdla,\mdlb)}}}
\def\Mdlassg{\ensuremath{{\color{\mycolorc}A}}}
\def\Rmdlorig#1{\ensuremath{{\color{\mycolora}R^{#1}(\mdla,\mdlb)}}}
\def\Rmdl#1{\ensuremath{{\color{\mycolora}R^{#1}(\mdla,\mdlb,\dn_e)}}}
\def\RmdlI#1{\ensuremath{{\color{\mycolora}R_1^{#1}(\mdla,\mdlb,\dn_e)}}}
\def\RmdlII#1{\ensuremath{{\color{\mycolora}R_2^{#1}(\mdla,\mdlb,\dn_e)}}}
\def\RmdlIII#1{\ensuremath{{\color{\mycolora}R_3^{#1}(\mdla,\mdlb,\dn_e)}}}
\def\polycons{\ensuremath{{\color{\mycolora}p}}}
\def\totalpolycons{\ensuremath{{\color{\mycolora}p}}}
\def\fix#1{\underline{\mathsf{#1}}}
\def\init#1{\tilde{\mathsf{#1}}}
\def\optval{{\color{\mycolorg}\varphi}}
\def\totoptval{{\color{\mycolorg}\psi}}
\def\K{\ensuremath{{\color{\mycolora}k}}} 

\def\circle{{\color{\mycolora}\delta}}
\def\MC{{\color{\mycolorc}MC}}
\def\star{{\color{\mycolora}*}}

\newcommand{\qedjp}{\em $\Halmos$}

\begin{document}

\ARTICLEAUTHORS{%
\AUTHOR{Nasim Mirzavand Boroujeni}
\AFF{Department of Industrial and Systems Engineering, University of Minnesota, 207 Church Street SE, Minneapolis, MN 55455, USA}

\AUTHOR{Jean-Philippe P. Richard}
\AFF{Department of Industrial and Systems Engineering, University of Minnesota, 207 Church Street SE, Minneapolis, MN 55455, USA} %\EMAIL{jrichar@umn.edu}.}

\AUTHOR{David Sterling}
\AFF{Department of Radiation Oncology, University of Minnesota, 516 Delaware Street SE, Minneapolis MN, 55455, USA}

\AUTHOR{Christopher Wilke}
\AFF{Department of Radiation Oncology, University of Minnesota, 516 Delaware Street SE, Minneapolis MN, 55455, USA}
}

\RUNAUTHOR{Mirzavand Boroujeni, Richard, Sterling, and Wilke}

\RUNTITLE{Optimizing needle placement in 3D-printed masks for HDR-BT}

\TITLE{Optimization models for needle placement in
3D-printed masks for high dose rate brachytherapy}

\ABSTRACT{
High dose rate brachytherapy (HDR-BT) is an appealing treatment option for superficial cancers that permits the delivery of higher local doses than other radiation modalities without a significant increase in toxicity. 
In order for HDR-BT to be used in these situations, needles through which the radiation source is passed must be strategically placed in close proximity to the patient's body. 
Currently, this crucial step is performed manually by physicians or medical physicists. 
The use of 3D-printed masks customized for individual patients has been advocated as a way to more precisely and securely position these needles, with the potential of producing better and safer treatment plans. 
In this paper, we propose optimization approaches for positioning needles within 3D-printed masks for HDR-BT, focusing on skin cancers. 
We numerically show that the models we propose efficiently generate more homogeneous plans than those derived manually and provide an alternative to manual placement that can save planning time and enhance plan quality.
}

\KEYWORDS{skin cancer, high dose rate brachytherapy,  needle placement, 3D printing, mask, homogeneity.}
% \HISTORY{}

\maketitle

%%%%%%%%%%%%%%
%%%%%%%%%%%%%%%%%%%%%%%%%%%%
%%%%%%%%%%%%%%%%%%%%%%%%%%%%%%%%%%%%%%%%%%%%%%%%%%%%%%%%
\section{Introduction}
\label{section:intro}
%%%%%%%%%%%%%%%%%%%%%%%%%%%%%%%%%%%%%%%%%%%%%%%%%%%%%%%%
%%%%%%%%%%%%%%%%%%%%%%%%%%%%
%%%%%%%%%%%%%%

%%%%%%%%%%%%%%%%%%%%%%%%%%%%
%%%%%%%%%%%%%%%%%%%%%%%%%%%%%%%%%%%%%%%%%%%%%%%%%%%%%%%%
\subsection{Background}
\label{section:intro:background}
%%%%%%%%%%%%%%%%%%%%%%%%%%%%%%%%%%%%%%%%%%%%%%%%%%%%%%%%
%%%%%%%%%%%%%%%%%%%%%%%%%%%%

High dose rate brachytherapy (HDR-BT) is an effective treatment for skin cancer, especially when surgery alone is not sufficient to remove all cancer cells. 
To treat superficial tumors, an afterloader is used to push a radioactive source through several needles placed in close proximity to the cancer cells on the skin's surface~\citep{holm2016heuristics}, stopping at selected dwell positions for the duration specified in the treatment plan. 
Compared to external beam radiation treatments, HDR-BT minimizes damage to nearby healthy tissues by delivering a comparatively smaller radiation dose to non-tumor regions~\citep{skowronek2015brachytherapy}.
Thus, the risk of side effects is reduced with HDR-BT.

A challenging issue in using HDR-BT for superficial cases is to ensure that needles stay correctly positioned throughout the treatment. 
Tape is sometimes used to secure needles to the skin.
Current clinical practice also employs Freiburg flaps, which are flexible two-dimensional mesh-like applicators composed of a grid of small silicon spheres through which needles can be inserted. 
Such applicators can conform fairly well to complicated patient geometries. 
However, they have two important limitations. 
First, they restrict the needles to follow two predefined directions in the mesh, restricting possible treatment plans. 
Recent research indicates that nonparallel needle configurations can help in better achieving dose objectives and in avoiding critical structures near tumor tissues~\citep{garg2012initial}. 
Second, air gaps can form between the skin and these applicators, which may compromise the accuracy of dose computations and impact the quality of plans.

Recently, 3D printing has been proposed as a solution for addressing these issues. 
In this approach, rapid prototyping technology is used for the physical production of applicators with the use of 3D computer aided design~\citep{chua2020introduction} to create 
customized masks tailored to individual patients along with internal channels for needles.  
Studies have shown that 3D-printed technologies facilitate the easier positioning of needles in HDR-BT~\citep{marar2022applying}.
Further, they allow a tight fit of the applicator to the skin and enable the use of complex needle configurations.
Currently, the placement of these needles must be determined manually.
This can pose a significant challenge in clinical settings due to the intricate geometry of patients, especially in areas like the face with few flat surfaces. 
In this paper, we propose optimization approaches to determining needle placements that result in high-quality plans while avoiding intersections with given structures. 
Although several studies have examined needle placement in interstitial HDR-BT, we are not aware of work that has addressed this problem for 3D-printed masks for skin cancer.

%%%%%%%%%%%%%%%%%%%%%%%%%%%%
%%%%%%%%%%%%%%%%%%%%%%%%%%%%%%%%%%%%%%%%%%%%%%%%%%%%%%%%
\subsection{Literature review}
\label{section:intro:litreview}
%%%%%%%%%%%%%%%%%%%%%%%%%%%%%%%%%%%%%%%%%%%%%%%%%%%%%%%%
%%%%%%%%%%%%%%%%%%%%%%%%%%%%

The problem we consider in this paper reduces to that of deciding the positions of channels since the actual geometry of the mask can be created around the selected channels.
Since needles will be inserted in the channels, we use these terms interchangeably. 
Traditionally, needle positions 
have often been determined by adhering to high-level recommendations and without the aid of optimization; see~\cite{kovacs2005gec} for prostate cancer. 

The use of optimization for needle placement has primarily been investigated for interstitial cancer cases. 
Some studies have proposed two-phase approaches in which the location of needles is optimized in the first phase with respect to target point coverage, and  dwell times for selected needles are optimized in the second phase; see~\cite{siauw2011ipip} for the case of prostate cancer.   
These approaches produce plans quickly, although they are limited by the fact that dose criteria are not considered in the needle selection phase.
Similarly, \cite{poulin2013adaptation} use a Centroidal Voronoi Tesselations (CVT) algorithm to position parallel needles in prostate cancer cases. 

Some studies have sought to optimize needle locations and dwell times simultaneously;
see~\cite{holm2016heuristics} and \cite{wang2021simultaneous} for the case of prostate cancer. 
The model proposed in the former article is a difficult mixed integer program (MIP) that is solved with heuristics to obtain good-quality solutions in the allocated time. 
This model does not contain constraints to prevent needles from being placed in the urethra.
The model of the latter article
has an objective function that includes a penalty term for the total number of needles used in addition to dosimetric objectives for the tumor and healthy organs.  
It only considers parallel needle configurations.
Further, the choice of penalty parameters is difficult and significantly impacts the quality of plans.

The approaches described above have the following limitations. 
First, most of them utilize parallel needles, which limit possible treatment plans. 
Second, they often require long computational times to produce sub-optimal dose results. 
Third, they allow needles to intersect with tumor tissues, which is undesirable in superficial skin cancer treatments.

To overcome these limitations, we introduce several optimization approaches 
that can be implemented using commercial optimization software and  
produce high-quality needle configurations.
These approaches compare favorably in terms of dosimetric indices with the clinical plan used in a case of skin cancer treated using a Freiburg flap.  
Further, they significantly reduce the time needed to devise plans and improve their quality, two challenging practical aspects of HDR-BT planning.

The rest of the paper organized as follows.
In Section~\ref{section:modelfeatures}, we describe how we model salient problem features. 
In Section~\ref{section:modelsmethods}, we introduce models and solution methods for optimizing needle positions. 
In Section~\ref{section:procedures}, we present our experimental setup. 
In Section~\ref{section:results}, we describe numerical results for a practical case. 
We conclude in Section~\ref{section:conclusion} with directions of future research.

%%%%%%%%%%%%%%
%%%%%%%%%%%%%%%%%%%%%%%%%%%%
%%%%%%%%%%%%%%%%%%%%%%%%%%%%%%%%%%%%%%%%%%%%%%%%%%%%%%%%
\section{Problem features and modeling constructs}
\label{section:modelfeatures}
%%%%%%%%%%%%%%%%%%%%%%%%%%%%%%%%%%%%%%%%%%%%%%%%%%%%%%%%
%%%%%%%%%%%%%%%%%%%%%%%%%%%%
%%%%%%%%%%%%%%

In this paper, we only consider creating straight channels in masks. 
While it is technologically possible to use curved channels, sharp curvatures may lead to increased risk of source retainment during delivery and thus would be
unsafe to deliver clinically.
%\newline\bynasim{remove:large curvatures may lead afterloaders to block and yield treatment plans that are difficult or unsafe to deliver in practice}. 
Straight needles also have simpler geometric descriptions and allow for simpler dose calculations.

We model needles as  
the union of all balls of radius $\radius$ centered at each point of a line segment, which we refer to as the needle \textit{core}.
Thus, every needle $\line$ can be specified through five parameters: 
($i$) a reference point $\start \in \Re^3$ on its core, 
($ii$) its direction $\dir \in \Re^3$, 
($iii$) the step size $\tmin$ needed to reach one end of the core from $\start$, 
($iv$) the step size $\tmax$ needed to reach the other end of the core from $\start$, 
and ($v$) the radius $\radius$ of the needle.  
Given these parameters, the set of points $\vachi$ in $\line$ is
%can be represented as
\begin{align}
\theneedle{\normp}{\radius}(\start,\dir,\tmin,\tmax) = \proj_{\vachi} \left \{ 
\begin{array}{c|c}
(\vachi,\vatau,\varho) \in \Re^3 \times \Re \times \Re^3 \,\,&\,\, 
\begin{array}{l}
\vachi= \start +\vatau \dir + \varho \\
\tmin \le \vatau  \le \tmax  \\
||\varho||_{\normp} \le \radius 
\end{array}
\end{array}
\right\},\label{def:cylindricalneedle}
\end{align}
where $\normp \ge 1$ is a norm-order.
When $\radius=0$, the needle reduces to its core. 
When $\tmin=-\infty$ and $\tmax=\infty$, the core is a line rather than a line segment. 
We say that the needle is \textit{infinite} in this case. 
We refer to needles as being \textit{finite} when $\tmin>-\infty$ and $\tmax<\infty$.
In \eqref{def:cylindricalneedle}, choosing $\normp=2$ when $\line$ is finite yields a hemispherical cylinder. 
Choosing $\normp$ to be $1$ or $\infty$ yields polyhedral objects. 
As, for all $v \in \Re^3$, $||v||_1 \le \sqrt{3} ||v||_{2}$ and $||v||_\infty \le ||v||_2$, it also holds that 
$\theneedle{2}{\radius}(\start,\dir,\tmin,\tmax) \subseteq \theneedle{1}{\sqrt{3}\radius}(\start,\dir,\tmin,\tmax)$ and
$\theneedle{2}{\radius} (\start,\dir,\tmin,\tmax) \subseteq \theneedle{\infty}{\radius}(\start,\dir,\tmin,\tmax)$, yielding polyhedral approximations of typical needle geometries.

%\begin{figure}[!htb]
%\centering
%\includegraphics[scale=0.4]{ cylinder2.pdf}
%\captionsetup{justification=centering}
%\caption{3D rendering of a hemispherical cylinder }
%\label{cylinder}
%\end{figure}

Next, we describe three salient problem features that models aiming to determine channel positions in 3D-printed masks must consider if their solutions are to be implemented in practice. 
We also present mathematical models for these features that will be used  in later sections.

% \begin{figure}[!htb]
% \centering
% \includegraphics[scale=0.4]{ sample.pdf}
% \captionsetup{justification=centering}
% \caption{Brachytherapy machine setting picture from~\citep{Prada13} with labels added as follows: ($A$) afterloader, ($B$) connecting tubes, ($C$) exiting plane where needles are connected to the tubes, and ($D$) needles}
% \label{planeinprac}
% \end{figure}

%%%%%%%%%%%%%%%%%%%%%%%%%%%%
%%%%%%%%%%%%%%%%%%%%%%%%%%%%%%%%%%%%%%%%%%%%%%%%%%%%%%%%
\subsection{Exiting planes}
\label{section:modelfeatures:exitingplanes}
%%%%%%%%%%%%%%%%%%%%%%%%%%%%%%%%%%%%%%%%%%%%%%%%%%%%%%%%
%%%%%%%%%%%%%%%%%%%%%%%%%%%%

To deliver a plan, each needle must be connected to the afterloader through its own transfer tube.  
As plans utilize multiple needles,  it is crucial to position them so that needle connections to the transfer tubes are possible. 
To ensure this, we introduce the concept of an \textit{exiting plane}. 
Exiting planes are flat surfaces around the patient's body where, at least conceptually, transfer tubes can be connected to the afterloader. 
We require each needle to intersect with at least one exiting plane and consider the intersection to be the needle starting point. 
Figure~\ref{figure:face} depicts the contours of the head of a patient with skin cancer on their nose.
For this case, we define the five exiting planes 
%(away from the face) 
of Figure~\ref{figure:exiting}.

\begin{figure}[!htb]
\centering
\begin{subfigure}{0.45\textwidth}
\centering
\includegraphics[scale=0.45]{ 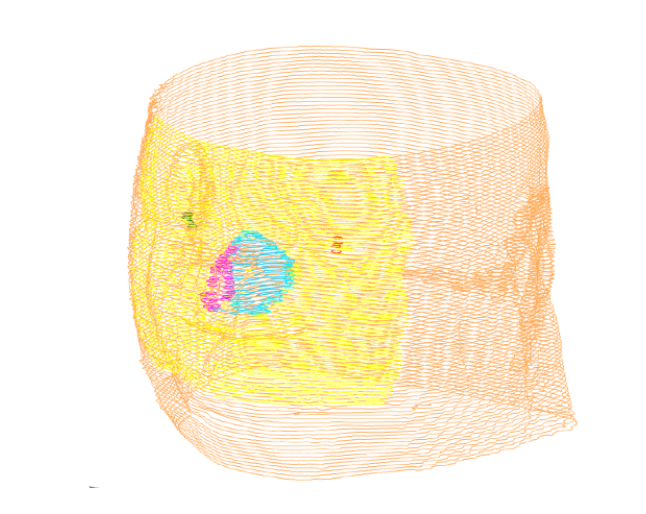}
\captionsetup{justification=centering,font=scriptsize}
\caption{3D representation of the patient's head}
\label{figure:face}
\end{subfigure}
\hfill
\begin{subfigure}{0.4\textwidth}
\centering
\includegraphics[scale=0.4]{ 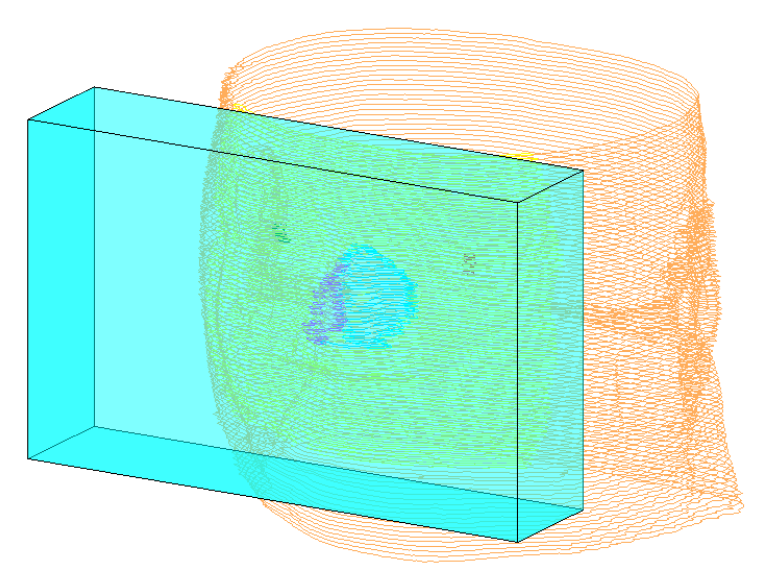}
\captionsetup{justification=centering,font=scriptsize}
\caption{3D representation of exiting planes}
\label{figure:exiting}
\end{subfigure}
\label{exitingm}
\caption{Case study of a patient with skin cancer on the nose}
\end{figure}

%%%%%%%%%%%%%%%%%%%%%%%%%%%%
%%%%%%%%%%%%%%%%%%%%%%%%%%%%%%%%%%%%%%%%%%%%%%%%%%%%%%%%
\subsection{Intersection between two needles}
\label{section:modelfeatures:needle-needle}
%%%%%%%%%%%%%%%%%%%%%%%%%%%%%%%%%%%%%%%%%%%%%%%%%%%%%%%%
%%%%%%%%%%%%%%%%%%%%%%%%%%%%

Within a 3D-printed mask, it is crucial that the path of the radioactive source within each channel is smooth and free of obstruction, as any obstruction could cause the radioactive source to become lodged inside of the channel.
It is thus preferred for channels and the needles within to not intersect to prevent treatment misadministrations.

Proposition~\ref{proposition:twofiniteneedles} presents a modeling construct that allows to determine or to impose that two needles, defined as in \eqref{def:cylindricalneedle},  do not intersect. 
This construction requires the introduction of $\nv{\leftrightarrow}:=12$ variables that we denote by $\dualvar$. 
It is obtained by requiring the set of intersection points between two needles to be empty and then by formulating its conic alternative system (see Proposition~2.4.2 in \cite{ben2001lectures}).  
The proof is given in \onlsup~\ref{section:os:needle-needle}.   
For $\normp \ge 1$, we define $C_{\normp}=\{(r,s) \in \Re^3 \times \Re \,|\, ||r||_{\normp} \le s \}$ to be the $\normp$-order cone of dimension $4$.
The dual cone of $C_{\normp}$, which we denote by $C_{\normp}^*$, is the $\normq$-order cone of dimension $4$ where $\normq$ is such that $\frac{1}{\normp} + \frac{1}{\normq}=1$.

\begin{restatable}{proposition}{Lemmatwofiniteneedles}\label{proposition:twofiniteneedles}
Let $(\start^i,\dir^i,\tmin^i,\tmax^i)$ and $(\start^j,\dir^j,\tmin^j,\tmax^j)$ be the parameters of two finite needles $\line^i$ and $\line^j$. 
Also let $\normp \ge 1$ and $\radius>0$.
The system
\begin{align}
\Systwoneedles^{\leftrightarrow}_{\normp,\radius} \left[\begin{array}{c}\start^i,\dir^i,\tmin^i,\tmax^i\\ \start^j,\dir^j,\tmin^j,\tmax^j \\ \dualvar \end{array}\right]:  
\left\{ 
\begin{array}{ll}
\multicolumn{2}{l}{
(\start^i-\start^j)\tr \dualvarzer   - \tmin^i \dualvarone  + \tmax^i \dualvartwo} \\
\multicolumn{2}{l}{ \qquad\qquad - \tmin^j \dualvarfou + \tmax^j \dualvarfiv  + \radius \dualvarthr  + \radius \dualvarsix \le -\epsilon,} \\
\multicolumn{2}{l}{
2\,\one\tr\dualvarzerp+\dualvarone+\dualvartwo+\dualvarthr+\dualvarfou+\dualvarfiv+\dualvarsix\leq 1,} \\
-\dualvarzer \le \dualvarzerp &\dualvarzer \le \dualvarzerp ,\\
 -(\dir^i)\tr \dualvarzer - \dualvarone + \dualvartwo = 0, \qquad & 
 (\dir^j)\tr \dualvarzer - \dualvarfou + \dualvarfiv = 0 ,\\
 (-\dualvarzer,\dualvarthr) \in C_{\normp}^*, &
 (\dualvarzer,\dualvarsix) \in C_{\normp}^*, \\
 \dualvarone \ge 0, \,\, \dualvartwo \ge 0, \,\, \dualvarthr \ge 0, & \dualvarfou \ge 0, \,\, \dualvarfiv \ge 0,  \,\, \dualvarsix \ge 0,
 \end{array}
 \right.
 \label{sys:twofiniteneedles}
\end{align}
in variables $\dualvar=(\dualvarzer;\dualvarzerp;\dualvarone,\dualvartwo,\dualvarthr,\dualvarfou,\dualvarfiv,\dualvarsix) \in \Re^3 \times \Re^3 \times \Re^6$, ensures that needles $\line^i$ and $\line^j$ do not intersect when $\epsilon$ is chosen positive. 
\end{restatable}

In Proposition~\ref{proposition:twofiniteneedles}, 
alternative theorems would naturally result in a right-hand side that is strictly less than $0$ for the first inequality, rather than being less than or equal to $-\epsilon$.
We opted for using a nonstrict inequality 
as strict inequalities cannot be implemented in solvers.
It can be established that, when $\epsilon$ is sufficiently small, System~\eqref{sys:twofiniteneedles} not only guarantees that needles do not intersect but also provides a fairly tight representation of this requirement. 
This is because when the alternative system that yields \eqref{sys:twofiniteneedles} is infeasible, the original system is ``almost feasible'';
see Proposition~2.4.2 in \cite{ben2001lectures}.
We acknowledge that choosing a suitable value of $\epsilon$ can be difficult in practice, especially for situations where the parameters of $\line^i$ and $\line^j$ are design variables.

When needles are infinite in length, System \eqref{sys:twofiniteneedles} becomes simpler as there is no need to introduce the dual variables $\dualvarone$, $\dualvartwo$, $\dualvarfou$, and $\dualvarfiv$.  
The resulting system, which we denote by $\Systwoneedles^{\infty}_{\normp,\radius}\left[\cdot\right]$ and describe explicitly in Corollary~\ref{lemma:twoinfiniteneedles} of \onlsup~\ref{section:os:needle-needle},
only requires the addition of $\nv{\infty}:=8$ variables $\dualvar$.
Considering needles to be infinite is, however, more restrictive 
as it prevents intersection of needles even outside of the restricted volume where the procedure takes place.  
%We obtain

%\begin{lemma}\label{lemma:twoinfiniteneedles}
%Let $(\start^i,\dir^i)$ and $(\start^j,\dir^j)$ be the parameters of two infinite needles $\line^i$ and $\line^j$. 
%Also let $\normp \ge 1$ and $\radius>0$.
%System \eqref{sys:twofiniteneedles} 
%in variables $\dualvar=(\dualvarzer;\dualvarzerp;\dualvarthr,\dualvarsix) \in \Re^3 \times \Re^3 \times \Re^2$ (where $\dualvarone$, $\dualvartwo$, $\dualvarfou$, and $\dualvarfiv$ are removed), which we denote by 
%$FR^{\infty}_{\normp,\radius}\left[\cdot\right]$,  ensures that needles $\line^i$ and $\line^j$ do not intersect when $\epsilon$ is chosen positive. 
%\end{lemma}

When the needles parameters $(\start^i,\dir^i,\tmin^i,\tmax^i)$ and $(\start^j,\dir^j,\tmin^j,\tmax^j)$ are fixed, System~\eqref{sys:twofiniteneedles} has linear constraints and some of its variables are required to belong to a $\normq$-order cone.
Models with these types of constraints can be solved using commercial software such as \gurobi, for the cases where $\normp \in \{1,2,\infty\}$. 
When the needle parameters are not fixed and are considered to be design variables, then \eqref{sys:twofiniteneedles} has bilinear and nonconvex constraints that multiply variables $(\start,\dir,\tmin,\tmax)$ with dual variables $\dualvar$. 
Such constraints can also be handled by \gurobi\xspace as they are special types of nonconvex quadratic constraints. 
The systems $\Systwoneedles^{\leftrightarrow}_{\normp,\radius}\left[\cdot\right]$ and $\Systwoneedles^{\infty}_{\normp,\radius}\left[\cdot\right]$ differ in the number of bilinear terms they have. 
System $\Systwoneedles^{\infty}_{\normp,\radius}\left[\cdot\right]$, which has fewer bilinear terms, tend to yield simpler models.

An alternative approach to ensure that infinite-length needles are spaced sufficiently apart to prevent intersections is to impose a minimum distance of $2\radius$ between their cores. 
 This approach is similar to the one described earlier, but it offers a simplified formula that can be used to create different formulations. 
We obtain the following result, whose proof can be found in \onlsup~\ref{section:os:needle-needle}.

\begin{restatable}{proposition}{Lemmaminimumdistanceofcores}\label{proposition:twoskewneedles}
Let $(\start^i,\dir^i)$ and $(\start^j,\dir^j)$ be the parameters of two skew needles $\line^i$ and $\line^j$.
The euclidean distance between the cores of needles $\line^i$ and $\line^j$ is at least $2\radius$ if their parameters satisfy the following constraint:
\begin{align}
\Systwoneedles^{\times}_{2,\radius} \left[\begin{array}{c}\start^i,\dir^i,\cdot,\cdot\\ \start^j,\dir^j,\cdot,\cdot \\ \cdot \end{array}\right]:    
\left\{ 
\begin{array}{ll}
\det\left(M[\start^i,\dir^i,\start^j,\dir^j]\right) \ge 2\radius \sqrt{
\begin{array}{l}
(\dir^i_2 \dir^j_3-\dir^i_3 \dir^j_2)^2 \\
\quad + (\dir^i_3 \dir^j_1-\dir^i_1 \dir^j_3)^2 \\
\quad \quad + (\dir^i_1 \dir^j_2-\dir^i_2 \dir^j_1)^2 
\end{array},}
\label{distanceskew}
\end{array}
\right.
\end{align}
where
$$M[\start^i,\dir^i,\start^j,\dir^j] = \left[ 
\begin{array}{ccc} 
\start_1^i - \start_1^j & \start_2^i - \start_2^j & \start_3^i - \start_3^j \\
\dir_1^i & \dir_2^i & \dir_3^i \\
\dir_1^j & \dir_2^j & \dir_3^j \\
\end{array}
\right].$$
\end{restatable}

We observe that $\Systwoneedles^{\times}_{2,\radius} \left[\cdot\right]$ does not require the introduction of variables $\dualvar$, so that $\nv{\times}:=0$.
For a pair of needles $\line^i$ and $\line^j$ with reference points $\start^i$ and $\start^j$ and directions $\dir^i$ and $\dir^j$, \eqref{distanceskew} is a second-order cone constraint in $\dir^i$ (\textit{resp.},  $\dir^j$) if all other needle parameters are fixed. 
Similarly, \eqref{distanceskew} is a linear constraint in $\start^i$ and $\start^j$ if all other parameters are fixed. 
Constraint~\eqref{distanceskew} requires the needles to be skew as otherwise both its left-hand side and right-hand sides are equal to zero. 
When needles are parallel, a simpler formula for the distance between needles can be computed, 
which can be found in Lemma~\ref{parneedles} of \onlsup~\ref{section:os:needle-needle}. 
This derivation yields a simplified version of system
$\Systwoneedles^{\times}_{2,\radius} \left[\cdot\right]$, which we call
$\Systwoneedles^{=}_{2,\radius} \left[\cdot\right]$, 
for which no new variables $\dualvar$ are needed, \ie $\nv{=}:=0$.

%\blue{I CHANGED M TO M[] ... DO SAME IN APPENDIX}

%%%%%%%%%%%%%%%%%%%%%%%%%%%%
%%%%%%%%%%%%%%%%%%%%%%%%%%%%%%%%%%%%%%%%%%%%%%%%%%%%%%%%
\subsection{Intersection between a needles and a body structure}
\label{section:modelfeatures:needle-polytope}
%\label{Intersection of needles with body structures}
%%%%%%%%%%%%%%%%%%%%%%%%%%%%%%%%%%%%%%%%%%%%%%%%%%%%%%%%
%%%%%%%%%%%%%%%%%%%%%%%%%%%%

The use of HDR-BT for superficial radiation treatment offers advantages when the needles are in close proximity to the treatment area, especially the tumor. 
However, in areas of the body that have complicated geometries, being close to the body might be difficult to achieve without crossing sensitive body structures. 
For instance, when treating skin cancer on the face,  
needle positioning is challenging due to the nose protruding from the rest of the face and proximity of the eyes.

We propose to represent protruding structures of interest as polytopes, which are sets of the form $P_{\polylhs,\polyrhs}:=\left\{ p \in \Re^3 \,|\, \polylhs p \leq \polyrhs \right\}$, where $\polylhs \in \Re^{n_{\polylhs,\polyrhs}\times 3}$ is a matrix and $\polyrhs \in \Re^{n_{\polylhs,\polyrhs}}$ is a vector.
In this description, $n_{\polylhs,\polyrhs}$ denotes the number of inequalities in $P_{\polylhs,\polyrhs}$. 
One approach to obtaining these polytopes is to select points that delineate the boundaries of each protruding structure and then to compute the convex hull of these points using an algorithm such as quickhull \citep{barber1996quickhull} or the double description method \citep{motzkin1953double}.  
%In ensuing sections, we often represent structures by few planes, as increasing the number of inequalities negatively affects computation.

Proposition~\ref{proposition:infiniteneedlepolytope}, whose proof is presented in \onlsup~\ref{section:os:needle-polytope}, describes 
a modeling construct that allows to determine or to impose that
a given needle $\line$ and a structure of interest described by a polytope $P_{\polylhs,\polyrhs}$ do not intersect. 
%and presents constraints that prevent it from intersecting with $P_{\polylhs,\polyrhs}$.
Similar to Proposition~\ref{proposition:twofiniteneedles}, this system is obtained using conic duality starting from a system that models the set of points that belong to both $\line$ and $P_{\polylhs,\polyrhs}$.  
It requires the introduction of $\nw{\leftrightarrow}:=n_{\polylhs,\polyrhs}+9$ variables. 

\begin{restatable}{proposition}{Lemmainfiniteneedlepolytope} \label{proposition:infiniteneedlepolytope}
Let $P_{\polylhs,\polyrhs}$ 
be a polyope and let $(\start,\dir,\tmin,\tmax)$ be the parameters of a finite needle $\line$. 
Also let $\normp \ge 1$ and let $\radius > 0$. 
The system
\begin{align}
\SysneedleP^{\leftrightarrow}_{\normp,\radius} \left[ \begin{array}{c} \start,\dir,\tmin,\tmax \\ \polylhs, \polyrhs \\ \dualvarw \end{array}\right] :  
\left\{ 
\begin{array}{l}
\polyrhs\tr \dualvarwone + \start\tr \dualvarwtwo - \tmin \dualvarwthr + \tmax \dualvarwfou + \radius \dualvarwfiv \le -\epsilon,\\
\polylhs\tr \dualvarwone + \dualvarwtwo = 0, \\
-\dir\tr \dualvarwtwo - \dualvarwthr + \dualvarwfou = 0, \\
\one\tr\dualvarwone + \one\tr \dualvarwtwop + \dualvarwthr+\dualvarwfou+\dualvarwfiv \le 1, \\
\,\,\,\, \dualvarwtwo \le \dualvarwtwop, \\
-\dualvarwtwo \le -\dualvarwtwop,\\
(-\dualvarwtwo,\dualvarwfiv) \in C^*_{\normp},\\
\dualvarwone \in \Re^n_+, \dualvarwtwo \in \Re^3, \dualvarwthr,  \dualvarwfou, \dualvarwfiv \in \Re_+,  \\
\end{array}
\right.
\label{sys:needlepolytope}
\end{align}
in variables $\dualvarw=(\dualvarwone;\dualvarwtwo;\dualvarwtwop;\dualvarwthr,\dualvarwfou,\dualvarwfiv) \in \Re^n \times \Re^3 \times \Re^3 \times \Re^3$ 
ensures that needle $\line$ and polytope $P_{\polylhs,\polyrhs}$ do not intersect when $\epsilon$ is chosen positive. 
\end{restatable}

Similar to Proposition~\ref{proposition:twofiniteneedles}, it can be established that, when $\epsilon$ is chosen to be small, System~\eqref{sys:needlepolytope} not only guarantees that there is no intersection between $\line$ and $P_{\polylhs,\polyrhs}$ but also provides a fairly tight representation of this requirement. 

When needles are infinite in length, System~\eqref{sys:needlepolytope} becomes simpler as there is no need to introduce the dual variables 
$\dualvarwthr$ and $\dualvarwfou$.
The resulting system, which we denote by 
$\SysneedleP^{\infty}_{\normp,\radius} \left[ \begin{array}{c} \cdot \end{array}\right]$ and describe explicitly in Corollary~\ref{corollary6.2.1} of \onlsup~\ref{section:os:needle-polytope}, only requires the addition of 
$\nw{\infty}=n_{\polylhs,\polyrhs}+7$ variables $\dualvarw$.
Considering the needle to be infinite is, however, more restrictive as it does not allow the continuation of a short needle to intersect the polyhedral structures to avoid. 

%\begin{lemma}\label{lemma2.5}
%Let $P_{\polylhs,\polyrhs}$ 
%be a polytope and let $\line$ be an infinite needle with parameters $(\start,\dir)$. Also let $\normp \ge 1$ and let $\radius > 0$. 
%System \eqref{sys:needlepolytope} in variables $\dualvarw=(\dualvarwone,\dualvarwtwo,\dualvarwtwop,\dualvarwfiv) \in \Re^{n_{\polylhs,\polyrhs}+7}$
%(where variables $\dualvarwthr$ and $\dualvarwfou$ are removed), which we denote by $\SysneedleP^{\infty}_{\normp,\radius} \left[ \begin{array}{c} \cdot \end{array}\right]$,  ensures needle $\line$ and polytope $P_{\polylhs,\polyrhs}$ do not intersect, when $\epsilon$ is chosen to be positive. 
%\end{lemma}

%A complete description of $\SysneedleP^{\infty}_{\normp,\radius} \left[ \begin{array}{c} \cdot \end{array}\right]$ can be found in \onlsup~\ref{col1}.  
When implementing the above sets of constraints inside of solution algorithms, we found that the simpler set of constraints requiring $P_{\polylhs,\polyrhs}$ and the core of needle $\line$ to not intersect can be successfully utilized.
Mathematically, this can be achieved by removing dual variables $\dualvarwthr$, $\dualvarwfou$, and $\dualvarwfiv$ from $\SysneedleP^{\leftrightarrow}_{\normp,\radius} \left[ \begin{array}{c} \cdot \end{array}\right]$ to create a system we refer to as 
$\SysneedleP^{\times}_{\normp,\radius} \left[ \begin{array}{c} \cdot \end{array}\right]$.
This system requires the addition of only $\nw{\times}=n_{\polylhs,\polyrhs}+6$ variables $\dualvarw$
and is described explicitly in Corollary~\ref{cor:needlecorepoly} of
%A description of $\SysneedleP^{\times}_{\normp,\radius} \left[ \begin{array}{c} \cdot \end{array}\right]$ can be found in 
\onlsup~\ref{section:os:needle-polytope}.
If needle $\line$ has positive radius, this system does not guarantee that $P_{\polylhs,\polyrhs}$ and $\line$ do not intersect. 
We found, however, that setting $\epsilon=\radius$ achieves this goal for the cases in our computational study.

%\begin{lemma}\label{lemma2.6}
%Let $P_{\polylhs,\polyrhs}$ 
%be a polyope and let $\line$ be an infinite needle with parameters $(\start,\dir)$. 
%Also let $\normp \ge 1$. 
%System \eqref{sys:needlepolytope} in variables $\dualvarw=(\dualvarwone,\dualvarwtwo,\dualvarwtwop) \in \Re^{n_{\polylhs,\polyrhs}+6}$
%(where variables $\dualvarwthr$, $\dualvarwfou$, and $\dualvarwfiv$ are removed), which we denote by $\SysneedleP^{\times}_{\normp,\radius} \left[ \begin{array}{c} \cdot \end{array}\right]$,  ensures that the core of needle $\line$ and polytope $P_{\polylhs,\polyrhs}$ do not intersect, when $\epsilon$ is chosen to be positive.
%\end{lemma}

When the needle parameters $\start$, $\dir$, $\tmin$, and $\tmax$ are fixed, the systems 
$\SysneedleP^{\mdlb}_{\normp,\radius} \left[ \begin{array}{c} \cdot \end{array}\right]$ where $\mdlb\in \{\leftrightarrow,$ $\infty, \times\}$
are linear for $\normp=1$ and $\normp=\infty$, and have second-order cone constraints for $\normp=2$.
When the needle parameters are not fixed and are considered to be design variables, then they contain bilinear and nonconvex constraints that multiply needle parameters variables $(\start,\dir,\tmin,\tmax)$ with dual variables $\dualvarw$. 
Such constraints can be handled by \gurobi\xspace as they are special types of nonconvex quadratic constraints. 
The systems 
$\SysneedleP^{\mdlb}_{\normp,\radius} \left[ \begin{array}{c} \cdot \end{array}\right]$
differ in the number of bilinear terms they have. 
Systems with fewer bilinear terms tend to yield simpler models.

%%%%%%%%%%%%%%
%%%%%%%%%%%%%%%%%%%%%%%%%%%%
%%%%%%%%%%%%%%%%%%%%%%%%%%%%%%%%%%%%%%%%%%%%%%%%%%%%%%%%
\section{Models and methods}
\label{section:modelsmethods}
%%%%%%%%%%%%%%%%%%%%%%%%%%%%%%%%%%%%%%%%%%%%%%%%%%%%%%%%
%%%%%%%%%%%%%%%%%%%%%%%%%%%%
%%%%%%%%%%%%%%
 
When designing models for the placement of channels, there are at least two modeling approaches for the representation of the needles to use inside of a 3D-printed mask. 

In the first modeling approach, we pre-generate a collection of candidate needles, much larger than what will be used. 
We can ensure that each candidate needle avoids specified regions of interest and intersects exiting planes by eliminating or shortening the needle.
A model can then be created that selects a small subset of these candidate needles that do not intersect with each other and yield a high-quality plan.   
We refer to such models as \textit{fixed needles models}. 

In the second modeling approach, the number of needles to use is fixed but their characteristics (such as core and direction) are decision variables. 
In such models, constraints must be formulated to ensure that needles intersect with exiting planes and to prevent them from intersecting with each other and with identified body structures.
We refer to such models as \textit{free needles models}

An advantage of fixed needles models is that dose computation is straightforward, which is not the case for free needles models. 
However, free needles models offer greater flexibility in needle positioning, promising the potential for better treatment plans.
In Section~\ref{section:modelsmethods:fixed}, we describe a fixed needles model.  
In Section~\ref{section:modelsmethods:free}, we follow a free needles approach and present a three-phase approach where the first identifies promising dwell positions, the second optimizes needle placements and directions to best cover these positions with few needles, and the third computes a treatment plan. 
For the second phase, we propose two different models.

%\byjp{THE ONLINE SUPPLEMENTS REFERENCES DO NOT SEEM TO FOLLOW LOGICALLY, A, B, C, ...}

%%%%%%%%%%%%%%%%%%%%%%%%%%%%
%%%%%%%%%%%%%%%%%%%%%%%%%%%%%%%%%%%%%%%%%%%%%%%%%%%%%%%%
\subsection{Fixed needles model} 
\label{section:modelsmethods:fixed}
%%%%%%%%%%%%%%%%%%%%%%%%%%%%%%%%%%%%%%%%%%%%%%%%%%%%%%%%
%%%%%%%%%%%%%%%%%%%%%%%%%%%%

In this section, we present a fixed needles model where a collection of candidate needles are pre-generated. 
We assume that tumor and organs-at-risk (OARs) have been voxelized (with voxel sets $\vox^\tumor$ and $\vox^\oars$, respectively) and that a prescription dose level $\prescr_i$ has been assigned to each voxel $i \in \vox := \vox^\tumor \cup \vox^\oars$. 
This dose level depends on whether the voxel is within the tumor or within an OAR. 
To measure the quality of a plan, we use continuous piecewise-linear convex functions $f_i^\tumor(\cdot)$ (\textit{resp.} $f_i^\oars(\cdot)$) that, for each for tumor voxel $i$ (\textit{resp.} for each OAR voxel $i$), penalize the discrepancy between the dose delivered at the voxel and the prescription.
Such penalties are commonly used in inverse planning; see \citep{holm2012impact,alterovitz2006optimization,lessard2001inverse}.

%\byjp{DO WE GIVE THE ACTUAL PARAMETERS ALPHA AND BETA ANYWHERE?}
%\byjp{I would suggest removing the plots and add the functions instead}

%{\color{red}
%\byjp{Move this paragraph to a following section?}
To generate the candidate needle positions the model requires, we use the following guiding principles.  
First, the needles must be in close proximity to the target area. 
If the dwell locations are too far away, long dwell times will be required to deliver an adequate dose to the tumor, resulting in a significant dose being deposited in other parts of the patient's body. 
Second, the candidate needles must not intersect with the patient's body. 
For example, a needle that crosses the bridge of the patient's nose should not be used.
Third, the candidate needles must intersect exiting planes.
In Section~\ref{section:results:fixed}, we present a specific procedure for generating these candidate needles that satisfy the above principles and gives rise to candidate needles that are orthogonal to one of the exiting planes. 
Restricting needle directions in this way, which is not compulsory, has the advantage of limiting the number of candidate needles and therefore the size of the ensuing model.  

%}

%\byjp{NAME OF SYSTEM SEEMS WRONG}

% \begin{figure}[!htb]
% \centering
% % \vskip -3cm
% \includegraphics[scale=0.2]{ nosepoints.pdf}
% %\vskip -2.5cm
% \caption{Candidate needle endpoints generated in axis-parallel method }
% \label{nosesurf2}
% \end{figure}

%\byjp{Need to say that these candidate needles do not intersect regions of interest and start at exiting planes}
%We refer to the needles produced by the above procedure as \textit{candidate} needles and 
We denote the collection of candidate needles by \sn.
For each candidate needle $\line$, we generate a set of dwell locations $\dwell_{\line}$ by considering all points equally spaced by $1$mm from the endpoint of $\line$. 
We then determine the collection $\ints$ of pairs of needles that have an intersection. 
%\byjp{IS SYSTEM NAME RIGHT?}
We introduce parameters $\Tmax$ to represent the maximum dwell time that can be applied to any single dwell position and $\Nmax$ to denote the maximum number of needles that can be used. 
Given any candidate needle $\line$ and any of its dwell location $k$, we use the linear-source formula of the American Association of Physicists in Medicine (AAPM) task group 186 (TG-43 186)~\citep{beaulieu2012report} to determine the dose $\doselinear_{i,k,\line}$ that is delivered per unit of time from $k$ to voxel $i$.

%\byjp{The form $f(y-\delta)$ DOES NOT SEEM TO MATCH THE PICTURE}
In addition to the data above,  
the model uses three sets of variables. 
Binary variable $\pickneedle_{\line}$ indicates if needle $\line$ is used or not. 
Continuous variable $\dwtimes_{k,\line}$ represents the dwell time applied to position $k$ of needle $\line$ and continuous variable $\dosevar_i$ represents the total dose delivered to voxel $i$.  
We write
\begin{subequations}
\label{model:fixedneedle}
\begin{alignat}{3} 
	 \min \qquad &  \sum_{i \in \vox^\tumor} f_i^\tumor(\dosevar_i -\prescr_i) + \sum_{i \in \vox^\oars} f_i^\oars(\dosevar_i-\prescr_i)  \label{model:fixedneedle:objective} \\
	\text{s.t.} \qquad &   \dosevar_{i} = \sum_{l \in \sn} \sum_{k \in \dwell_{\line}} {\doselinear_{i,k,\line} \dwtimes_{k,\line}}, &\quad& \forall i \in \vox, \label{model:fixedneedle:doseconstraint} \\ 
	& \dwtimes_{k,\line}\leq \Tmax \pickneedle_{\line}, && \forall k \in \dwell_{\line},~\forall \ell \in  \sn, \label{model:fixedneedle:timetobinary} \\
	&  \sum_{\line \in \sn} \pickneedle_{\line} \leq  \Nmax,\label{model:fixedneedle:totalbin} \\
	&   \pickneedle_{\line_1}+\pickneedle_{\line_2}\leq 1, && \forall (\line_1,\line_2) \in \ints, \label{model:fixedneedle:intersect} 
\end{alignat}
\end{subequations}
with the requirements that 
($i$) $\dwtimes_{k,\line} \geq 0$, $\forall k \in \dwell_{\line},~\forall \line \in \sn$, 
($ii$) $\dosevar_i \geq 0$, $\forall i \in \vox$, and
($iii$) $\pickneedle_{\line} \in \{0,1\}$, $\forall \line \in \sn$.
Objective \eqref{model:fixedneedle:objective} minimizes the total penalty of deviations between dose delivered and prescription across all voxels. 
Constraint \eqref{model:fixedneedle:doseconstraint} computes the dose delivered to voxel $i$ from all activated dwell locations on all needles assuming dose accumulates linearly. 
Constraint~\eqref{model:fixedneedle:timetobinary} guarantees that only dwell times from selected needles are activated.  
Constraint \eqref{model:fixedneedle:totalbin} limits the number of needles used.
Constraint \eqref{model:fixedneedle:intersect} enforces that, for each pair of intersecting needles, at most one is selected. 

Model~\eqref{model:fixedneedle} is a mixed integer nonlinear program (MINLP) because its objective is nonlinear. 
For piecewise-linear penalty functions $f_i^\tumor(\cdot)$ and $f_i^\oars(\cdot)$, which are common in applications,
traditional linearization techniques can be used to rewrite the model as an MIP; see \citep{lin2013review}. 
Although simple in structure, \eqref{model:fixedneedle} can be time-consuming to solve as we demonstrate in Section~\ref{section:results:comparison}. 
Similar models proposed in~\citep{holm2016heuristics,wang2021simultaneous} suffer from the same issue. 

%%%%%%%%%%%%%%%%%%%%%%%%%%%%
%%%%%%%%%%%%%%%%%%%%%%%%%%%%%%%%%%%%%%%%%%%%%%%%%%%%%%%%
\subsection{Free needles model}
\label{section:modelsmethods:free}
%%%%%%%%%%%%%%%%%%%%%%%%%%%%%%%%%%%%%%%%%%%%%%%%%%%%%%%%
%%%%%%%%%%%%%%%%%%%%%%%%%%%%

In this section, we present an alternative method for selecting needles locations in  which candidate needle locations are not pre-generated.
This method consists of three phases. 
First, we identify a set of candidate dwell positions by solving a linear program (LP). 
Second, we construct a small set of needles that best cover these candidate dwell locations.  
For this phase, we propose two approaches.  
The first generates all needles passing through pairs of candidate dwell positions and solves an MIP using commercial software to identify a small subset that covers the largest amount of dwell times.
The second uses an MINLP to minimize the distance of candidate dwell points to their assigned needles and solves this model with an alternating heuristic.  
%This two-step procedure only generates needle positions and does not generate an associated treatment plan.  
Once needle positions are known, 
optimal dwell times are obtained by solving an LP, which we refer to as the third phase.

%%%%%%%%%%%%%%%%%%%%%%%%%%%%%%%%%%%%%%%%%%%%%%%%%%%%%%%%
\subsubsection{Phase 1: Finding candidate dwell positions.}
\label{section:modelsmethods:free:phase1}
%%%%%%%%%%%%%%%%%%%%%%%%%%%%%%%%%%%%%%%%%%%%%%%%%%%%%%%%

In this phase, we determine a collection of dwell positions for a good treatment plan by generating a vast pool of prospective dwell points and by using an optimization model to select the best ones among them. 
As dwell points near the tumor are preferred, we only generate prospective points close to the tumor regions and ensure that they are outside the body structures that need to be avoided. 
In Section~\ref{section:results:free:phase1}, we present a specific procedure for generating these prospective dwell points that satisfy the above principles. 
%\byjp{standardize possible, prospective and candidate terminology}

For each prospective dwell position $k \in \dwell$, we compute $\doselinear^{\star}_{i,k}$, the dose delivered from $k$ to voxel $i \in \vox$ per unit time, using the point-source formula by AAPM task group 186~\citep{beaulieu2012report}. 
We use the point-source formula instead of the linear-source formula as the latter requires knowledge of the needle directions.
The model is similar to \eqref{model:fixedneedle} and uses two sets of variables.
Continuous variable $\dwtimes_{k}$ represents the dwell time at position $k\in \dwell$ while variable $\dosevar_i$ represents the dose delivered to voxel $i$. 
We write
\begin{subequations}
\label{model:freeneedle:phase1}
\begin{alignat}{3} 
	 \min \qquad & \sum_{i \in \vox^\tumor} f_i^\tumor(\dosevar_i-\prescr_i) + \sum_{i \in \vox^\oars} f_i^\oars(\dosevar_i-\prescr_i)  \label{model:freeneedle:phase1:objective} \\
	\text{s.t.} \qquad &  \dosevar_{i} =\sum_{k \in \dwell} {\doselinear^{\star}_{i,k} \dwtimes_{k}}, &\quad& \forall i \in \vox \label{model:freeneedle:phase1:doseconstraint}, \\ 
	& \dwtimes_{k}\leq \Tmax, && \forall k \in \dwell \label{model:freeneedle:phase1:timetobinary}, 
\end{alignat}
\end{subequations}
with the requirements that 
($i$) $\dwtimes_{k} \geq 0$, $\forall k \in \dwell$, and 
($ii$) $\dosevar_i \geq 0$, $\forall i \in \vox$.

In our computational experience, many prospective dwell positions are activated (\ie $\dwtimes^*_k>0$), often for only brief time periods. 
To address this issue, we select the $\Imax$ positions with the largest dwell times, yielding a smaller collection $\points$ of \textit{candidate} dwell positions. 
Controlling the size of $\points$ is critical to keep the solution of models efficient, but $\Imax$ must be large enough to ensure that high-quality treatment plans can be obtained.
%\byjp{$\Nmax$ is overloaded, it can mean max number of points or needles}

%%%%%%%%%%%%%%%%%%%%%%%%%%%%%%%%%%%%%%%%%%%%%%%%%%%%%%%%
\subsubsection{Phase 2: Finding needles associated with candidate dwell positions.}
\label{section:modelsmethods:free:phase2}
%%%%%%%%%%%%%%%%%%%%%%%%%%%%%%%%%%%%%%%%%%%%%%%%%%%%%%%%

Given a set $\points$ of candidate dwell positions, the second phase aims to determine the best way of ``covering'' them using a small number of needles. 
We investigate two approaches.

\paragraph{Method 1: Maximum coverage model.}

A simple approach is to generate all needles that pass through pairs of candidate dwell points, assuming that these needles do not intersect structures of interest and intersect with an exiting plane. 
We assume that a needle covers a dwell location only when it passes directly through it. 
Then, we select a limited number of these needles that maximizes the sum of dwell times associated with the dwell points they cover.  %collectively best cover the dwell positions, \ie to find a small subset that maximizes the amount of dwell times associated with the dwell positions covered by the generated needles.
The resulting model, which we call \textit{maximum coverage} model, 
%and present in \onlsup~\ref{using maximum coverage}, 
is an MIP that provides a baseline for comparison. 
%This section presents an alternative approach to solve the problem of fitting needles through the set of dwell positions obtained in Section~\ref{firstoftwo}. 
%The proposed model is an MIP with the objective of 
%\byjp{Move the last line to the ending paragraph of this section?}

%\begin{table}[!htbp]
%\centering
%\begin{tabular}{ll}
%\toprule
%Symbols & Description      \\
%\midrule
%Data \\
%\midrule
%$\cand$ & Set of all possible needles \\
%$\points$ & Set of dwell positions to be covered\\ 
%$\oncand_g$ & Set of dwell positions covered by needle $g\in \cand$ \\
%$\Up_g$ & Aggregated dwell time of all positions covered by needle $g\in %\cand$\\
%$\ints$ & Set of pairs of needles that intersect \\
%$N_{\max}$ & Maximum number of needles that can be used\\
%\midrule
%Variables \\
%\midrule
%$\ndlvar_g$   & Binary decision variable that takes value $1$ if needle $g \in \cand$ \\
%&is used and $0$ otherwise\\
%\bottomrule
%\end{tabular}
%\caption{Description of symbols}
%\label{table3: Description of symbols}
%\end{table}

%\byjp{I don't think that $\Up_g$ is defined}
To formulate the model, we consider a set $\cand$ of candidate needles that cover pairs of dwell points $i \in \points$, of which we will select no more than $\Nmax$. 
We denote the sum of dwell times of dwell positions covered by needle $g$ by $\Up_g$. 
We denote by $\oncand_g$ the set of dwell points covered by candidate needle $g$ and let $\ints$ be the set of pairs of needles that intersect. 
In addition to this data, 
%which is summarized in Table~\ref{table3: Description of symbols}, 
we introduce binary variable $\ndlvar_g$ to indicate whether needle $g$ is selected. 
We write
\begin{subequations}
\label{model:freeneedle:phase2:coverage}
%\label{mp}
\begin{alignat}{3}
&& \max \quad & \sum_{g \in \cand} \Up_g\ndlvar_g \label{model:freeneedle:phase2:coverage:objective} \\
(\MC) \qquad && \text{s.t.} \quad &\sum_{g\in \cand \,:\, i\in \oncand_g} \ndlvar_g\leq1, &\qquad& \forall i \in \points, \label{model:freeneedle:phase2:coverage:con1}\\
&&& \ndlvar_g+\ndlvar_{g'}\leq 1, && \forall (g,g') \in \ints, \label{model:freeneedle:phase2:coverage:con2}\\
&&& \sum_{g \in \cand} \ndlvar_g \leq \Nmax ,\label{model:freeneedle:phase2:coverage:con22}
\end{alignat}
\end{subequations}
with the requirement that $\ndlvar_g \in\{0,1\}$, $\forall g \in \cand$.
The objective function~\eqref{model:freeneedle:phase2:coverage:objective} maximizes the total amount of dwell times covered.  
Constraint~\eqref{model:freeneedle:phase2:coverage:con1} ensures that every dwell position is covered by at most one needle as covering a dwell position with two needles would imply that needles cross. 
Constraint~\eqref{model:freeneedle:phase2:coverage:con2} ensures that at most one of any pair of intersecting needles is chosen in the treatment plan.
Constraint~\eqref{model:freeneedle:phase2:coverage:con22} limits the total number of needles selected to be no more than $\Nmax$.

This formulation has of the order of $|\points|^2$ variables since the needles we generate typically contain two dwell locations and since we associate a variable to each candidate needle in the plan. 
While the requirement that needles pass exactly through two dwell positions may seem restrictive, it is not necessarily so when the dwell locations are selected from a ``grid'' of candidate locations, which is the case in our computation.
The model can also be used when needles are not passing exactly through dwell positions but are assigned to cover all sufficiently close dwell locations.  
%were generated that ``cover" three or more dwell positions. 
In this case, the needle best fitting the ``close" dwell positions would be created, potentially not passing through any of them.
Although it is intuitive that the use of such needles could result in better solutions, our pilot computational experiments concluded that they do not appear to translate into better treatment plans. 
We attribute this to the fact that, by trying to cover multiple dwell positions approximately, the model end up picking needles that do not contain any of the preferred dwell positions.

\paragraph{Method 2: Clustering approach.}

Another method, which we refer to as \textit{clustering} approach,  entails the heuristic solution of an MINLP we present next.
Given a collection $\points$ of candidate dwell points $\dwpt^{i}$, this model computes the parameters of a set $\needles$ of non-intersecting needles, each originating from one of the exiting planes $\dn_e$ for $e \in \en$, that do not intersect polytopes $P_{\polylhs^m,\polyrhs^m}$ of interest for $m \in \polyhyd$, with the property that the sum of squared distances from each dwell point in $\points$ to its closest point on a needle is minimized. 
We refer to the subset of finite needles of $\needles$ as $\needles^{\leftrightarrow}$.
%The model also imposes that the starting point of each needle belongs to one of the exiting planes .

The model uses several variables. 
Continuous variables $(\start^j,\dir^j,\tmin^j,\tmax^j)$ represent the parameters of each needle $j \in \needles$.
Continuous variables $\dualvar^{j,j'}$ are used to impose that needles $j$ and $j'$ do not intersect whereas continuous variables $\dualvarw^{j,m}$ are used to impose that needle $j$ does not intersect polyhedral structure of interest $m$. 
Continuous variables $\stept_{i,j}$ describe the step size to apply from $\start^j$ to obtain the projection of dwell point $i$ onto needle $j$.
Finally, binary variable $\asgt_{i,j}$ indicates whether dwell point $i$ is closest to needle $j$.
We use $\mdla$ to indicate which model $\Systwoneedles^{\mdla}_{\normp,\radius}\left[\cdot\right]$ of  Section~\ref{section:modelfeatures} is chosen to impose that needles do not intersect, and use $\mdlb$ to denote which model $\SysneedleP^{\mdlb}_{\normp,\radius} \left[ \cdot \right]$ is chosen to impose that a needle does not intersect with a structure of interest. 
We restrict parameters $(\mdla,\mdlb)$ to $\{(\leftrightarrow,\leftrightarrow), (\infty,\infty), (\times,\infty),  (=,\infty), (\times,\times),  (=,\times)\}$
 and $\normp \in \{1,2,\infty\}$.
Further, we impose that $\needles^{\leftrightarrow}=\needles$ when $\mdla=\leftrightarrow$ or $\mdlb=\leftrightarrow$ and $\needles^{\leftrightarrow}=\emptyset$ otherwise.
We write 
\begin{subequations}
\label{model:freeneedle:phase2:clustering}
\begin{alignat}{5} 
&&\min \quad & \sum_{ i\in \points}\sum_{ j\in \needles}{||\dwpt^{i}-\start^{j}-\stept_{i,j}\dir^{j}||_2^2} \,\asgt_{i,j}\label{model:freeneedle:phase2:clustering:obj}\\ 
&&\text{s.t.} \quad & \sum_{j\in \needles} \asgt_{i,j}=1, &\quad& \forall i\in \points \label{model:freeneedle:phase2:clustering:con1},\\
&&& \sum_{i\in \points} \asgt_{i,j}\geq2, &&  \forall j\in \needles \label{model:freeneedle:phase2:clustering:con11},\\
\Mmdl \qquad &&& {||\dir^j||}^2 = 1, && \forall j \in  \needles\label{model:freeneedle:phase2:clustering:con66}, \\
&&& \Systwoneedles_{\normp,\radius}^{\mdla} \left[\begin{array}{c}\start^{j'},\dir^{j'},\tmin^{j'},\tmax^{j'},\\
 \start^j,\dir^j,\tmin^j,\tmax^j \\ \dualvar^{j,{j'}} \end{array}\right] , && \forall j',j \in \needles:~j'\neq j\label{model:freeneedle:phase2:clustering:con2},\\
&&& \SysneedleP_{\normp,\radius}^{\mdlb} \left[ \begin{array}{c} \start^j,\dir^j,\tmin^j,\tmax^j \\ \polylhs^m, \polyrhs^m \\ \dualvarw^{j,m} \end{array}\right],&& \forall j \in \needles ,~\forall  m \in \polyhyd\label{model:freeneedle:phase2:clustering:con4},\\
&&& \tmin^j\leq \stept_{i,j}\leq \tmax^j, && \forall i \in \points,~\forall j \in \needles^{\leftrightarrow}\label{model:freeneedle:phase2:clustering:con772}, \\
&&&  \start^j\in \bigcup_{e \in \en}~\dn_e, && \forall j \in \needles \label{model:freeneedle:phase2:clustering:con7771}, 
\end{alignat}
\end{subequations}
with the requirements that 
($i$) $\dir^j\in \Re^3$, $\forall j \in \needles$, 
($ii$) $(\tmin^j,\tmax^j) \in \Re^2$,  $\forall j \in \needles^{\leftrightarrow}$
($iii$) $\dualvar^{j,j'} \in \Re^{\nv{\mdla}}$, $\forall j' \neq j \in \needles$,
($iv$) $\dualvarw^{j,m} \in \Re^{\nw{\mdlb}}$, $\forall j \in \needles ,~\forall  m \in \polyhyd$, and 
($v$) $\stept_{i,j}\in \Re$ and $\asgt_{i,j} \in \{0,1\}$, $\forall i \in \points,~\forall j \in \needles$.
Objective function \eqref{model:freeneedle:phase2:clustering:obj} minimizes the total squared distance between the candidate dwell locations and their projections onto their assigned needles. 
Constraint \eqref{model:freeneedle:phase2:clustering:con1}  requires that each dwell position is assigned to exactly one needle. 
Constraint \eqref{model:freeneedle:phase2:clustering:con11}  ensures that each needle is assigned at least two dwell positions so that its associated direction is meaningful. 
Constraint \eqref{model:freeneedle:phase2:clustering:con66} normalizes the direction of each needle.
Constraint \eqref{model:freeneedle:phase2:clustering:con2} imposes that pairs of distinct needles do not intersect.
Constraint~\eqref{model:freeneedle:phase2:clustering:con4} prevents needles from intersecting with the body structures to avoid. 
Constraint~\eqref{model:freeneedle:phase2:clustering:con772} restricts the projection points on a finite needle to be between its extremities.
Constraint~\eqref{model:freeneedle:phase2:clustering:con7771} restricts the starting point of each needle to belong to one of the exiting planes.

Model $\Mmdl$ can be viewed as a clustering model that creates groups of candidate dwell positions, each associated with a needle.  
Unlike traditional clustering models, $\Mmdl$ does not assign a centroid to a cluster but a straight line (needle) that best fits the collection of its points. 
%Exact solution methods for traditional 
Clustering models are computationally demanding to solve globally and heuristics  are typically favored. 
We thus propose a heuristic for $\Mmdl$, 
similar to the k-means algorithm of \cite{hartigan1979algorithm},
that  alternates between two steps, each solving a different restriction of the model until convergence is achieved. 
The first assigns dwell positions to needles whose characteristics are fixed, with the objective of minimizing squared distances. 
The second calculates best parameters for the needles based on the specific set of dwell points assigned to them. 
It thus allows needles to more accurately match the dwell points in their clusters.
We give details next.

\subparagraph{Step 1 (Assigning dwell positions to needles).}
\label{step1}
%%%%%%%%%%%%%%%%%%%%%%%%%%%%%%%%%%%%%%%%%%%%%%%%%%%%%%%%

In this step, we are given a collection of nonintersecting needles with known parameters $(\fix{\start}^j,\fix{\dir}^j,\fix{\tmin}^j,\fix{\tmax}^j)$ for $j \in \needles$ 
that do not intersect with the body structures we wish to avoid but intersect with exiting planes.
Our goal is to assign each dwell position to one of these needles so as to minimize the total squared distance between dwell points and their assigned needles. 
This step is performed by solving an optimization problem over the variables $\asgt_{i,j}$ and constraints \eqref{model:freeneedle:phase2:clustering:con1} and \eqref{model:freeneedle:phase2:clustering:con11}. 
Because parameters $(\fix{\start}^j,\fix{\dir}^j,\fix{\tmin}^j,\fix{\tmax}^j)$ for $j \in \needles$ are fixed, the objective function \eqref{model:freeneedle:phase2:clustering:obj} reduces to $\sum_{ i\in \points}\sum_{ j\in \needles}{\Dis_{i,j}^2} \asgt_{i,j}$ where 
$\Dis_{i,j}^2= ||\dwpt^i-\fix{\start}^j-\lambda^* \fix{\dir}^j||_2^2\label{eu_eqn2}$,
$\lambda^*=\max\{\fix{\tmin}^j,\min\{\fix{\tmax}^j,\vatau^* \}\}$, and   $\vatau^*=\frac{(\dwpt^i-\start^j) \tr \dir^j}{\mynorm{\dir^j}_2^2}$.
This expression captures that the projection is either at one of the endpoints of the needle, or corresponds to what the projection would have been, $\fix{\start}^j+ \vatau^* \fix{\dir}^j$, if the needle was infinite.
When the needle is infinite, it holds that $\lambda^*=\vatau^*$. 
Without \eqref{model:freeneedle:phase2:clustering:con11}, this model could be solved with a greedy heuristic. 
With \eqref{model:freeneedle:phase2:clustering:con11}, this model is a transportation problem where flow is transported between points $\points$ and needles $\needles$. 
Transportation models admit a variety of efficient solution algorithms.
They are also easy to solve using commercial solvers.
%\byjp{Should we define $\points_{k}$ here?}

\subparagraph{Step 2 (Finding needle directions).}
\label{step2} 
%%%%%%%%%%%%%%%%%%%%%%%%%%%%%%%%%%%%%%%%%%%%%%%%%%%%%%%%

In this step, we assume an assignment where each dwell position is matched to exactly one needle, and each needle has at least two dwell positions.
Thus, the dwell positions are organized into clusters of size at least two. 
The goal is to choose the parameters ($\start^j$,$\dir^j$,$\tmin^j$,$\tmax^j$) of each needle/cluster $j \in \needles$ to minimize total squared distance to each of its assigned dwell points.   
This corresponds to solving Model $\Mmdl$ in which variables $\asgt_{i,j}$ are fixed. 
This model has objective \eqref{model:freeneedle:phase2:clustering:obj} and constraints 
\eqref{model:freeneedle:phase2:clustering:con66}, 
\eqref{model:freeneedle:phase2:clustering:con2}, 
\eqref{model:freeneedle:phase2:clustering:con4},
\eqref{model:freeneedle:phase2:clustering:con772}, and
\eqref{model:freeneedle:phase2:clustering:con7771}.

Because of the requirement that needles do not intersect, Model $\Mmdl$ does not decompose into a collection of independent smaller models for each needle and, ideally, should be solved as a monolithic optimization problem.   
This approach is not tractable with current commercial solvers.
To overcome this difficulty, we use a greedy approach that determines the characteristics of each needle $j \in \needles$,
one at a time, after fixing the parameters of those considered previously.
As the complexity of these models increases with $j$, we also apply an alternating heuristic to solve them. 
This heuristic considers three restrictions -- in which the variables being optimized are ($i$) $\dir^j$, ($ii$) $\start^j$ and $\stept_{i,j}$, and ($iii$) $\tmin^j$, $\tmax^j$, and $\stept_{i,j}$, respectively -- in sequence, as  detailed in \onlsup~\ref{section:os:clusteringalgo}.

%%%%%%%%%%%%%%%%%%%%%%%%%%%%%%%%%%%%%%%%%%%%%%%%%%%%%%%%
\subsubsection{Phase 3: Deriving treatment plans.}
\label{section:modelsmethods:free:phase3}
%%%%%%%%%%%%%%%%%%%%%%%%%%%%%%%%%%%%%%%%%%%%%%%%%%%%%%%%

Once needle positions have been determined, dwell locations spaced apart by $1$mm -- which is the spacing used in the clinically available treatment
planning software -- are created on each needle. 
A treatment plan is then generated by solving \eqref{model:freeneedle:phase1} with the main difference that dose parameters $\doselinear_{i,k,l}$ are now computed using the line-source formula of the AAPM task group 186~\citep{beaulieu2012report} since the orientation of the source at each dwell location is now known from the needle direction.

%%%%%%%%%%%%%%
%%%%%%%%%%%%%%%%%%%%%%%%%%%%
%%%%%%%%%%%%%%%%%%%%%%%%%%%%%%%%%%%%%%%%%%%%%%%%%%%%%%%% 
\section{Experimental procedures}
%\label{clinical method}
\label{section:procedures}
%%%%%%%%%%%%%%%%%%%%%%%%%%%%%%%%%%%%%%%%%%%%%%%%%%%%%%%%
%%%%%%%%%%%%%%%%%%%%%%%%%%%%
%%%%%%%%%%%%%%

To evaluate the effectiveness of the approaches presented above, we study the case of the patient with skin cancer on the nose shown in Figure~\ref{figure:face}.  
Needles for this procedure have radius $\radius=1.5$mm. 
To ensure ample clearance in the solution, we increase the radius to $\radius=1.55$mm.
There are four basic regions of interest: ($i$) the tumor areas of the nose, ($ii$) the right eye lens, ($iii$) the left eye lens, and ($iv$) the skin.
This latter OAR is shallow and extends only a few millimeters below the surface.
Radiation oncologists prefer treatment plans where similar doses are delivered on the surface area (\textit{resp.}, on the boundary) of the tumor. 
This is complicated, in this case, by the fact that OARs and tumor regions intersect, and by the fact that the geometry of the nose may lead to dose differences between its left and right side. 
For this reason, the basic structures are divided further. 
Specifically, a new OAR is created by removing from the skin OAR  all areas that also belong to the tumor.
Further, the left and right tumor areas are divided into their surface and boundary components.

The resulting structures are depicted in Figure~\ref{figure:dvhstructures}. 
They include OARs listed as left/right eye lens (LE/RE),  skin excluding tumor regions (SW), skin including tumor regions (ST), tumor regions defined as tumor boundary on the left/right side of the nose (LB/RB), and surface tumor on the left/right side of the nose (LS/RS).

\begin{figure}[!htb]
\centering
\includegraphics[scale=0.6]{ 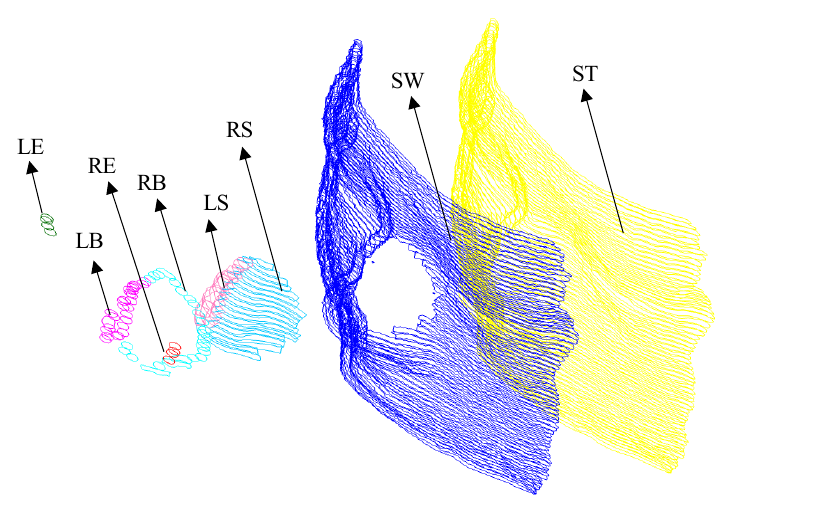}
\caption{3D representation of body structures included in DVHs computations}
\label{figure:dvhstructures}
\end{figure}

This patient was treated using a Freiburg flap applicator with six needles.
The clinical plan uses the active dwell positions represented by spheres in Figure~\ref{figure:needle:clinic}. 
We use the clinical dwell times to compute dose volume histograms (DVHs) and associated dosimetric indices.
To this end, we voxelize the basic structures using voxel with side length of $1.5$mm, resulting in $21219$ total voxels, of which $2140$ correspond to the tumor and $19079$ do not. 
Resulting DVHs are shown in Figure~\ref{figure:dvh:clinic}.
%The voxelizations we use for optimization models will typically be coarser, as we describe later. 
%\byjp{say that we do this in general for all DVHs... but for optimization models, we might use coarser voxelizations to account for the difficulty of the optimization problem}

%For computing DVHs, we use $21219$ voxels with side length of $1.5$mm,  of which $2140$ correspond to tumor and $19079$ do not. 

\begin{figure}[!htb]
\begin{subfigure}{.5\textwidth}
\centering
\includegraphics[scale=0.5]{ 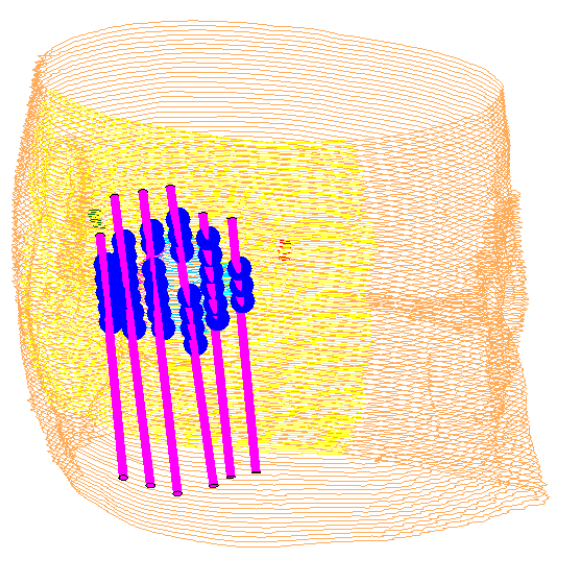}
%\vskip 0.3 cm
\captionsetup{justification=centering,font=scriptsize}
\caption{Needle configuration}
\label{figure:needle:clinic}
\end{subfigure}
\hfill
\begin{subfigure}{.5\textwidth}
\centering
\includegraphics[scale=0.5]{ 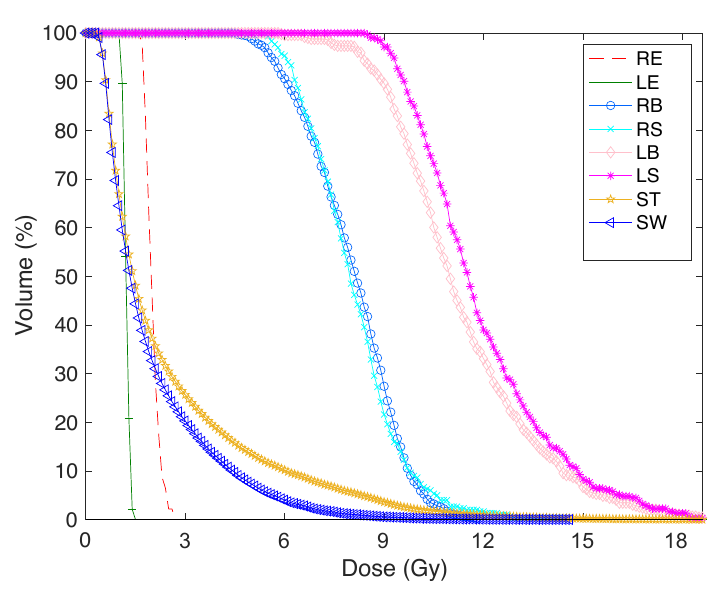}
\captionsetup{justification=centering,font=scriptsize}
\caption{DVHs}
\label{figure:dvh:clinic}
\end{subfigure}
\caption{Clinical results}
\label{figure:clinic}
\end{figure}  
Figure~\ref{figure:dvh:clinic} shows that doses among the various body structures are inhomogeneous. 
Specifically, the left side of the tumor surface receives higher doses compared to its right side, and there is a noticeable difference in the amount of radiation received by the boundary voxels on the left and right side of the nose. 
Further, the skin tissues exhibit long tails in their DVHs and most voxels on the left side of the tumor are exposed to a high dose.

In the following section, we will compare the characteristics of this clinical plan with those obtained with the approaches we propose.
To make the comparison easier, we normalize the treatment plans by uniformly scaling their dwell times so that their $RB_{V_{100}}$ are approximately the same. 
%\byjp{I added the following:}\bynasim{Please check that this makes sense. Answer: it makes sense}
In this process, $RB_{V_{100}}$ describes the volume of target tissue $RB$ that receives 100\% of the prescribed radiation dose. 
The maximum normalization error is $0.001$\%.
This normalization error is calculated as $\sfrac{\left|RB_{V_{100}}^1-RB_{V_{100}}^2\right|}{\max\left(RB_{V_{100}}^1,RB_{V_{100}}^2\right)}*100$, where $RB_{V_{100}}^1$ and $RB_{V_{100}}^2$ are the dosimetric values, after normalization, for the two treatment plans 
compared.

In our models, we consider the five exiting planes shown in Figure~\ref{figure:exiting}.
To avoid negative bounds on some of the variables, which  can affect the quality of relaxations of bilinear terms, we translate all of the body structures by the same vector, so that they all belong to the positive orthant.  
Then, all needle reference points can be chosen among the faces of the box $[1,173] \times [1,60] \times [95,200]$. %which are the exiting planes discussed in Section~\ref{section:modelfeatures:exitingplanes}.

%\byjp{This needs some checking....}
For all optimization models that require dose computation, we use piecewise-linear penalty function $f_i^\tumor(s) = \max\{-5000s,0,5000(s-3) \}$ with $\prescr_i=6$ for each voxel $i$ in a tumor region.
We use penalty function $f_i^\oars(s) = \max \{0,5000(s-2) \}$ with $\prescr_i=0$ for each voxel $i$ not in a tumor region. 
We also set the parameter $\Tmax=10$s.

\section{Results}
\label{section:results}
%%%%%%%%%%%%%%%%%%%%%%%%%%%%%%%%%%%%%%%%%%%%%%%%%%%%%%%%
%%%%%%%%%%%%%%%%%%%%%%%%%%%%
%%%%%%%%%%%%%%

%%%%%%%%%%%%%%%%%%%%%%%%%%%%
%%%%%%%%%%%%%%%%%%%%%%%%%%%%%%%%%%%%%%%%%%%%%%%%%%%%%%%%
\subsection{Fixed needles model}
\label{section:results:fixed}
%%%%%%%%%%%%%%%%%%%%%%%%%%%%%%%%%%%%%%%%%%%%%%%%%%%%%%%%
%%%%%%%%%%%%%%%%%%%%%%%%%%%%

We evaluate Model~\eqref{model:fixedneedle} on the case described in Section~\ref{section:procedures}. 
%The parameters used for this model are discussed in detail in \onlsup~\ref{mod1param}.
To generate needles close to the tumor,
%these concerns, 
we perform a voxelization of the boundary of the nose. 
%in our case study.
We use voxels with side lengths of $4.5$mm. 
We then shift the centers of these voxels away from the face by $3$mm, $6$mm, and $9$mm to produce a set of points that we will use as prospective endpoints for the needles. 
%\byjp{WHAT IS THE Y DIRECTION HERE? IS IT SPECIFIC TO THIS PATIENT?} \\
We discard all points that are inside of the patient's body. 
From each of the remaining ones, we create a half-line to each of the exiting planes, using the normal vectors to exiting planes as directions. 
We consider the point where the half-line intersects its corresponding exiting plane as the starting end of the needle. 
%This collection of points gives rise to a large number of possible needles.  
Finally, we eliminate all needles that intersect the patient's body using system $\SysneedleP^{\leftrightarrow}_{2,\radius}[\cdot]$ for each structure of interest.
We determine whether pairs of needles intersect using $\Systwoneedles^{\leftrightarrow}_{2,\radius}[\cdot]$.

For the case under study, the above procedure generates a collection of $138$ endpoints that are then considered in conjunction with the five exiting planes of Figure~\ref{figure:exiting}.
After removing needles that intersects with the patient body, as modeled by the polytope with $142$ planes represented in Figure~\ref{figure:convexhull:142points} of \onlsup~\ref{section:os:faceconvexhull}, 
% %using $PP^{\leftrightarrow}_{2,1.55\text{mm}}[\cdot]$, as discussed in Section~\ref{Intersection of needles with body structures}, 
we are left with a collection $\sn$ with 100 candidate needles. 

To create the collection of voxels required by Model~\eqref{model:fixedneedle}, we voxelize all the structures other than the skin using voxels with side length of $1.5$mm.
For the skin, we use voxels with side length of $4.5$mm 
instead as using the finer voxelization would have produced $19067$ voxels, which is much larger than the $2235$ voxels of other regions and yields untractable models.  
Using this mixed voxelization yields a total of $3681$ voxels, of which $2140$ correspond to tumor and $1541$ do not.

We solve Model~\eqref{model:fixedneedle} with \gurobi\xspace 9.0.2 with parameter $\Nmax=6$. .
%to find an optimal set of needles and their corresponding active dwell positions and dwell times. 
The resulting configuration of needles and the corresponding DVHs are presented in Figures~\ref{figure:needle:fixed} and~\ref{figure:dvh:fixed}, respectively.
In Figure~\ref{figure:needle:fixed}, the blue spheres represent the dwell positions that are activated through the solution of the model. 
%\red{JP: I THINK WE NEED TO SAY AT LEAST THAT WE ARE GENERATING NEEDLES PARALLEL TO THE MAIN AXES...}
Figure~\ref{figure:dvh:fixed} shows that the treatment plan generated is more homogeneous than that produced using a Freiburg flap, as the difference between the doses delivered to the right and left sides of the nose is reduced. 
Further, this plan 
%created from needles we selected 
ensures that both surface and boundary structures of the nose's right and left sides receive more similar doses compared to the clinical plan presented in Figure~\ref{figure:clinic}.

\begin{figure}[!htb]
\begin{subfigure}{.5\textwidth}
\centering
\includegraphics[scale=0.5]{ 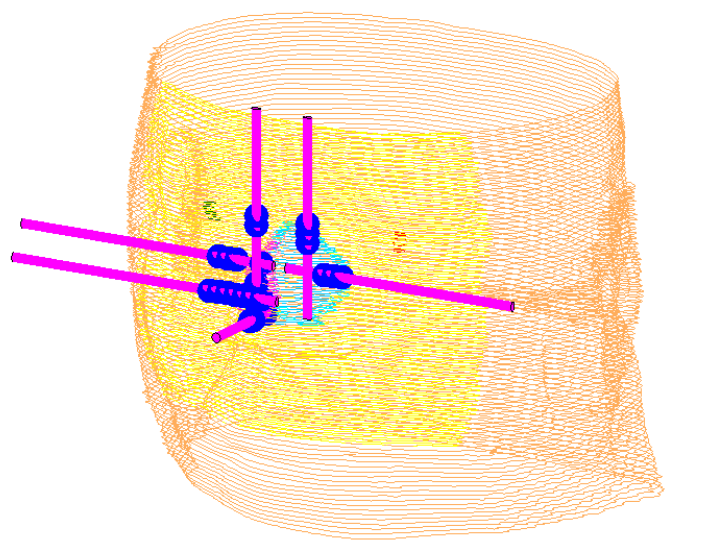}
\vskip 0.3 cm
\captionsetup{justification=centering,font=scriptsize}
\caption{Needle configuration}
\label{figure:needle:fixed}
\end{subfigure}
\begin{subfigure}{.5\textwidth}
\centering
\includegraphics[scale=0.5]{ 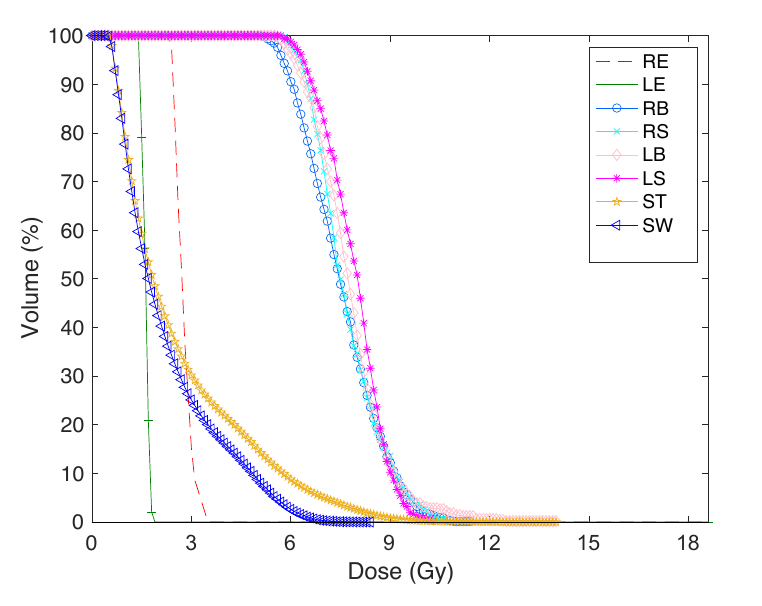}
\captionsetup{justification=centering,font=scriptsize}
\caption{DVHs}
\label{figure:dvh:fixed}
\end{subfigure}
\caption{Fixed needles model results}
\label{figure:fixed}
\end{figure}

%%%%%%%%%%%%%%%%%%%%%%%%%%%%
%%%%%%%%%%%%%%%%%%%%%%%%%%%%%%%%%%%%%%%%%%%%%%%%%%%%%%%%
\subsection{Free needles models}
\label{section:results:free}
%%%%%%%%%%%%%%%%%%%%%%%%%%%%%%%%%%%%%%%%%%%%%%%%%%%%%%%%
%%%%%%%%%%%%%%%%%%%%%%%%%%%%

%%%%%%%%%%%%%%%%%%%%%%%%%%%%%%%%%%%%%%%%%%%%%%%%%%%%%%%%
\subsubsection{Phase 1: Finding candidate dwell positions.}
\label{section:results:free:phase1}
%%%%%%%%%%%%%%%%%%%%%%%%%%%%%%%%%%%%%%%%%%%%%%%%%%%%%%%%

We evaluate Model~\eqref{model:freeneedle:phase1} for the case of Section~\ref{section:procedures}.
%We use the same parameters used in the previous section. 
To generate the prospective dwell positions, we voxelize the nose using voxels with $1.5$mm side-length.
%, which are ten times smaller than those used earlier. 
%\byjp{DO WE GENERATE THOSE ONLY ON THE BOUNDARY OF THE FACE?... just boundary of nose + move in the right direction}
After voxelizing the boundary of the nose, we compute the centers of all voxels and shift them away from the face by $3$mm, $6$mm, and $9$mm as the needles that will cover these points are $3$mm in diameter. 
We eliminate locations inside body structures to produce a set $\dwell$ of prospective dwell points. 
This finer voxelization is used because we are generating dwell locations and not only needle endpoints (as in Section~\ref{section:results:fixed}) and because we can solve LP \eqref{model:freeneedle:phase1} even for voxelizations with side-length 1.5mm. 
The set $\points$ of candidate dwell points required for the subsequent models is formed by choosing the $\Imax=50$ positions with largest dwell times. 
Model~\eqref{model:freeneedle:phase1} is solved using \gurobi\xspace 9.0.2. 
Figure~\ref{figure:promisingdwell}
%in \onlsup~\ref{section_nosefarkas233n}
 provides a graphical representation of the solution. 
%\bynasim{What is the voxelization used for Model~\eqref{model:freeneedle:phase1}: Answer: $1.5$mm side-length voxels}

\begin{figure}[!htb]
\centering
\vskip -3cm
\includegraphics[scale=0.4]{ 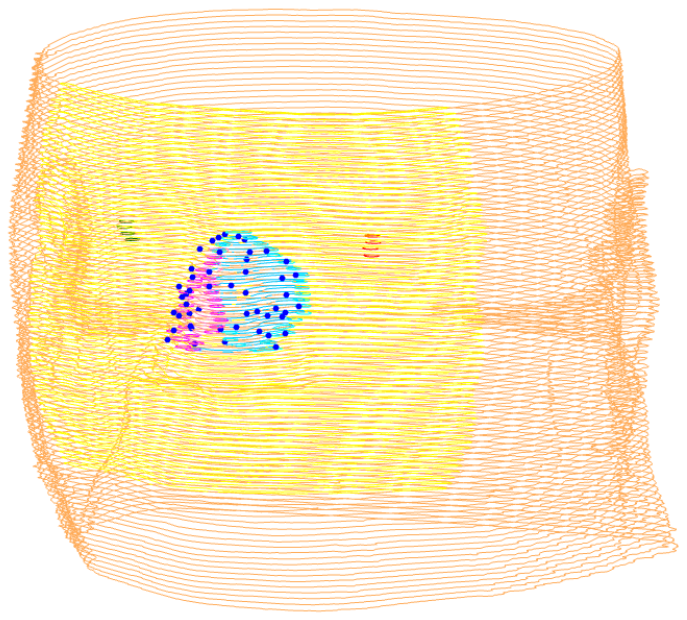}
\vskip -3cm
\captionsetup{justification=centering}
\caption{3D representation of candidate dwell positions}
\label{figure:promisingdwell}
\end{figure}

%%%%%%%%%%%%%%%%%%%%%%%%%%%%%%%%%%%%%%%%%%%%%%%%%%%%%%%%
\subsubsection{Phase 2: Finding needles associated with candidate dwell positions.}
\label{section:results:free:phase2}
%%%%%%%%%%%%%%%%%%%%%%%%%%%%%%%%%%%%%%%%%%%%%%%%%%%%%%%%

%Next, we solve model $\Mmdl$ for the dwell positions $\points$ using Algorithm~\ref{Algorithm1}.
% The parameter used for this model are discussed in detail in \onlsup~\ref{modres2}.

\paragraph{Method 1: Maximum coverage model.}

We apply Model~\eqref{model:freeneedle:phase2:coverage} with the candidate dwell positions obtained above. 
%also apply maximum coverage algorithm proposed in Section~\ref{section:modelsmethods:free:phase2} on the case described in Section~\ref{section:procedures}.
%The details of the parameters used for this model and the results can be found in \onlsup~\ref{modres2}.
%We apply a maximum coverage model with $\Nmax=50$ \byjp{overload in notation} dwell positions that were given the longest dwell times in the solution of model described in Section~\ref{section:modelsmethods:free:phase1}.
Associated with each pair of dwell positions is a single straight line.  
If this line intersects two exiting planes, we may create up to three needles; one originating from the first exiting plane, another originating from the second exiting plane, and a last one having one endpoint in each of the exiting planes. 
Allowing for shorter needles to be created help make the later process of selecting needles that do not intersect each other and body structures easier. 
%This collection of points gives rise to a large number of possible needles. 
We use the polytope with $142$ planes described in Figure~\ref{figure:convexhull:142points} of \onlsup~\ref{section:os:faceconvexhull} as the representation of the body structure not to be crossed by needles. 
We use this representation to remove needles that intersect with the patient body, as discussed in Section~\ref{section:modelfeatures:needle-polytope}. 
We are left with a collection of 333 candidate needles.
%We use the same parameters used in \ref{section:results:fixed} for this model to be consistent. 
Model~\eqref{model:freeneedle:phase2:coverage} is solved using \gurobi\xspace 9.0.2. 
The generated configuration of needles is represented in Figure~\ref{figure:needle:free:coverage}.

%\red{JP: WHAT IS THE TAKEAWAY FROM THAT DISCUSSION?}
%The configuration of needles used for the plan is shown in Figure~\ref{figure:needle:free:coverage}.

\begin{figure}
\centering
\begin{subfigure}[b]{0.5\textwidth}
\centering
\includegraphics[scale=0.6]{ 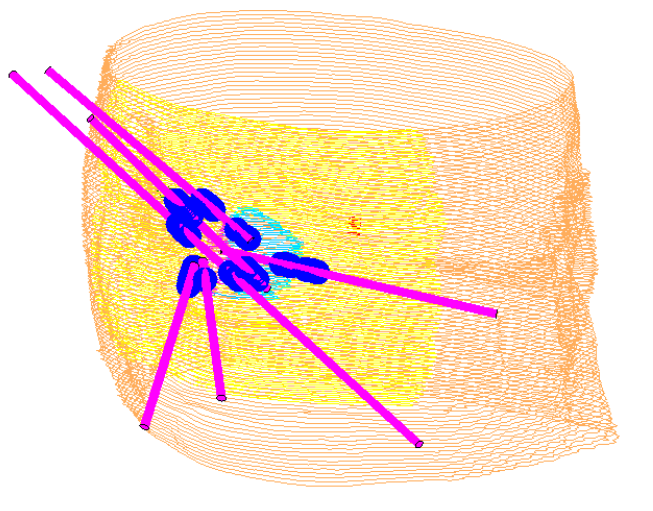}
\captionsetup{justification=centering,font=scriptsize}
\caption{Needle configuration}
\label{figure:needle:free:coverage}
\end{subfigure}
\hfill
\begin{subfigure}[b]{0.45\textwidth}
\centering
\includegraphics[scale=0.6]{ 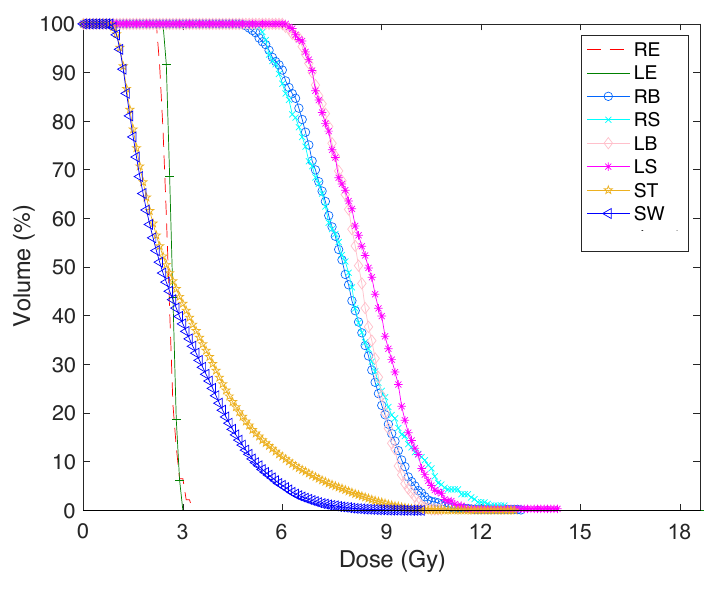}
\captionsetup{justification=centering,font=scriptsize}
\caption{DVHs}
\label{figure:dvh:free:coverage}
\end{subfigure}
\caption{Maximum coverage results} 
\label{figure:free:coverage}
\end{figure}

\paragraph{Method 2: Clustering approach.}
We apply Model~\ref{model:freeneedle:phase2:clustering} with the candidate dwell positions obtained above.
This model requires the choice of a system $\Systwoneedles^{\mdla}_{\normp,\radius}[\cdot]$ to impose that needles do not intersect and of a system $ \SysneedleP^{\mdlb}_{\normp,\radius}[\cdot]$ to guarantee that needles do not cross the face. 
Here, we select $\normp:=2$, $\mdla:=``\times/="$ and $\mdlb=``\times"$. 
%$\Systwoneedles^{\times/=}_{2,\radius}[\cdot]$ to prevent needle intersections and we use $ \SysneedleP^{\times}_{2,\radius}[\cdot]$ to prevent intersection of needles with the body structure illustrated in 
To keep the model sas simple as possible, we 
describe the body structures not to be crossed by needles to be the polytope defined by $7$ planes depicted 
%A graphical depiction of these planes is given 
in Figure~\ref{figure:convexhull:7points} of \onlsup~\ref{section:os:faceconvexhull}.
%The discussion in  allows for different variants of the clustering approach to be created. 
%\byjp{How many restarts do we do?}
%\bynasim{What is the number? Answer:10}
We use six needles, \ie $|\needles|=6$.

Being inherently heuristic, the clustering approach depends on the initial choice the needles. 
%\byjp{How many needles do we use. 6?}
%\bynasim{Please check:Answer: 6 needles as clinical case has 6 needles}
The needle configurations obtained from two different starting points 
are shown in Figure~\ref{figure:needles:free:clustering}.
%\red{JP:SHOULD WE SIMPLY GIVE THE VALUES oF $\square$ and $\triangle$ INSTEAD? IS THE POINT OF THIS PARAGRAPH THE FACT THAT STARTING CONDITIONS ARE NOT VERY IMPORTANT?}
This illustrates that the clustering algorithm does not appear to be very sensitive to the initial conditions as the configuration of needles (and their associated DVHs which we will discuss next) are very similar.
In our implementations, we select the best solution obtained after ten restarts, \ie $L=10$ in Algorithm~\ref{alg:jp} of \onlsup~\ref{section:os:clusteringalgo}.
%\byjp{THE DWELL POINTS SEEM TO BE VERY DIFFERENT IN THE TWO CASES. IS THAT NORMAL? THESE ARE THE DWELL POISITIONS CHOSEN IN PHASE 3....}
An analysis of the performance of variants for other choices of $\normp$, $\mdla$, and $\mdlb$ can be found in 
%\onlsup~\ref{modres2}.
%\red{JP: WHAT IS THE TAKEAWAY FROM THAT DISCUSSION?}
%\byjp{Need to introduce a reference to the online supplement}
\onlsup~\ref{section:os:clusteringvariants}.

\begin{figure}[!htb]
\begin{subfigure}{.5\textwidth}
\centering
\includegraphics[scale=0.5]{ 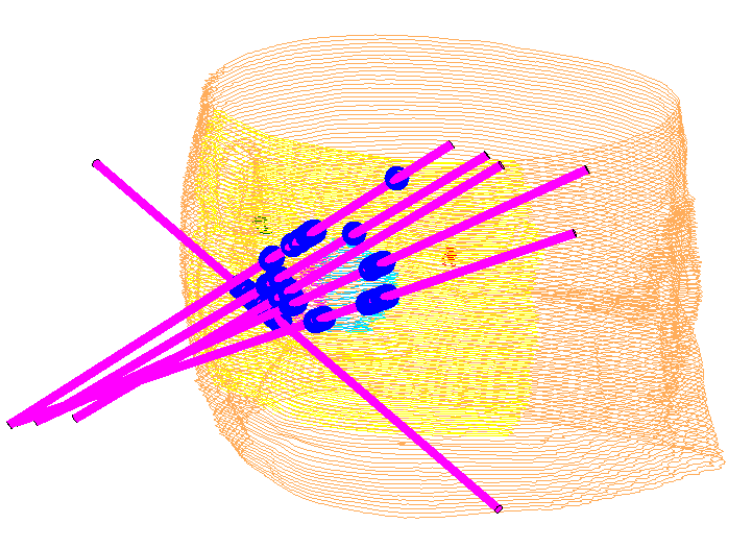}
\captionsetup{justification=centering,font=scriptsize}
\caption{Starting direction $\dir=[0.3, 0.9,0.1]$}
\label{figure:needle1:free:clustering}
\end{subfigure}
\begin{subfigure}{.5\textwidth}
\centering
\includegraphics[scale=0.5]{ 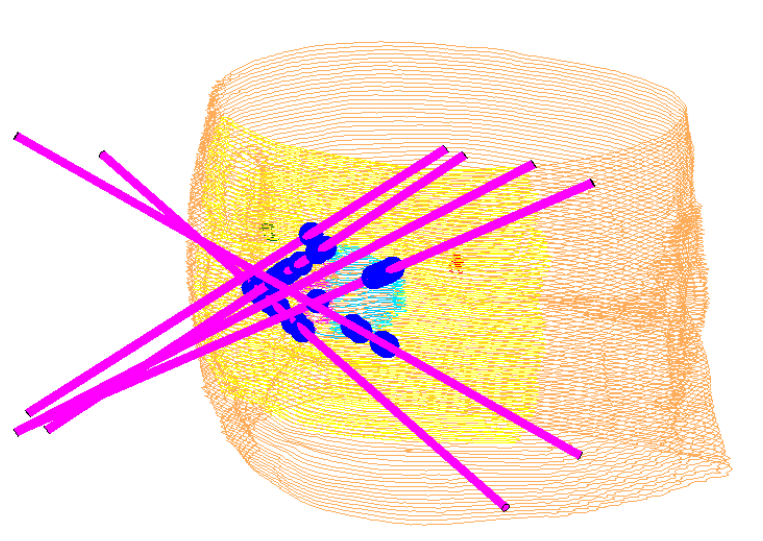}
\captionsetup{justification=centering,font=scriptsize}
\caption{Starting direction $\dir=[0.7, -0.5,0.5]$}
\label{figure:needle2:free:clustering}
\end{subfigure}
\caption{Clustering approach needle configurations with different initializations}
\label{figure:needles:free:clustering}
\end{figure}

\begin{figure}[!htb]
\begin{subfigure}{.5\textwidth}
\centering
\includegraphics[scale=0.5]{ 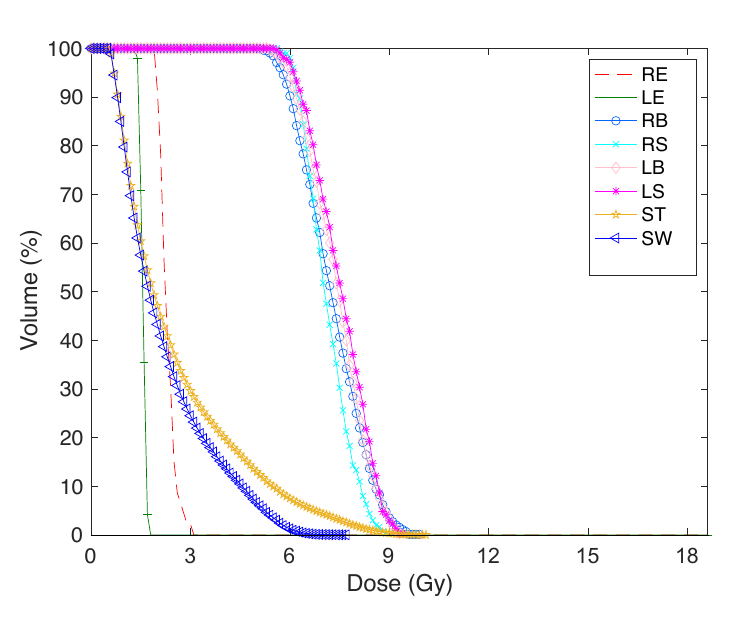}
\captionsetup{justification=centering,font=scriptsize}
\caption{DVHs for needle configuration of Figure~\ref{figure:needles:free:clustering}(a)}
\label{figure:dvh1:free:clustering}
\end{subfigure}
\begin{subfigure}{.5\textwidth}
\centering
\includegraphics[scale=0.5]
{ 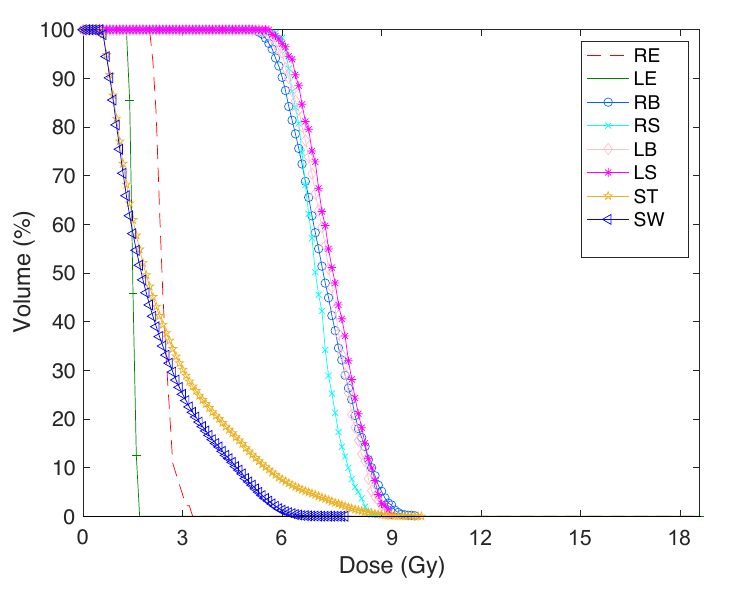}
\captionsetup{justification=centering,font=scriptsize}
\caption{DVHs for needle configuration of Figure~\ref{figure:needles:free:clustering}(b)}
\label{figure:dvh2:free:clustering}
\end{subfigure}
\caption{Clustering approach DVHs}
\label{figure:dvhs:free:clustering}
\end{figure}

%%%%%%%%%%%%%%%%%%%%%%%%%%%%%%%%%%%%%%%%%%%%%%%%%%%%%%%%
\subsubsection{Phase 3: Deriving treatment plans.}
\label{section:results:free:phase3}
%%%%%%%%%%%%%%%%%%%%%%%%%%%%%%%%%%%%%%%%%%%%%%%%%%%%%%%%

Dwell locations spaced apart by $1$mm are created on each of the needles that are selected. 
A treatment plan is generated using the LP described in Section~\ref{section:modelsmethods:free:phase3} for a voxelization with side-length $1.5$mm. 
Unlike the fixed-needle case, the needles are not aligned with the principal axes. 
The blue spheres in Figures~\ref{figure:needle:free:coverage}, \ref{figure:needle1:free:clustering}, and \ref{figure:needle2:free:clustering} represent the dwell positions that are activated through the solution of this LP. 
%\byjp{Are we solving an IP or an LP? Is it Section 3.1 that we should reference?}
%\bynasim{Can you check if it is an IP or an LP?}
%\bynasim{A LP model is solved, dwell times (continuous variables are optimized in this phase) }
%\byjp{What is the voxelization used for the treatment plans?}
%\bynasim{Is it the same voxelization as 5.2.1? Yes same as 5.2.1, $1.5$mm side-length voxels} \\
%\lastcmt{For Nasim: I changed the text based on your email, since the earlier version used to say that we were solving a MIP similar to the fixed needle case...} \\
%\bynasim{This one makes sense. Previous version was not accurate.}
%\bynasim{Are the blue dots in Figure~\ref{figure:needles:free:clustering} the dwell positions with positive dwell times? If so, should there be more of them? What are the blue dots? }
%\bynasim{Answer: we discussed this, as these are the final dots that are chosen in phase 3, they are different. }

Figure~\ref{figure:dvh:free:coverage} presents the DVHs of plans obtained using the needle configuration generated by the maximum coverage algorithm. 
It can be observed that this configuration results in a more homogeneous plan compared to the one created using the Freiburg flap. 
This is evident from the fact that the difference between the doses delivered to the right and left sides of the nose is reduced. 
Moreover, the surface and boundary structures in both the right-hand-side and left-hand-side of the nose receive lower doses, in comparison to the original plan whose DVHs were given in Figure~\ref{figure:dvh:clinic}.

The DVHs corresponding to the two plans obtained with the needle configurations generated by the clustering approach are represented in Figure~\ref{figure:dvhs:free:clustering}. 
We see that these configurations permit the creation of more homogeneous plans than those obtained with the clinically-used Freiburg flap. 
This is evidenced by the fact that the difference between the doses delivered to the right and left sides of the nose are reduced. 
Further, for the plan generated from the needles we selected,  surface and boundary structures in both the right-hand-side and left-hand-side of the nose receive more similar doses and the skin structures receive much less doses, when compared to the original plan whose DVHs were given in Figure~\ref{figure:dvh:clinic}.

% \begin{figure}
%      \centering
%      \begin{subfigure}[b]{0.5\textwidth}
%         \centering
% \includegraphics[scale=0.6]{ colneedfinal.pdf}

% \captionsetup{justification=centering,font=scriptsize}
% \caption{Final needle configuration using method introduced in \ref{using maximum coverage} }
% \label{nosefarkas15}
%      \end{subfigure}
%      \hfill
%      \begin{subfigure}[b]{0.45\textwidth}
%       \centering
% \includegraphics[scale=0.6]{ dvhcolfinal.pdf}
% \captionsetup{justification=centering,font=scriptsize}
% \caption{DVH plot using method introduced in \ref{using maximum coverage} }
% \label{nosefarkas25}
% \end{subfigure}
%  \caption{Outputs of maximum coverage algorithm} 
% \label{nosefarkas25op}
% \end{figure}

%%%%%%%%%%%%%%%%%%%%%%%%%%%%
%%%%%%%%%%%%%%%%%%%%%%%%%%%%%%%%%%%%%%%%%%%%%%%%%%%%%%%%
\subsection{\textbf{Models comparison}}
\label{section:results:comparison}
%%%%%%%%%%%%%%%%%%%%%%%%%%%%%%%%%%%%%%%%%%%%%%%%%%%%%%%%
%%%%%%%%%%%%%%%%%%%%%%%%%%%%

First, we compare the plans with respect to their solution times.
In Table~\ref{table:solutiontimes:comparison}, $(P_1)$ is the problem of determining candidate dwell positions to be covered with needles, 
$(P_2)$ is the problem of determining needle positions, and 
$(P_3)$ is the problem of creating a plan. 
Steps $(P_1)$ and $(P_2)$ are not performed in the fixed needle approach.
For the free needle approaches, $(P_3)$ requires only the solution of an LP.
For the fixed needle approach, $(P_3)$ requires the solution of an MIP.

\begin{table}[!htb]
\centering
\begin{tabular}{|l| l|c| c| c| c|} 
\hline
Method &  Submethod & $P_1$ & $P_2$ & $P_3$ & Total \\
\hline
Fixed Needle  & -- & -- & -- & 5060.72 & 5060.72\\ 
\hline
Free Needle & Maximum coverage &  42.13  &296.06 & 11.52 & 349.71\\
\hline
Free Needle & Clustering &  42.13 & 35.15 &  10.21 & 87.49\\
\hline
\end{tabular}
\caption{Solution times (s)}
\label{table:solutiontimes:comparison}
\end{table}%
% \lastcmt{For Nasim: Problems $P_1$ and $P_3$ should be similar in that they involve solving an LP. Why are the run times so different. Is it related to the number of dwell points in each?}

% \bynasim{We have much less candidate dwell positions in P3, I believe that is why the solving time is different, I will run the P1,P3 and make sure about the run time. }

% \bynasim{Update: I ran it, time is upadted now, the previous time was for an IP version of model where I defined binary variables to activate 50 points but later we changed the method and sort the times and picked the largest ones.  }

Figure~\ref{figure:result:tumorvs} compares the normalized  $V_{100}$, $V_{150}$, and $V_{200}$ dosimetric indices for the three plans that were created for this case. 
For the clustering approach, the specific plan being analyzed is that corresponding to the needle configuration depicted in Figure~\ref{figure:needle1:free:clustering}.
%\byjp{Which one of the clustering plans is used?}
%\bynasim{Is the first clustering needle configuration used or the second one? Answer: First one is used.}
The axis of abscissas lists the four different tumor regions of interest (see Section~\ref{section:procedures}), whereas the axis of ordinates represents the value of the dosimetric index. 
The different plans are represented using different line styles: the clinical plan is dash-dotted, the fixed needle plan is dotted,  the clustering plan is dashed, and the maximum coverage plan is continuous.  
Figure~\ref{figure:result:OARvs}  similarly compares the normalized dosimetric indices relevant to the different OARs for the four plans we consider.

\begin{figure}[!htb]
\centering
\begin{subfigure}[b]{0.32\textwidth}
\centering
\includegraphics[width=\textwidth]{ 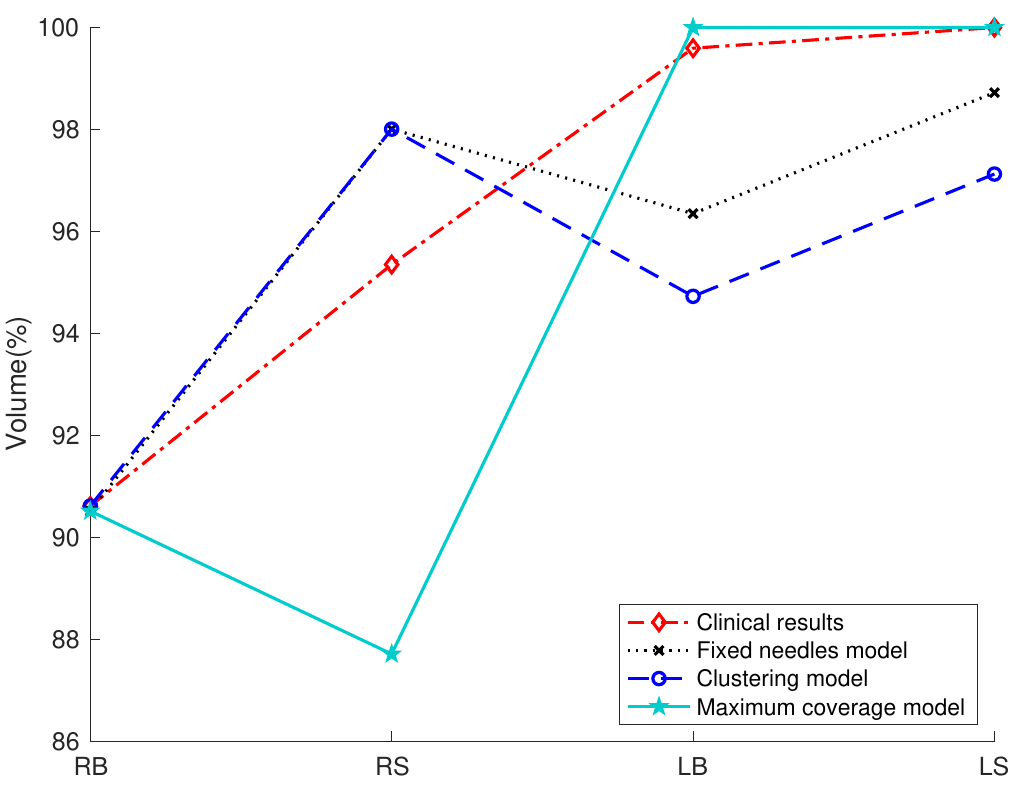}
\captionsetup{justification=centering,font=scriptsize}
\caption{$V_{100}$}
\label{figure:result:tumorv100}
\end{subfigure}
\hfill
\begin{subfigure}[b]{0.32\textwidth}
\centering
\captionsetup{justification=centering,font=scriptsize}
\includegraphics[width=\textwidth]{ 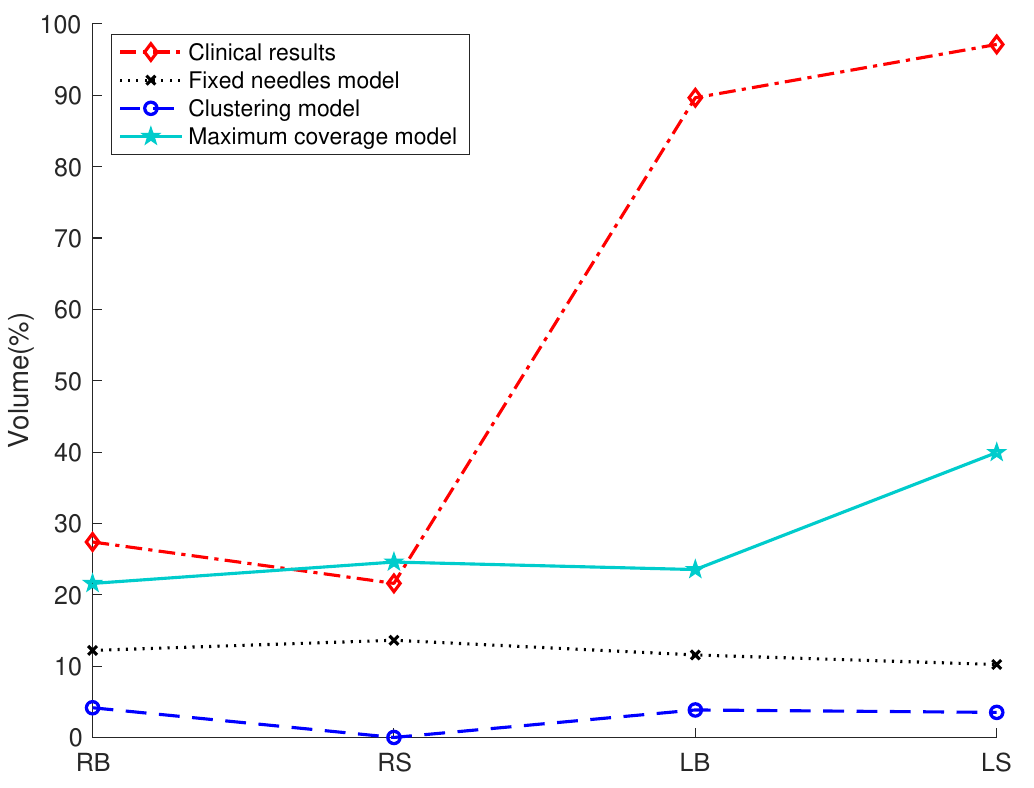}     
\caption{$V_{150}$ }
\label{figure:result:tumorv150}
\end{subfigure}
\hfill
\begin{subfigure}[b]{0.32\textwidth}
\centering
\captionsetup{justification=centering}
\includegraphics[width=\textwidth]{ 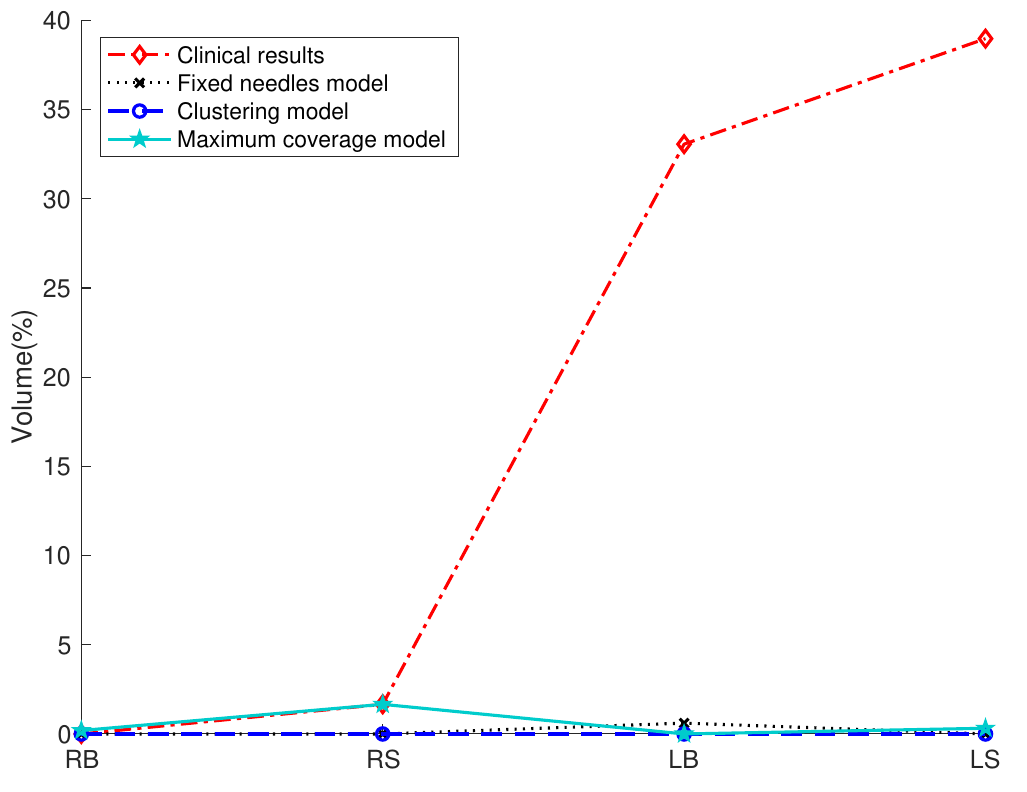}
\captionsetup{justification=centering,font=scriptsize}
\caption{$V_{200}$}
\label{figure:result:tumorv200}
\end{subfigure}
%     \hfill
%     \begin{subfigure}[b]{0.4\textwidth}
%         \centering
%         \captionsetup{justification=centering}
%         \includegraphics[width=\textwidth]{ empty.pdf}
%         \label{restumorempty}
%     \end{subfigure}
\caption{Dosimetric indices for tumor regions}
\label{figure:result:tumorvs}
\end{figure}
%\byjp{In Figure 8 and 9, can we replace the word axis parallel with "fixed needles" or "Fixed"}
%\bynasim{Graphs are modified with adjusted legends}

\begin{figure}[!htb]
\centering
\begin{subfigure}[b]{0.32\textwidth}
\centering
\includegraphics[width=\textwidth]{ 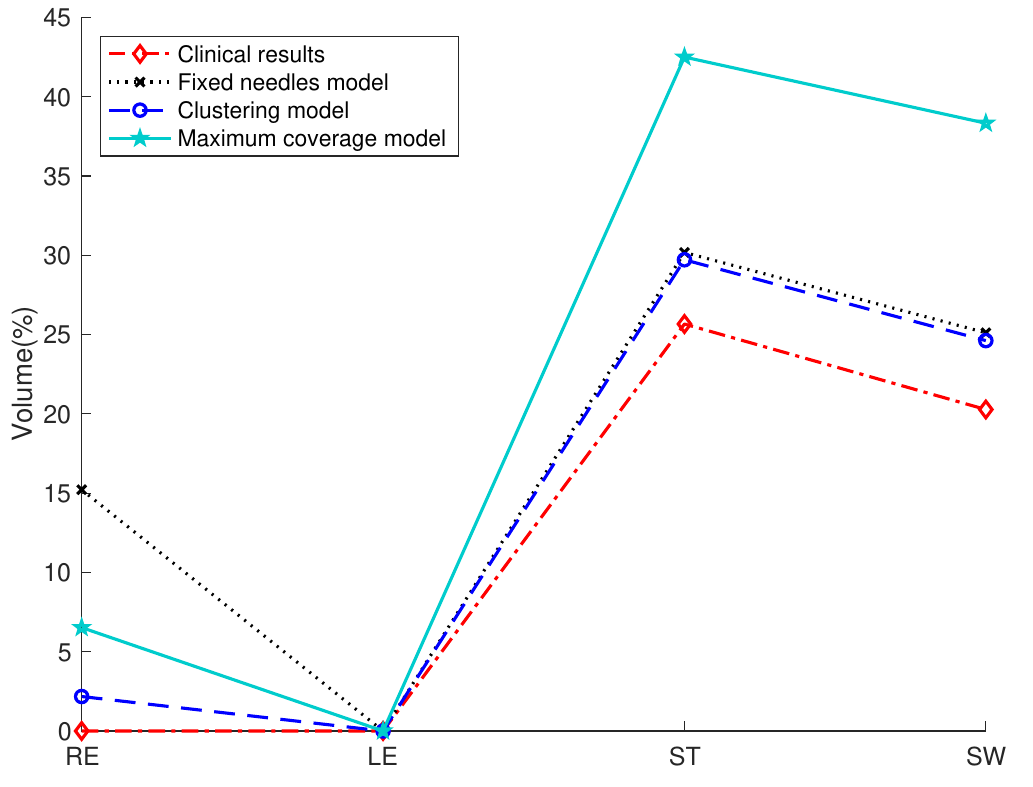}
\captionsetup{justification=centering,font=scriptsize}
\caption{$V_{50}$}
\label{figure:result:OARv50}
\end{subfigure}
\hfill
\begin{subfigure}[b]{0.32\textwidth}
\centering
\includegraphics[width=\textwidth]{ 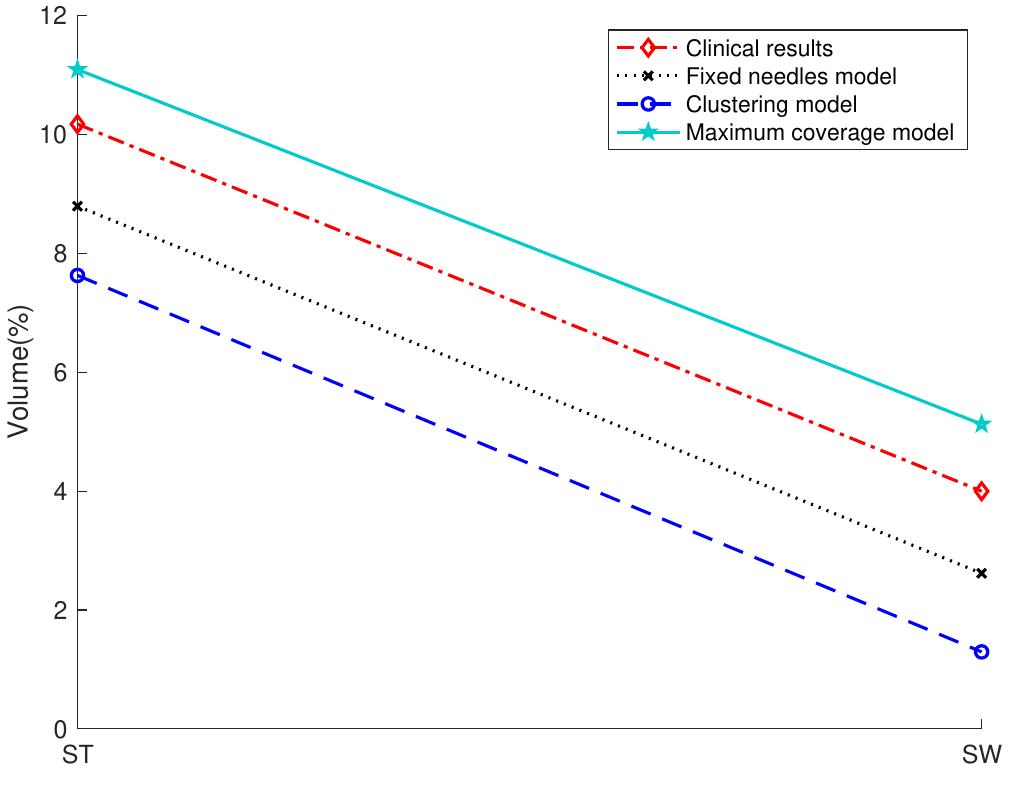}
\captionsetup{justification=centering,font=scriptsize}
\caption{$V_{100}$}
\label{figure:result:OARv100}
\end{subfigure}
\hfill
\begin{subfigure}[b]{0.32\textwidth}
\centering
%         \hspace{-1.5cm}
\includegraphics[width=\textwidth]{ 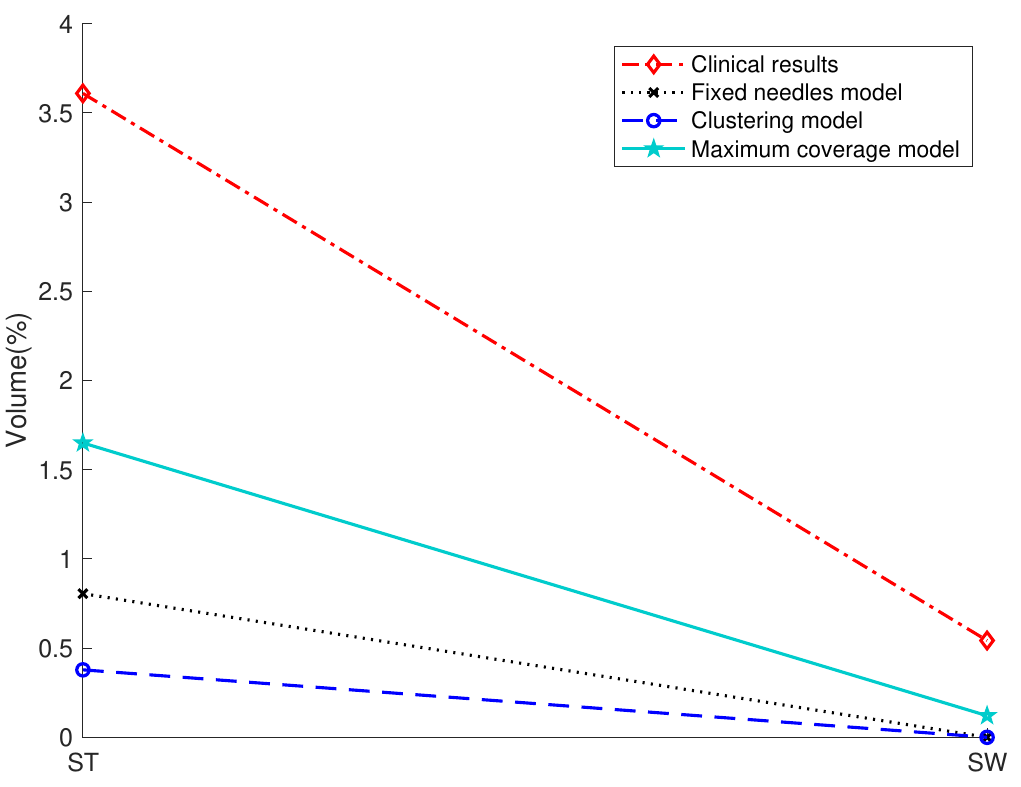}
\captionsetup{justification=centering,font=scriptsize}
\caption{$V_{150}$}
\label{figure:result:OARv150}
\end{subfigure}
\hfill
\begin{subfigure}[b]{0.32\textwidth}
\centering
\includegraphics[width=\textwidth]{ 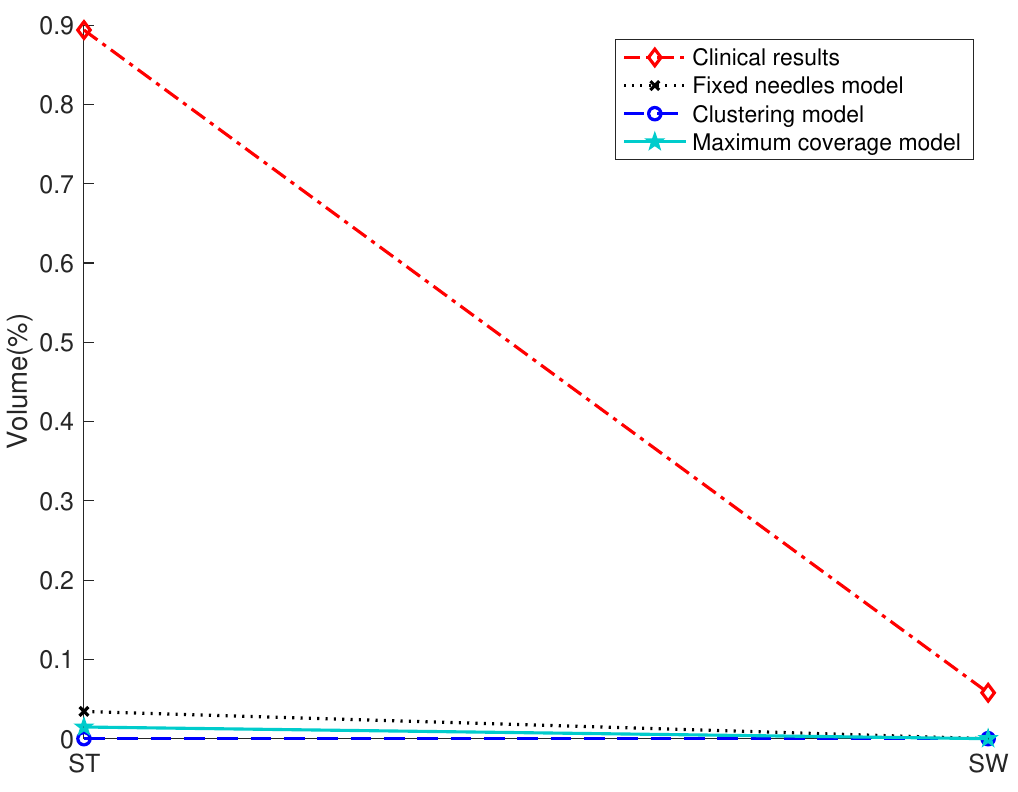}
\captionsetup{justification=centering,font=scriptsize}
\caption{$V_{200}$}
\label{figure:result:OARv200}
\end{subfigure}
\caption{Dosimetric indices for OARs}
\label{figure:result:OARvs}
\end{figure}

These figures show that plans generated with the models we introduce have better characteristics than the clinical plan as they improve most of the dosimetric indices.
This is especially true for $V_{150}$ and $V_{200}$, which are better across all tissues. 
The three proposed models tend to deliver more homogeneous plans, when compared to the clinical plan. 
This can be clearly observed when comparing the DVHs in Figures~\ref{figure:dvh:clinic}, \ref{figure:dvh:fixed}, \ref{figure:dvh:free:coverage} and  \ref{figure:dvh1:free:clustering}. 
%\red{JP: SHOULD WE CITE 16b here SINCE IT IS IN THE APPENDIX?}

When comparing the models introduced in this study, it is apparent that each has its advantages and limitations. 
The fixed needle model is the easiest to implement and can be solved to global optimality.  
However, its performance is average compared to the other methods due to its limited choices of directions.
Additionally, it is not as fast as the other methods in generating treatment plans, which could be problematic if the radiation oncologist requests multiple modifications to the plan.
The second model, which uses free needles and a maximum coverage approach, improves $V_{100}$ for the left side of tumorous tissues and $V_{150}$ and $V_{200}$ for tumorous tissues and OARs compared with the clinical plan. 
The algorithm is faster than the fixed needle approach but slower than the clustering method, as evidenced in Table~\ref{table:solutiontimes:comparison}.
The third model, which uses free needles and a clustering approach, improves $V_{150}$ and $V_{200}$ for tumorous tissues and $V_{100}$, $V_{150}$, and $V_{200}$ for OARs. 
Moreover, this algorithm is the fastest approach the proposed methods, as evidenced in Table~\ref{table:solutiontimes:comparison}.
Finally, it appears that the clinical plan outperforms all of the proposed methods with respect to $V_{50}$ for eye lenses and skin structures. 
However, the $V_{50}$ values for all of the proposed models fall within the range that is considered appropriate for this case.

%%%%%%%%%%%%%%
%%%%%%%%%%%%%%%%%%%%%%%%%%%%
%%%%%%%%%%%%%%%%%%%%%%%%%%%%%%%%%%%%%%%%%%%%%%%%%%%%%%%%
\section{Conclusion}
\label{section:conclusion}
%%%%%%%%%%%%%%%%%%%%%%%%%%%%%%%%%%%%%%%%%%%%%%%%%%%%%%%%
%%%%%%%%%%%%%%%%%%%%%%%%%%%%
%%%%%%%%%%%%%%

We study the problem of positioning channels in 3D-printed masks for HDR-BT, allowing more freedom in needles placement compared to Freiburg flaps.  
We propose three models that show potential in obtaining plans with improved dosimetric indices.

\bibliographystyle{informs2014.bst}
\bibliography{References.bib}

\clearpage

%%%%%%%%%%%%%%
%%%%%%%%%%%%%%%%%%%%%%%%%%%%
%%%%%%%%%%%%%%%%%%%%%%%%%%%%%%%%%%%%%%%%%%%%%%%%%%%%%%%%
\section*{\onlsup}
%%%%%%%%%%%%%%%%%%%%%%%%%%%%%%%%%%%%%%%%%%%%%%%%%%%%%%%%
%%%%%%%%%%%%%%%%%%%%%%%%%%%%
%%%%%%%%%%%%%%

\addcontentsline{toc}{section}{Appendices}
\renewcommand{\thesubsection}{\Alph{subsection}}
\renewcommand{\arraystretch}{0.8}

%%%%%%%%%%%%%%%%%%%%%%%%%%%%
%%%%%%%%%%%%%%%%%%%%%%%%%%%%%%%%%%%%%%%%%%%%%%%%%%%%%%%%
\subsection{Intersection of needles: models and proofs}
\label{section:os:needle-needle}
%%%%%%%%%%%%%%%%%%%%%%%%%%%%%%%%%%%%%%%%%%%%%%%%%%%%%%%%
%%%%%%%%%%%%%%%%%%%%%%%%%%%%

In this section, we present models designed to impose or verify that two needles do not intersect and provide proofs of their validity. 

\Lemmatwofiniteneedles*

\begin{proof}{Proof.}
Imposing that $\line^i$ and $\line^j$ do not intersect is equivalent to requiring that the conic system
\begin{align*}
PR_{\normp,\radius} (\line^i,\line^j)=  
\left\{ 
\begin{array}{l|l}
\vaxi \in \Re^{13} &
\begin{array}{lll}
\vachi - \dir^i \vatau^i - \varho^i = \start^i,  &\, & \vachi - \dir^j \vatau^j - \varho^j = \start^j, \\
-\vatau^i \le -\tmin^i, && -\vatau^j \le \tmin^j,  \\ 
\vatau^i \le \tmax^i, && \vatau^j \le \tmax^j, \\
\vasigma^i =\radius, && \vasigma^j =\radius, \\
(\varho^i,\vasigma^i) \in C_{\normp}, && (\varho^j,\vasigma^j) \in C_{\normp} \\
\end{array} \\
\end{array} 
\right\}
\end{align*}
is infeasible, where $\vaxi$ is used as a shorthand for the vector $(\vachi,\vatau^i,\varho^i,\vasigma^i,\vatau^j,\varho^j,\vasigma^j)$.
The dual conic system of $PR_{\normp,\radius} (\line^i,\line^j)$ is the set
\begin{align*}
DR_{\normp,\radius} (\line^i,\line^j)=  
\left\{ 
\begin{array}{l|l}
\bar{\dualvar} \in \Re^{12} & 
\begin{array}{ll}
\multicolumn{2}{l}{
(\start^i)\tr \dualvarzer  + (\start^j)\tr \dualvarsev   - \tmin^i \dualvarone  + \tmax^i \dualvartwo} \\
\multicolumn{2}{l}{ \qquad\qquad - \tmin^j \dualvarfou + \tmax^j \dualvarfiv  + \radius \dualvarthr  + \radius \dualvarsix < 0,} \\
 \dualvarzer + \dualvarsev = 0, \\
 -(\dir^i)\tr \dualvarzer - \dualvarone + \dualvartwo = 0, & 
 -(\dir^j)\tr \dualvarsev - \dualvarfou + \dualvarfiv = 0, \\
 (-\dualvarzer,\dualvarthr) \in C_{\normp}^*, &
 (-\dualvarsev,\dualvarsix) \in C_{\normp}^*, \\
 \dualvarone \ge 0, \,\, \dualvartwo \ge 0, \,\, \dualvarthr \ge 0, & \dualvarfou \ge 0, \,\, \dualvarfiv \ge 0,  \,\, \dualvarsix \ge 0
 \end{array}
\end{array} 
\right\}
\end{align*}
where $\bar{\dualvar}$ is used as a shorthand for the vector $(\dualvarzer,\dualvarone,\dualvartwo,\dualvarthr,\dualvarfou,\dualvarfiv,\dualvarsix,\dualvarsev)$. 
 
In the above formulation, for $k=i$ (resp. $j$), $\dualvarzer$ (resp. $\dualvarsev$) $ \in \Re^3$ is the dual variable on constraint $\vachi-\dir^k \vatau^k - \varho^k=\start^k$, $\dualvarone$ (resp. $\dualvarfou$) $\in \Re$ is the dual variable on constraint $-\vatau^k \le -\tmin^k$, $\dualvartwo$ (resp. $\dualvarfiv$) $\in \Re$ is the dual variable on constraint $\vatau^k \le \tmax^k$, and $\dualvarthr$ (resp. $\dualvarsix$) $\in \Re$ is the dual variable on constraint $\vasigma^k = \radius$. 
It follows from Proposition 2.4.2 in \citep{ben2001lectures} that if $DR_{\normp,\radius} (\line^i,\line^j)$ is feasible 
then the system $PR_{\normp,\radius} (\line^i,\line^j)$ is infeasible. 
Next, we observe that the constraints of $DR_{\normp,\radius} (\line^i,\line^j)$ define a cone as any feasible solution $\bar{\dualvar}$ yields another feasible solution when scaled by a positive factor. 
Therefore, adding the normalization constraint $||\dualvarzer||_{1}+||\dualvarsev||_{1}+|\dualvarone|+|\dualvartwo|+|\dualvarthr|+|\dualvarfou|+|\dualvarfiv|+|\dualvarsix|\leq 1$ keeps the system feasible.
The absolute values in the above constraint can be removed since they apply to nonnegative variables. 
Next, we use the equality $\dualvarzer+\dualvarsev=0$ to eliminate variable $\dualvarsev$. 
To create a linear representation of $||\dualvarzer||_1$, we introduce variable $\dualvarzerp \in \Re^3$ and constraints $\dualvarzerp \ge \dualvarzer$ and $\dualvarzerp \ge -\dualvarzer$.  
We then obtain the result after observing that if \eqref{sys:twofiniteneedles} has a feasible solution where $(\start^i)\tr \dualvarzerp  - (\start^j)\tr \dualvarzerp  -\tmin^{i} \dualvarone + \tmax^{i}\dualvartwo -\tmin^j \dualvarfou+\tmax^j \dualvarfiv+ \radius\dualvarthr+ \radius \dualvarsix \leq -\epsilon$, then this solution clearly satisfies the strict inequality of $DR_{\normp,\radius} 
(\line^i,\line^j)$. 
It follows that $PR_{\normp,\radius} (\line^i,\line^j)$ is infeasible whenever \eqref{sys:twofiniteneedles} is feasible for a positive $\epsilon$.
\qedjp
\end{proof}

%\subsubsection{Explicit description of Lemma~\ref{lemma:twoinfiniteneedles}} \label{appx:infiniteneedles}

%\begin{lemma}\label{lemma:twoinfiniteneedles}
%Let $(\start^i,\dir^i)$ and $(\start^j,\dir^j)$ be the parameters of two infinite needles $\line^i$ and $\line^j$. 
%Also let $\normp \ge 1$ and $\radius>0$.
%System \eqref{sys:twofiniteneedles} 
%in variables $\dualvar=(\dualvarzer;\dualvarzerp;\dualvarthr,\dualvarsix) \in \Re^3 \times \Re^3 \times \Re^2$ (where $\dualvarone$, $\dualvartwo$, $\dualvarfou$, and $\dualvarfiv$ are removed), which we denote by 
%$\Systwoneedles^{\infty}_{\normp,\radius}\left[\cdot\right]$,  ensures that needles $\line^i$ and $\line^j$ do not intersect when $\epsilon$ is chosen positive. 
%\end{lemma}

In stating Proposition~\ref{proposition:twofiniteneedles}, we 
did not use the system $DR_{\normp,\radius} (\line^i,\line^j)$ directly but instead  eliminated some of its variables $\bar{\dualvar}$ and bounded the remaining variables $\dualvar$ through a normalization constraint. 
Eliminating variables has the advantage of reducing the number of nonconvexities that arise in models that use the system. 
Bounding variables has computational advantages when the formulation is incorporated in models that are solved through spatial branch-and-bound. 
We choose this particular normalization because it can be represented through linear inequalities.
Other normalizations are possible.
In particular, the coefficient of $2$ in front of $||\dualvarzer||_1$ could be chosen equal to $1$, although such a change appeared to yield longer computational times in our pilot computations. 

When needles are infinite, the dual variables $\dualvarone$, $\dualvartwo$, $\dualvarfou$, and $\dualvarfiv$ associated with the lower and upper bounds on the step sizes of the needles do not need to be introduced, leading to  
\begin{corollary}\label{lemma:twoinfiniteneedles}
Let $(\start^i,\dir^i)$ and $(\start^j,\dir^j)$ be the parameters of two infinite needles $\line^i$ and $\line^j$. 
Also let $\normp \ge 1$ and $\radius>0$. 
The system
\begin{align}
\Systwoneedles^{\infty}_{\normp,\radius} \left[\begin{array}{c}\start^i,\dir^i,\cdot,\cdot\\ \start^j,\dir^j,\cdot, \cdot \\ \dualvar \end{array}\right]:    
\left\{ 
\begin{array}{ll}
\multicolumn{2}{l}{
(\start^i-\start^j)\tr \dualvarzer   
+ \radius \dualvarthr  + \radius \dualvarsix \le -\epsilon,} \\
\multicolumn{2}{l}{
2\, \one\tr\dualvarzerp+\dualvarthr+\dualvarsix\leq 1,} \\
-\dualvarzer \le \dualvarzerp, & \dualvarzer \le \dualvarzerp, \\
 -(\dir^i)\tr \dualvarzer  = 0, & 
 (\dir^j)\tr \dualvarzer  = 0, \\
 (-\dualvarzer,\dualvarthr) \in C_{\normp}^*, &
 (\dualvarzer,\dualvarsix) \in C_{\normp}^*, \\
  \dualvarthr \ge 0, &  \dualvarsix \ge 0,
 \end{array}
\right.
 \label{sys:twoinfiniteneedles}
\end{align}
in variables $\dualvar=(\dualvarzer;\dualvarzerp;\dualvarthr,\dualvarsix) \in \Re^3 \times \Re^3 \times \Re^2$  
ensures that needles $\line^i$ and $\line^j$ do not intersect, when $\epsilon$ is chosen positive.
\end{corollary}

In the arguments of $\Systwoneedles^{\infty}_{\normp,\radius}\left[\cdot\right]$ in Corollary~\ref{lemma:twoinfiniteneedles}, we use $\cdot$ to identify input parameters that are not used in this formulation but might be used in others. 
This allows us to adopt consistent notations for Proposition~\ref{proposition:twofiniteneedles} and Corollary~\ref{lemma:twoinfiniteneedles}.
Next, we describe a different system that does not require the addition of dual variables and ensures that the cores of two infinite needles are spaced by a distance of at least $2\radius$. 

%\subsubsection{Proof of Lemma~\ref{minimumdistanceofcores}}
%\label{proofoflemma2}

\Lemmaminimumdistanceofcores*

\begin{proof}{Proof.}
It is clear that the euclidean distance between the cores of two infinite skew needles $\line^i$ and $\line^j$ is at least $2\radius$ if 
and only if $PR^{\infty}_{2,\circle} (\line^i,\line^j)$ (defined in the proof of Proposition~\ref{proposition:twofiniteneedles}) is empty for all $\circle < \radius$. 
As argued earlier, the latter requirement is implied if $DR^{\infty}_{2,\circle} (\line^i,\line^j)$ is feasible for all 
$\circle < \radius$. 
The constraints of $DR^{\infty}_{2,\circle} (\line^i,\line^j)$ are
\begin{align*}
&(\start^i-\start^j)\tr \dualvarzer + \circle (\dualvarthr  + \dualvarsix) < 0,  \\
& (\dir^i)\tr \dualvarzer  = 0, \quad
 (\dir^j)\tr \dualvarzer  = 0 ,\\
& ||\dualvarzer||_2 \le \dualvarthr, \quad 
 ||\dualvarzer||_2 \le \dualvarsix, \\
&  \dualvarthr \ge 0,  \quad 
\dualvarsix \ge 0.
\end{align*}
This system is feasible if and only if
\begin{align*}
& (\start^i-\start^j)\tr \dualvarzer + 2\circle ||\dualvarzer||_2  < 0,  \\
& (\dir^i)\tr \dualvarzer  = 0, \quad
 (\dir^j)\tr \dualvarzer  = 0. 
\end{align*}
The equalities
$(\dir^i)\tr \dualvarzer  = 0$ and $(\dir^j)\tr \dualvarzer  = 0$
can be written in matrix form as  $\bar{M} \dualvarzer = 0$ where
$\bar{M} = \left[ \begin{array}{rrr} \dir^i_1 & \dir^i_2 & \dir^i_3 \\ \dir^j_1 & \dir^j_2 & \dir^j_3 \end{array} \right]$.
Since $\line^i$ and $\line^j$ are assumed to be skew,
$\bar{M}$ has rank $2$.
Thus the solutions to the system $\bar{M} \dualvarzer =0$ are multiples of $\dualvarzerbar$ $= \left((\dir^i_2 \dir^j_3-\dir^i_3 \dir^j_2),(\dir^i_3 \dir^j_1-\dir^i_1 \dir^j_3),(\dir^i_1 \dir^j_2-\dir^i_2 \dir^j_1)\right)$ $=\left(-\det(M_1),\det(M_2),-\det(M_3)\right)$ where 
$M_1 := \left[ \begin{array}{rr} \dir^i_2 & \dir^i_3 \\ \dir^j_2 & \dir^j_3 \end{array} \right]$, 
$M_2 := \left[ \begin{array}{rr} \dir^i_1 & \dir^i_3 \\ \dir^j_1 & \dir^j_3 \end{array} \right]$, 
and
$M_3 := \left[ \begin{array}{rr} \dir^i_1 & \dir^i_2 \\ \dir^j_1 & \dir^j_2 \end{array} \right]$. 
Showing that $DR^{\infty}_{2,\circle} (\line^i,\line^j)$ is feasible therefore reduces to show that at least one of the multiples of $\dualvarzerbar$ satisfies the constraint $(\start^i-\start^j)\tr \dualvarzer + 2\circle ||\dualvarzer||_2  < 0$.
Since this constraint is unchanged when vector $\dualvarzer$ is scaled, we conclude that the set $DR^{\infty}_{2,\circle} (\line^i,\line^j)$ is feasible if and only if $(\start^i-\start^j)\tr \dualvarzerbar + 2\circle ||\dualvarzerbar||_2  < 0$. 
This latter condition can be written as
\begin{align*}
2\circle \sqrt{\det(M_1)^2+\det(M_2)^2+\det(M_3)^2} < \sum_{k=1}^3 (-1)^k (\start^i_k-\start^j_k)  \det(M_k) := \det(M[\start^i,\dir^i,\start^j,\dir^j]).  
\end{align*}
where the last equality follows from the definition of $M[\start^i,\dir^i,\start^j,\dir^j]$ and
Laplace formula for determinants. 
%in the statement of the result. 
Since this strict inequality must hold for all $\circle < \radius$, we conclude that \eqref{distanceskew} holds. 
\qedjp
\end{proof}

System~\eqref{distanceskew} is not applicable when needles are not skew since for parallel needles -- even those that do intersect -- both its left- and right-hand sides are zero.
For parallel neeldes, we have
%, system \eqref{distanceskew} can be adapted as follows. 

%\subsection{Refined version of Lemma~\ref{minimumdistanceofcores} for parallel needles }
%\label{proofoflemma3}

\begin{lemma}\label{parneedles}
Let $(\start^i,\dir)$ and $(\start^j,\dir)$ be the parameters of two parallel needles $\line^i$ and $\line^j$. 
The euclidean distance between the cores of needles $\line^i$ and $\line^j$ is at least $2\radius$ if and only if their parameters satisfy the following constraint 
\begin{align}
\Systwoneedles^{=}_{2,\radius} \left[\begin{array}{c}\start^i,\dir,\cdot,\cdot\\ \start^j,\dir,\cdot,\cdot \\ \cdot \end{array}\right]:    
\left\{ 
\begin{array}{ll}
 \\
\sqrt{\mynorm{\start^i-\start^j}_2^2 \, \mynorm{\dir}_2^2 - \left((\start^i-\start^j)\tr \dir\right)^2} \ge 2\radius \mynorm{\dir}_2.  \label{distancepar}
\\ \textrm{ } 
\end{array}
\right.
\end{align}
\end{lemma}

\begin{proof}{Proof.} Consider a pair of infinite parallel lines,  $\line^i$ and $\line^j$, with parameters $(\start^i,\dir)$ and $(\start^j,\dir)$, respectively.
The shortest distance between these two lines can be computed by solving
$D^2_{\min}=\min_{(\vatau^i,\vatau^j) \in \Re^2}  \mynorm{\start^i-\start^j+ (\vatau^i-\vatau^j)\dir}^2_2$. 
Posing $A= \mynorm{\dir}_2^2$, $B= 2(\start^i-\start^j) \tr \dir$, $C=\mynorm{\start^i-\start^j }_2^2$, and $\vatau=\vatau^i-\vatau^j$, we can rewrite this problem as $D_{\min}^2=\min_{\vatau \in \Re} A \vatau^2+B\vatau+C.$
Since $A>0$, $\vatau^*=\frac{-B}{2A}$ is optimal and  $D^2_{\min}=\frac{-(\frac{B}{2})^2+CA}{A}$. 
Substituting $A$, $B$, and $C$ with their values, the constraint that $D_{\min} \ge 2\radius$ reduces to \eqref{distancepar}. 
\qedjp
\end{proof}

%%%%%%%%%%%%%%%%%%%%%%%%%%%%
%%%%%%%%%%%%%%%%%%%%%%%%%%%%%%%%%%%%%%%%%%%%%%%%%%%%%%%%
\subsection{Intersection of needles: models and proofs}
\label{section:os:needle-polytope}
%%%%%%%%%%%%%%%%%%%%%%%%%%%%%%%%%%%%%%%%%%%%%%%%%%%%%%%%
%%%%%%%%%%%%%%%%%%%%%%%%%%%%

In this section, we present models to impose or verify that a needle and a polytope do not intersect and provide proofs of their validity. 

\Lemmainfiniteneedlepolytope*

\begin{proof}{Proof.}
Needle $\theneedle{\normp}{\radius}(\start,\dir,\tmin,\tmax)$ and polytope $P_{\polylhs,\polyrhs}$ do not intersect if and only if the system 
\begin{align*}
PP_{\normp,\radius} (\line, P_{\polylhs,\polyrhs}) = 
\left\{ 
\begin{array}{l|l}
\vaxi \in \Re^{13} &
\begin{array}{l}
\polylhs \vachi \leq \polyrhs,	\\
\vachi - \vatau \dir - \varho = \start, \\ 
-\vatau \le -\tmin,    \\
\vatau \le \tmax, \\
 \vasigma= \radius, \\
(\varho,\vasigma) \in C_{\normp},\\
\vachi \in \Re^3, \vatau \in \Re, \vasigma \in \Re_+	
\end{array}
\end{array}
\right\}
\end{align*}
is infeasible, where $\vaxi$ is a shorthand for the vector $(\vaxi,\vatau,\varho,\vasigma)$. 
The dual conic system of this set is 
%$PP_{\normp,\radius} (\line, P_{\polylhs,\polyrhs})$ is the set:
\begin{align*}
DP_{\normp,\radius}(\line,P_{\polylhs,\polyrhs})=  
\left\{ 
\begin{array}{l|l}
\bar{\dualvarw} \in \Re^{n+6} &
\begin{array}{lll}
& \polyrhs\tr \dualvarwone +\start\tr \dualvarwtwo - \tmin \dualvarwthr+ \tmax \dualvarwfou + \radius \dualvarwfiv <0, \\
 & \polylhs\tr \dualvarwone + \dualvarwtwo=0, \\
&- \dir\tr \dualvarwtwo - \dualvarwthr+ \dualvarwfou=0, \\
&(-\dualvarwtwo,\dualvarwfiv) \in C_{\normp}^*, \\
& \dualvarwone \in \Re_{+}^{n}, \dualvarwtwo\in \Re^3, \dualvarwthr, \dualvarwfou, \dualvarwfiv\in \Re_{+} 
\end{array}
\end{array}\right\}
\end{align*}
where $\bar{\dualvarw}=(\dualvarwone,\dualvarwtwo,\dualvarwthr,\dualvarwfou,\dualvarwfiv)$ are the dual variables corresponding to the constraints of $PP_{\normp,\radius} (\line, P_{\polylhs,\polyrhs})$ in the order they appear in.
It follows from Proposition 2.4.2 in \citep{ben2001lectures} that if $DP_{\normp,\radius}(\line,P_{\polylhs,\polyrhs})$ is feasible then the system $PP_{\normp,\radius} (\line, P_{\polylhs,\polyrhs})$ is infeasible. 

Next we observe that $DP_{\normp,\radius}(\line,P_{\polylhs,\polyrhs})$ is a cone as any feasible solution $\bar{\dualvarw}$ yields another feasible solution when scaled by a positive factor. 
Thus, adding the normalization constraint $||\dualvarwone||_{1}+||\dualvarwtwo||_{1}+|\dualvarwthr|+|\dualvarwfou|+|\dualvarwfiv|\leq 1$ keeps the system feasible. 
The absolute values in this constraint can be removed since they apply to nonnegative variables. 
We obtain the result after observing that if \eqref{sys:needlepolytope} has a feasible solution where 
$\polyrhs\tr \dualvarwone+\start\tr \dualvarwtwo - \tmin \dualvarwthr+ \tmax \dualvarwfou+ \radius \dualvarwfiv \leq -\epsilon$, then this solution satisfies the strict inequality of $DP_{\normp,\radius}(\line,P_{\polylhs,\polyrhs})$, implying that
%It follows that 
$PP_{\normp,\radius} (\line, P_{\polylhs,\polyrhs})$ is infeasible.
%whenever \eqref{sys:needlepolytope} is feasible for a positive $\epsilon$. 
\qedjp 
\end{proof}

When the needle is infinite, dual variables $\dualvarwthr$ and $\dualvarwfou$ do not need to be introduced, yielding
\begin{corollary}\label{corollary6.2.1}
Let $P_{\polylhs,\polyrhs}=\{x \in \Re^3 \,|\, \polylhs x \le \polyrhs \}$ be a polyope and let $\line$ be an infinite needle with parameters $(\start,\dir)$. Also let $\normp \ge 1$ and let $\radius > 0$. 
The system
\begin{align}
\SysneedleP^{\infty}_{\normp,\radius} \left[ \begin{array}{c} \start,\dir,\cdot,\cdot \\ \polylhs, \polyrhs \\ \dualvarw \end{array}\right]:
\left\{ 
\begin{array}{l}
\polyrhs\tr \dualvarwone +\start\tr \dualvarwtwo+ \radius \dualvarwfiv \leq -\epsilon, \\
\polylhs\tr \dualvarwone+ \dualvarwtwo=0, \\ 
-\dir\tr \dualvarwtwo =0,\\
\one\tr \dualvarwone+ \one\tr\dualvarwtwop +\dualvarwfiv\leq 1,\\
\,\,\,\, \dualvarwtwo \le \dualvarwtwop, \\
-\dualvarwtwo \le -\dualvarwtwop, \\
(-\dualvarwtwo,\dualvarwfiv) \in C^*_{\normp}, \\
 \dualvarwone \in \Re_{+}^{n}, \dualvarwtwo\in \Re^3, \dualvarwtwop \in \Re^3, \dualvarwfiv\in \Re_{+}, 
\end{array}
\right.
\end{align}
in variables $\dualvarw=(\dualvarwone,\dualvarwtwo,\dualvarwtwop,\dualvarwfiv) \in \Re^{n+7}$
ensures needle $\line$ and polytope $P_{\polylhs,\polyrhs}$ do not intersect when $\epsilon$ is chosen to be positive. 
\end{corollary}

When an infinite reduces to its core, dual variable $\dualvarwfiv$ can further be eliminted to obtain
%\subsection{Supplementary Material for Lemma~\ref{lemma2.6}}\label{col2}
\begin{corollary}\label{cor:needlecorepoly}
Let $P_{\polylhs,\polyrhs}=\{x \in \Re^3 \,|\, \polylhs x \le \polyrhs \}$ be a polyope and let $\line$ be an infinite needle with parameters $(\start,\dir)$.  
The system 
\begin{align}
\SysneedleP^{\times}_{\normp,\radius} \left[ \begin{array}{c} \start,\dir,\cdot,\cdot \\ \polylhs, \polyrhs \\ \dualvarw \end{array}\right]:
\left\{ 
\begin{array}{l}
\polyrhs\tr \dualvarwone +\start\tr \dualvarwtwo \leq -\epsilon, \\
\polylhs\tr \dualvarwone+ \dualvarwtwo=0, \\ 
-\dir\tr \dualvarwtwo =0,\\
\one\tr \dualvarwone+ \one\tr\dualvarwtwop \leq 1, \\
\,\,\,\, \dualvarwtwo \le \dualvarwtwop, \\
-\dualvarwtwo \le -\dualvarwtwop, \\
 \dualvarwone \in \Re_{+}^{n}, \dualvarwtwo\in \Re^3, \dualvarwtwop \in \Re^3, 
\end{array}
\right.
\end{align}
in variables $\dualvarw=(\dualvarwone,\dualvarwtwo,\dualvarwtwop) \in \Re^{n+6}$
ensures that the core of needle $\line$ and polytope $P_{\polylhs,\polyrhs}$ do not intersect when $\epsilon$ is chosen to be positive. 
\end{corollary}

Corollary~\ref{cor:needlecorepoly} can be established using linear programming duality results.
Thus, it could be expressed as a necessary and sufficient condition for the core of $\line$ to not intersect with $P_{\polylhs,\polyrhs}$.

\subsection{Solution algorithm for clustering approach}
\label{section:os:clusteringalgo}
%%%%%%%%%%%%%%%%%%%%%%%%%%%%%%%%%%%%%%%%%%%%%%%%%%%%%%%%
%%%%%%%%%%%%%%%%%%%%%%%%%%%%

In this section, we give details of the algorithm presented in Section~\ref{section:modelsmethods:free:phase2} for the solution of the clustering model 
$\Mmdl$.
A pseudo-code is given in Algorithm~\ref{alg:jp}. 
Because this algorithm is heuristic in nature, we will use a multi-start strategy in which we run the heuristic multiple times ($l=1,\ldots,L$) from different starting solutions and keep the solution with the smallest objective value among those obtained.
For each restart, we obtain a solution by iterating between the two steps presented next, until convergence is achieved. 

%%%%%%%%%%%%%%%%%%%%%%%%%%%%
\subsubsection{Step 1: Assigning dwell positions to needles.}
\label{section:os:clusteringalgo:step1}
%%%%%%%%%%%%%%%%%%%%%%%%%%%%

The first step takes as inputs a collection of needles parameters $(\fix{\start}^j,\fix{\dir}^j,\fix{\tmin}^j,\fix{\tmax}^j)$ for $j \in \needles$ and will determine clusters $\points_k$ of dwell points to assign to each needle $k$. 
This is done by solving the following network flow problem
\begin{subequations}
\label{orgnetworkbasedmodel1}
\begin{alignat}{4} 
 && \min \quad & \sum_{ i\in \points}\sum_{ j\in \needles}{\Dis_{i,j}^2} \asgt_{i,j}\label{eq:fixmodelkmeansobj}\\
 (\Mdlassg)\qquad &\quad& \text{s.t.} \quad & \sum_{j\in \needles} \asgt_{i,j}=1, && \forall i\in \points, \label{eq:fixmodelkmeanscon1}\\
 &&& \sum_{i\in \points} \asgt_{i,j}\geq 2, &&\forall j\in \needles, \label{eq:fixmodelkmeanscon11}  
\end{alignat}
\end{subequations}
with the requirement that $\asgt_{i,j} \in \{0,1\}$ and where $\Dis_{i,j}^2= ||\dwpt^i-\fix{\start}^j-\lambda^* \fix{\dir}^j||_2^2$, 
$\lambda^*=\max\{\fix{\tmin}^j,\min\{\fix{\tmax}^j,\vatau^* \}\}$, and $\vatau^*=\frac{(\dwpt^i-\start^j) \tr \dir^j}{\mynorm{\dir^j}_2^2}$ for $i\in \points$ and $j\in \needles$.
Given an optimal solution $\asgt{}^*$ of Model~\eqref{orgnetworkbasedmodel1}, we define $\points_k = \{ i \in \points \,|\, \asgt_{i,k}^* = 1 \}$ for $k \in \needles$.

%%%%%%%%%%%%%%%%%%%%%%%%%%%%
\subsubsection{Step 2: Finding needle directions.} 
\label{section:os:clusteringalgo:step2}
%%%%%%%%%%%%%%%%%%%%%%%%%%%%
The second steps takes as input a collection of dwell points clusters $\points_k$ corresponding to each of the needle $k \in \needles$. 
It then finds a solution to the restriction of $\Mmdl$ where variables $\asgt{}^*$ are fixed.
Thus, it determines parameters $(\fix{\dir}^{j},\fix{\start}^{j},\fix{\tmin}^j,\fix{\tmax}^j)$ for all needles $j \in \needles$. 
It does so by finding the characteristics of each needle one at the time, assuming that the characteristics of needles considered previously are now immutable. 
At iteration $k$, it finds the parameter of needle $k$ $(\dir^{k},\start^{k},\tmin^k,\tmax^k)$ 
to best fit the cluster of dwell positions $\points_k$ 
assuming that the parameters 
of every one of the $k-1$ needles $\fix{\line}^j$ have already been selected to be 
$(\fix{\dir}^{j},\fix{\start}^{j},\fix{\tmin}^j,\fix{\tmax}^j)$ for $j \in \needles_k=\{1,\ldots,k-1\}$.
We refer to this problem as $\Rmdlorig{k}$. 
Because this problem must decide which exiting plane needle $k$ originates from, it can be itself decomposed into separate subproblems, each requiring the reference point of needle $k$ to be belong to a different exiting plane $\dn_e$.
This yields the following model 
\begin{subequations}
\label{mpr}
\begin{alignat}{3}
&\min \quad && \sum_{ i\in \points_{k}}{||\dwpt^{i}-\start^k-\stept_{i,k}\dir^{k}||^2}\label{eq:morefixclustermodelkmeansobj1}\\ &\text{s.t.} && ||\dir^k||^2=1,\label{morefixclustermodelcon661} \\
(\Rmdl{k}) \qquad &&&\Systwoneedles^{\mdla}_{\normp,\radius} \left[\begin{array}{c}\fix{\start}^j,\fix{\dir}^j,\fix{\tmin}^j,\fix{\tmax}^j\\ \start^k,\dir^k,\tmin^k,\tmax^k \\ \dualvar^{k,j} \end{array}\right], &\quad &\forall j \in \needles_k,\label{eq:morefixclustermodelcon21}\\
 &&& \SysneedleP^{\mdlb}_{\normp,\radius} \left[ \begin{array}{c} \start^k,\dir^k,\tmin^k,\tmax^k \\ \polylhs^m, \polyrhs^m \\ \dualvarw^{k,m} \end{array}\right], && \forall m \in \polyhyd,\label{morefixclustermodelcon31}\\
&&& \tmin^k \leq \stept_{i,k}\leq \tmax^k,&& \forall i \in \points_{k},\label{morefixclustermodelcon721} \\
&&&  \start^k\in \dn_e. \label{morefixclustermodelcon7771} 
\end{alignat}
\end{subequations}
with the requirements that 
($i$) $\dir^k\in \Re^3$
($ii$) $\tmin^k,\tmax^k\in \Re$
($iii$) $\dualvar^{k,j} \in \Re^{\nv{\mdla}}, \forall j \in \needles_k$, and
($iv$) $\dualvarw^{k,m} \in \Re^{\nw{\mdlb}}, \forall m \in \polyhyd$.
where $(\mdla,\mdlb)\in \{(\leftrightarrow,\leftrightarrow), (\infty,\infty), (\times,\infty),  (=,\infty), (\times,\times),  (=,\times)\}$
 and $\normp \in \{1,2,\infty\}$.
 In particular, we will have that $\tmin^k=-\infty, \tmax^k=+\infty$  when $(\mdla,\mdlb)\in \{(\infty,\infty), (\times,\infty),  (=,\infty), (\times,\times),  (=,\times)\}$, in which case Constraints \eqref{morefixclustermodelcon721} %and \eqref{morefixclustermodelcon7771} 
 are not relevant and can be discarded.
 Constraints \eqref{eq:morefixclustermodelkmeansobj1}-\eqref{morefixclustermodelcon721} have the same meaning as in \Mmdl.  
Constraint~\eqref{morefixclustermodelcon7771} enforces that the starting point of the needle is in exiting plane $\dn_e$.
This requirement has the advantage of fixing one of the components of $\start^j$ and to provide upper and lower bounds for the remaining components. 
It also allows the fixing of $\tmin^k$ to $0$. 
It is known that deriving better bounds on variables occurring in bilinear terms is crucial for producing relaxation bounds in branch-and-bound codes. 

Even though Model~\eqref{mpr} only optimizes the parameters of a single needle, it remains difficult to solve. 
In our numerical experience, \gurobi\xspace takes a large amount of time to generate good quality feasible solutions and rarely is able to prove optimality when $k$ is large.
Hence, we propose to obtain solutions for \eqref{mpr} using an alternating heuristic procedure.
%{\color{red} We solve this model over all $e \in \en$ to determine the best exiting plane to select for $\start^k$. }
%\byjp{Talk about looping over the planes}
This heuristic is composed of three steps, each one corresponding to the solution of a different model -- $\RmdlI{k}$, $\RmdlII{k}$, and $\RmdlIII{k}$ -- that optimizes a subset of variables $(\start^k,\dir^k,\tmin^k,\tmax^k)$ while the others are fixed. 
The solution of these three models is repeated in sequence until convergence is achieved. 
The process is initialized with a feasible solution obtained by using a commercial solver on a variant of $\Mmdl$. 
In this variant, the objective function is replaced with $0$ and the direction of needle $k$ is fixed to a predetermined value $\init{\dir}{}^{k}$. 
This procedure is fast in practice.

In Model~$\RmdlI{k}$, we fix the reference point of needle $k$ to $\fix{\start}^k$, minimum and maximum step sizes $\fix{\tmin}^k$ and $\fix{\tmax}^k$, and step size to projection of point $i$ to be $\fix{\stept}_{i,k}$.
We then optimize over the needle direction $\dir^k$. 
Specifically, we solve 
the model composed of 
(\ref{eq:morefixclustermodelkmeansobj1},
\ref{morefixclustermodelcon661},
\ref{eq:morefixclustermodelcon21},
\ref{morefixclustermodelcon31}), which
is a nonconvex optimization model.

In Model~$(\RmdlII{k})$, we fix the direction $\dir^k$ to $\fix{\dir}^k $, the minimum and maximum step sizes to $\fix{\tmin}^k$ and $\fix{\tmax}^k$, and require the needle to pass through exiting plane $\dn_e$.
We then optimize the reference point $\start^k$ together with step size to projection of point $i$, $\stept_{i,k}$.
Specifically, we solve the model composed of 
(\ref{eq:morefixclustermodelkmeansobj1},
\ref{eq:morefixclustermodelcon21},
\ref{morefixclustermodelcon31},
\ref{morefixclustermodelcon721},
\ref{morefixclustermodelcon7771}), which 
%Model $(\RmdlII{k})$ 
is a continuous nonconvex quadratic program. 

In Model~$(\RmdlIII{k})$,  We fix $\dir^k$ and $\start^k$ to $\fix{\dir}^k$ and $\fix{\start}^k$, respectively. 
We then optimize the minimum and maximum step sizes, $\tmin^k$, $\tmax^k$ together with step size to projection of point $i$, $\stept_{i,k}$.
Specifically, we solve the model composed of 
(\ref{eq:morefixclustermodelkmeansobj1},
\ref{eq:morefixclustermodelcon21},
\ref{morefixclustermodelcon31},
\ref{morefixclustermodelcon721}).
We have that $\tmin^k=-\infty, \tmax^k=+\infty$  when $(\mdla,\mdlb)\in \{(\infty,\infty), (\times,\infty),  (=,\infty), (\times,\times),  (=,\times)\}$.
Model $\RmdlIII{k}$ is a continuous nonlinear optimization models whose objective function is convex and quadratic. 
If $\Systwoneedles^{\infty}_{\normp,\radius}[\cdot]$ and $\SysneedleP^{\infty}_{\normp,\radius}[\cdot]$ are used, then  $(\RmdlIII{k})$ is an unconstrained convex quadratic program whose solution can be derived in closed-form.

Models $\RmdlI{k}$, $\RmdlII{k}$, and $\RmdlIII{k}$ differ in the number of linear, convex quadratic, bilinear, $\normq$-order cone, and nonconvex quadratic constraints that they have. 
This number is also affected by the models $\mdla$ and $\mdlb$ used for representing needle intersections and intersections with structures of interest. 
Table~\ref{table:modelsize:free} summarizes the size of the different models  
%\subsection{Complexity analysis of models $\RmdlI{k}$, $\RmdlII{k}$, and $\RmdlIII{k}$}
%\label{tableofcomplexity}
$\RmdlI{k}$, $\RmdlII{k}$, and $\RmdlIII{k}$ as a function of $\K$ (the number of needles being considered in the problem, $\K-1$ of them fixed), $|\polyhyd|$ (the number of structures that needle $k$ should avoid) and $|\points_k|$ (the number of dwell positions in cluster $k$).
The constraints of these models are classified as being either linear ($\tt L$), convex quadratic ($\tt CQ$), bilinear ($\tt BL$), \normq-order cone ($\tt qOC$) or nonconvex quadratic ($\tt NQ$).
When a constraint satisfies multiple denominations in the preceding list,  we associate with it the first label it satisfies.
We introduce $\polycons_m$ to denote the total number of constraints in the polyhedron describing structure $m\in \polyhyd$ to avoid and define $\totalpolycons := \sum_{m \in \polyhyd} \polycons_m$.

\newcolumntype{P}[1]{>{\centering\arraybackslash}p{#1}}
\begin{table}[!htb]
\fontsize{10pt}{10pt}
\begin{center}
\begin{tabular}
{|P{0.75in}|c|c|P{1.3in}|P{0.6in}|P{0.2in}|P{0.6in}|P{0.2in}|P{1.5in}|} 
\hline
 &  &  & \multicolumn{5}{|c|}{Number of constraints by type} & Number of variables \\
 \hline
 & $\mdla$ & $\mdlb$ & {\tt L} & {\tt BL} & {\tt CQ} & {\tt qOC} & {\tt NQ} & \\
 \hline 
 & $\leftrightarrow$ & $\leftrightarrow$ & $4\K+5|\polyhyd|$  & $2\K+|\polyhyd|$ & $-$  &$2\K+|\polyhyd|$ & $1$ & $12\K+\totalpolycons+9|\polyhyd|+3$\\ 
  & $\infty$ & $\infty$& $4\K+5|\polyhyd|$  & $2\K+|\polyhyd|$ & $-$  &$2\K+|\polyhyd|$ & $1$ & $8\K+\totalpolycons+7|\polyhyd|+3$\\
$\RmdlI{k}$  & $\times$ & $\infty$ & $5|\polyhyd|$  & $|\polyhyd|$ &\K  &$|\polyhyd|$ & $1$ & $\totalpolycons+7|\polyhyd|+3$\\
  & $=$ & $\infty$ & $5|\polyhyd|$  & $|\polyhyd|$ &\K  &$|\polyhyd|$ & $1$ & $\totalpolycons+7|\polyhyd|+3$\\
  & $\times$ & $\times$& $5|\polyhyd|$  & $|\polyhyd|$ &\K  & $-$ & $1$ & $\totalpolycons+6|\polyhyd|+3$\\
  & $=$ & $\times$ & $5|\polyhyd|$  & $|\polyhyd|$ &\K  & $-$ & $1$ & $\totalpolycons+6|\polyhyd|+3$\\
\hline
& $\mdla$ & $\mdlb$ & {\tt L} & {\tt BL} & {\tt CQ} & {\tt qOC} & {\tt NQ} & \\
 \hline 
 & $\leftrightarrow$ & $\leftrightarrow$ & $5\K+2|\points_k|+5|\polyhyd|+5$  & $\K+|\polyhyd|$ & $-$  &$2\K+|\polyhyd|$ & $-$ & $12\K+\totalpolycons+9|\polyhyd|+3+|\points_k|$\\ 
  & $\infty$ & $\infty$& $5\K+2|\points_k|+5|\polyhyd|+5$  & $\K+|\polyhyd|$ & $-$  &$2\K+|\polyhyd|$ & $-$ & $8\K+\totalpolycons+7|\polyhyd|+3+|\points_k|$\\
$\RmdlII{k}$  & $\times$ & $\infty$ & $2|\points_k|+5|\polyhyd|+5$  & $|\polyhyd|$ &\K  &$|\polyhyd|$ & $-$ & $\totalpolycons+7|\polyhyd|+3+|\points_k|$\\
  & $=$ & $\infty$ & $2|\points_k|+5|\polyhyd|+5$  & $|\polyhyd|$ &\K  &$|\polyhyd|$ & $-$ & $\totalpolycons+7|\polyhyd|+3+|\points_k|$\\
  & $\times$ & $\times$& $2|\points_k|+5|\polyhyd|+5$  & $|\polyhyd|$ &\K  & $-$ & $-$ & $\totalpolycons+6|\polyhyd|+3+|\points_k|$\\
  & $=$ & $\times$ & $2|\points_k|+5|\polyhyd|+5$  & $|\polyhyd|$ &\K  & $-$ & $-$ & $\totalpolycons+6|\polyhyd|+3+|\points_k|$\\
  \hline
 & $\mdla$ & $\mdlb$ & {\tt L} & {\tt BL} & {\tt CQ} & {\tt qOC} & {\tt NQ} & \\
 \hline 
 & $\leftrightarrow$ & $\leftrightarrow$& 
$5\K+2|\points_k|+5|\polyhyd|$  & $\K+|\polyhyd|$ & $-$  &$2\K+|\polyhyd|$ & $-$ & $12\K+\totalpolycons+9|\polyhyd|+2+|\points_k|$\\ 
  & $\infty$ & $\infty$& $5\K+2|\points_k|+6|\polyhyd|$  & $\K$ & $-$  &$2\K+|\polyhyd|$ & $-$ & $8\K+\totalpolycons+7|\polyhyd|+2+|\points_k|$\\
$\RmdlIII{k}$  & $\times$ & $\infty$ & $2|\points_k|+6|\polyhyd|$  & $-$ & $-$  &$|\polyhyd|$ & $-$ & $\totalpolycons+7|\polyhyd|+2+|\points_k|$\\
  & $=$ & $\infty$ & $2|\points_k|+6|\polyhyd|$  & $-$ & $-$  &$|\polyhyd|$ & $-$ & $\totalpolycons+7|\polyhyd|+2+|\points_k|$\\
  & $\times$ & $\times$& $2|\points_k|+6|\polyhyd|$  & $-$ & $-$  & $-$ & $-$ & $\totalpolycons+6|\polyhyd|+2+|\points_k|$\\
  & $=$ & $\times$ & $2|\points_k|+6|\polyhyd|$  & $-$ & $-$  & $-$ & $-$ & $\totalpolycons+6|\polyhyd|+2+|\points_k|$\\
  \hline
\end{tabular}
\caption{Size and types of submodels}
\label{table:modelsize:free}
\end{center}
\end{table}

%%%%%%%%%%%%%%%%%%%%%%%%%%%%
\subsubsection{Pseudo-code.} 
\label{section:os:clusteringalgo:step3}
%%%%%%%%%%%%%%%%%%%%%%%%%%%%

The overall architecture of the heuristic is presented in Algorithm~\ref{alg:jp}.
We see that the heuristic is run $L$ time and that the starting clusters are obtained using an application of k-means from a random initialization.
The selection of needles characteristics is performed in Lines~\ref{algo_line7}-\ref{algo_line20}. 
In particular, the determination of each individual needle is done by solving the needle placement for each exiting plane and choosing among the planes the one leading to the best needle (Line~\ref{algo_line19}). 
For each plane and cluster, the needle is found by running the three models $\RmdlI{}$, $\RmdlII{}$, $\RmdlIII{}$ (Lines~\ref{algo_line13}-\ref{algo_line15}) until convergence is achieved.
Clusters are then updated through the solution of Model (\Mdlassg) in Lines~\ref{algo_line24}-\ref{algo_line27}. 
If clusters have changed, the process is iterated.
Otherwise, run $l$ of the multi-start heuristic is completed and the solution obtained in this run is compared with the best solution encountered so far in Line~\ref{algo_line29}. 
The best solution from all runs is then returned in Line~\ref{algo_line31}.

\begin{algorithm}
\caption{Solution algorithm for $\Mmdl$}
\label{alg:jp}
\begin{algorithmic}[1]
\State $\totoptval^* \gets \infty$, $\epsilon \gets 0.001$
\For{$l=1,\ldots,L$} \Comment{{\tt \footnotesize Perform restarts}}
\State done $\gets$ \no \Comment{{\tt \footnotesize Clusters initializations}}
\State Generate clusters $\points_k$ for $k \in \needles$ using k-means from random startpoint 
\While{(done == \no)} \Comment{{\tt \footnotesize Iterate until convergence}} 
\State $\totoptval_l \gets 0$, done $\gets$ \yes
\For{$k \in \needles$} \Comment{{\tt \footnotesize Find needle characteristics}} \label{algo_line7}
    \For{$e \in \en$}
        \State $\optval_1^{k,e} \gets 0$, $\optval_2^{k,e} \gets 0$, $\optval_3^{k,e} \gets \infty$, $\optval^* \gets \infty$, $e^* \gets 0$
%        \If{($k=1$)}
%            \State Randomly generate $\init{\dir}{}^1 \in \Re^3$
%        \Else
%            \State Set $\init{\dir}{}^k=\fix \dir{}^{k-1}$
%        \EndIf
        \State \textbf{if} ($k=1$) \textbf{then} Randomly generate $\init{\dir}{}^{1,e} \in \Re^3$ \textbf{else} $\init{\dir}{}^{k,e}=\fix \dir{}^{k-1,e}$ \textbf{end if}
        \State Solve $\Mmdl$ with $0$-objective and fixed $\init{\dir}{}^{k,e}$ to obtain $(\init{\start}{}^{k,e^*}, %\init{\dir}{}^{k,e^*},
        \init{\tmin}{}^{k,e^*}, \init{\tmax}{}^{k,e^*},\init{\stept_{i}}{}^{k,e})$
        \While{($\optval_3^{k,e}-\optval_2^{k,e}<-\epsilon$)}
            \State Solve $(\RmdlI{k})$ with optimal value $\optval_1^{k,e}$ and solution $(\init{\dir}{}^{k,e})$ \label{algo_line13}
            \State Solve $(\RmdlII{k})$ with optimal value $\optval_2^{k,e}$ and solution $(\init{\start}{}^{k,e}, \init{\stept_{i}}{}^{k,e})$
            \State Solve $(\RmdlIII{k})$ with optimal value $\optval_3^{k,e}$ and solution $(\init{\tmin}{}^{k,e}, \init{\tmax}{}^{k,e}, \init{\stept_{i}}{}^{k,e})$
            \label{algo_line15}
        \EndWhile
%        \If{($\optval_3^{k,e}<\optval^*$)}
%            \State Set $\optval^* \gets \optval_3^{k,e}$, $e^* \gets e$
%        \EndIf
        \State \textbf{if} ($\optval_3^{k,e}<\optval^*$) \textbf{then} $\optval^* \gets \optval_3^{k,e}$, $e^* \gets e$ \textbf{end if}
    \EndFor
    \State  $(\fix{\start}^{k},\fix{\dir}^{k},\fix{\tmin}^k,\fix{\tmax}^k) \gets (\init{\start}{}^{k,e^*}, \init{\dir}{}^{k,e^*}, \init{\tmin}{}^{k,e^*}, \init{\tmax}{}^{k,e^*})$, $\totoptval_l \gets \totoptval_l + \optval^*$ \label{algo_line19}
\EndFor \label{algo_line20}
\For{($i \in \points$)}
\State $\vatau^*_{i,k} \gets \frac{(\dwpt^i-\start^k) \tr \dir^k}{\mynorm{\dir^k}_2^2}$, $\lambda^*=\max\{\fix{\tmin}^k,\min\{\fix{\tmax}^k,\vatau^*_{i,k} \}\}$, 
$\Dis_{i,k}^2= ||\dwpt^i-\fix{\start}^k-\lambda^* \fix{\dir}^k||_2^2$.
\EndFor
\State Solve (\Mdlassg) with distances $\Dis_{i,k}^2$ to obtain optimal solution $\asgt{}^*$ \label{algo_line24} \Comment{{\tt \footnotesize Update clusters}}
\For{($k \in \needles$)} 
\State $\points'_k \gets \{ i \in \points \,|\, \asgt_{i,k}^* = 1 \}$,  
 \textbf{if} ($\points'_k \neq \points_k$) \textbf{do} done $\gets$ \no\xspace  \textbf{end if} , $\points_k \gets \points'_k$
\EndFor \label{algo_line27}
\EndWhile
%\If{($\totoptval_l<\totoptval^*$)}
%\State $\totoptval^* \gets \totoptval_l$, $(\dir^*,\start^*,\tmin^*,\tmax^*) \gets (\fix{\dir},\fix{\start},\fix{\tmin},\fix{\tmax})$
%\EndIf
\State \textbf{if} ($\totoptval_l<\totoptval^*$)
\textbf{then} $\totoptval^* \gets \totoptval_l$, $(\dir^*,\start^*,\tmin^*,\tmax^*) \gets (\fix{\dir},\fix{\start},\fix{\tmin},\fix{\tmax})$ 
\textbf{end if} \label{algo_line29}
\EndFor
\State \Return $(\dir^*,\start^*,\tmin^*,\tmax^*)$ \label{algo_line31}
\end{algorithmic}
\end{algorithm}

\subsection{Convex hull }
\label{section:os:faceconvexhull}

This section presents in Figure~\ref{figure:convexhull} the polyhedra used to determine or to impose that needles do not intersect with important body structures of the patient. 

\begin{figure}[!htb]
\begin{subfigure}{.49\textwidth}
\centering
\includegraphics[scale=0.4]{ 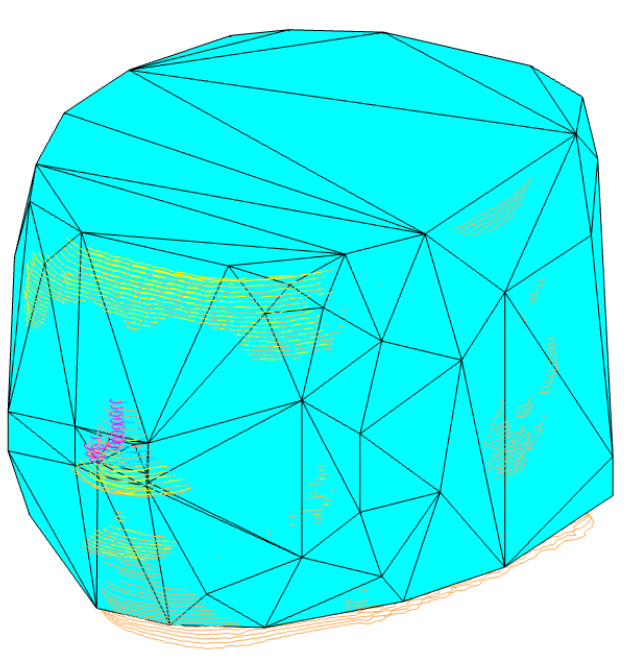}
\captionsetup{justification=centering,font=scriptsize}
\caption{3D representation of the convex hull (with 142 planes) of points representing the head structure of a patient with skin cancer on the nose}
\label{figure:convexhull:142points}
\end{subfigure}
\begin{subfigure}{.49\textwidth}
  \centering   
\vskip -3cm
\includegraphics[scale=0.4]{ 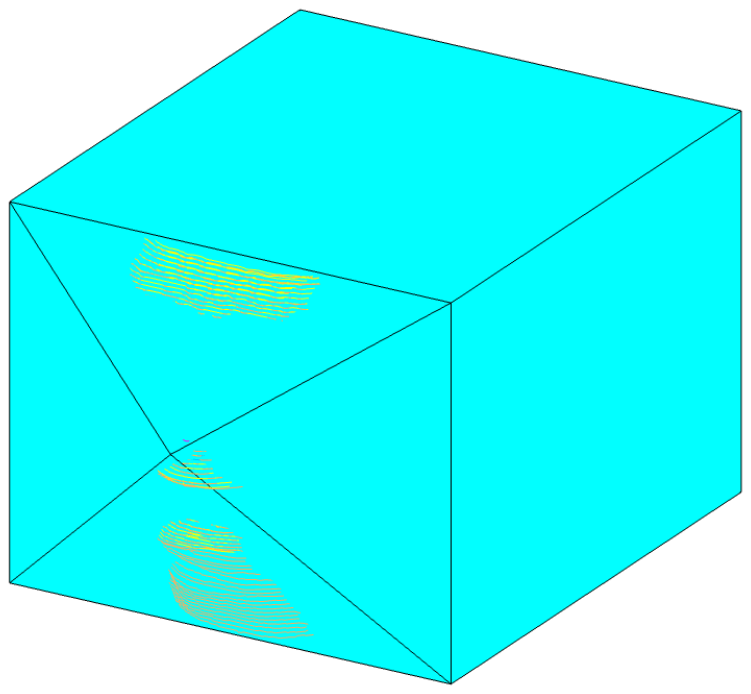}
\vskip -3cm
\captionsetup{justification=centering,font=scriptsize}
\caption{Polytope with seven planes that needles should not intersect}
\label{figure:convexhull:7points}
\end{subfigure}
\caption{Polyhedral approximations of the patient face}
\label{figure:convexhull}
\end{figure}

\subsection{Analyzing variants of clustering approach}
\label{section:os:clusteringvariants}

%\paragraph{Analyzing different variations of clustering approach}
%\label{diffvar}

The clustering approach discussed in Section~\ref{section:modelsmethods:free:phase2} allows for different sets of constraints
$\Systwoneedles^{\mdla}_{\normp,\radius}\left[\cdot\right]$ 
and $\SysneedleP^{\mdlb}_{\normp,\radius} \left[ \cdot \right]$ to be used in $\Mmdl$. 
In this section, we analyze the performance of five variants whose characteristics are described in  Table~\ref{table:characteristics:free:variants}, with respect to needle positions and the following dosimetric indices: ($i$) $V_{100}$, $V_{150}$, and $V_{200}$ for tumor regions ($RV$, $RS$, $LV$, and $LS$) and ($ii$) $V_{50}$, $V_{100}$, $V_{150}$, and $V_{200}$ for OARs ($SW$, $ST$, $RE$, and $ LE$) with only $V_{50}$ computed for $RE$ and $LE$.  
%\begin{itemize}
%\item For , we use .
%We compute these metrics for .
%\item For OARs, we use .
%We compute these metrics for  .
%\end{itemize}

\begin{table}[!htb]
\begin{center}
\begin{tabular}{|c|c|c|c|c|c|c|c} 
 \hline
\backslashbox{Parameter}{Method} & \multicolumn{1}{c|}{Variant-1} &  \multicolumn{1}{c|}{Variant-2}  & \multicolumn{1}{c|}{Variant-3} &  \multicolumn{1}{c|}{Variant-4} & \multicolumn{1}{c|}{Variant-5}  \\
 \hline
$\mdla$	&  $\times/=$	&	 $\times/=$	&	$\times/=$	&	 $\infty$	&  $\infty$	\\	
\hline
$\mdlb$ 	&   $\times$	&	 $\infty$	&	 $\infty$	&	 $\infty$	&  $\infty$\\
\hline
 $\normp$	& 2	&	1	&	2	&	1	&	2	\\
  \hline
\end{tabular}
\caption{Parameters of variants of the clustering approach}
\label{table:characteristics:free:variants}
\end{center}
\end{table}

Figure~\ref{figure:needle:free:clustering:variants} graphically depicts the needles obtained from each variant.
In some variants, such as Variant-4 in Figure~\ref{figure:needle:free:clustering:variant4}, it may appear that needles intersect at the same point. 
However, this is only due to the 2D display limitation of the needle configurations.
%and there is actually no intersection between any of the needles in 3D space. 
Different angles of view illustrate that needles do not intersect, as shown for Variant-4 in Figure~\ref{figure:needle:free:clustering:variant4:angles}.

\begin{figure}
\begin{subfigure}{.32\textwidth}
\centering   
\includegraphics[width=1\linewidth]{ needtest1}  
\captionsetup{justification=centering,font=scriptsize}
\caption{Variant-1}
\label{figure:needle:free:clustering:variant1}
\end{subfigure}
\begin{subfigure}{.32\textwidth}
\centering
\includegraphics[width=1\linewidth]{ 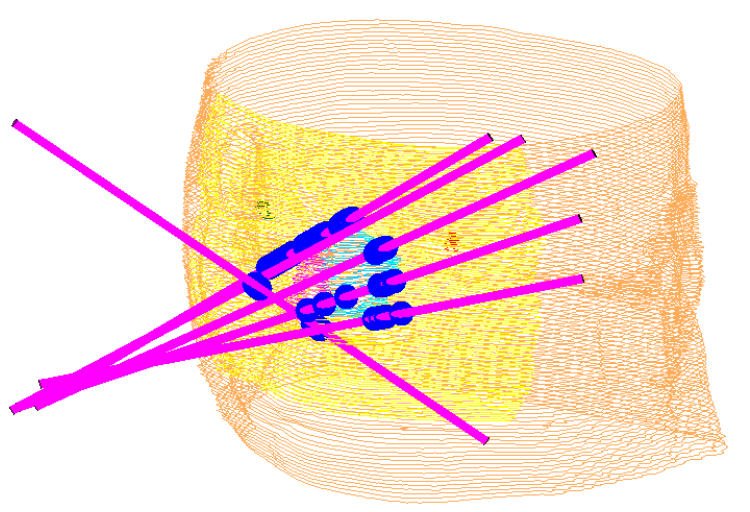}  
\captionsetup{justification=centering,font=scriptsize}
\caption{Variant-2}
\label{figure:needle:free:clustering:variant2}
\end{subfigure}
%\newline
\begin{subfigure}{.32\textwidth}
\centering
\includegraphics[width=1\linewidth]{ 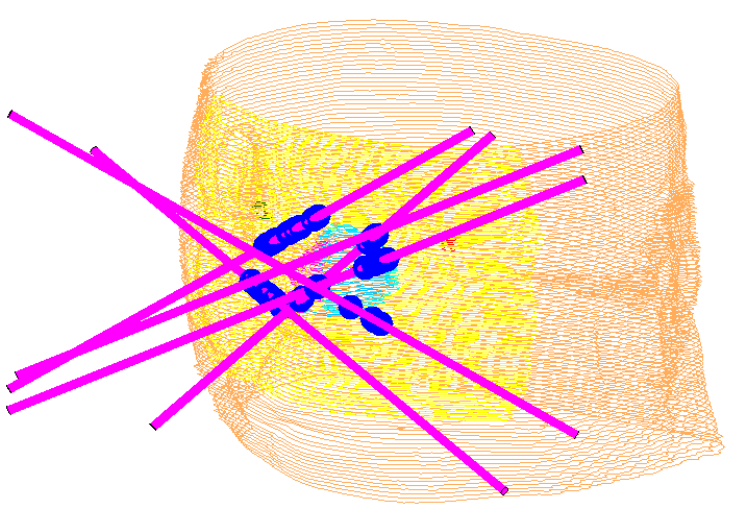} 
\captionsetup{justification=centering,font=scriptsize}
\caption{Variant-3}
\label{figure:needle:free:clustering:variant3}
\end{subfigure}
\begin{subfigure}{.32\textwidth}
\centering
\includegraphics[width=1\linewidth]{ 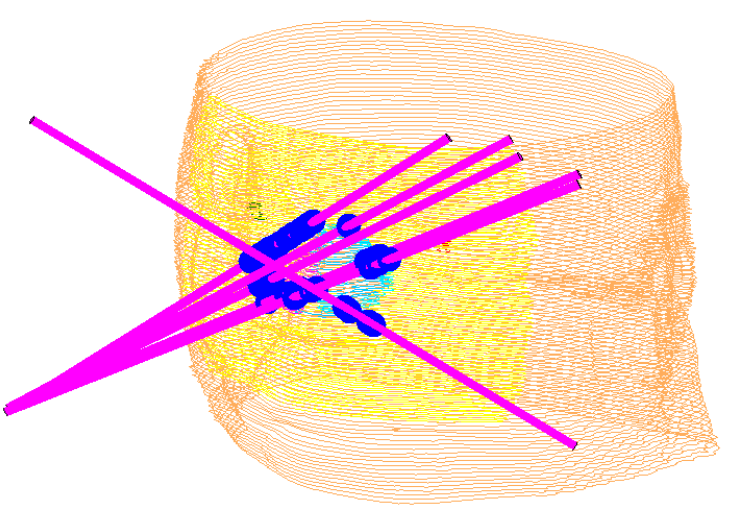}
\captionsetup{justification=centering,font=scriptsize}
\caption{Variant-4}
\label{figure:needle:free:clustering:variant4}
\end{subfigure}
%\newline
\begin{subfigure}{.32\textwidth}
\centering
\includegraphics[width=1\linewidth]{ 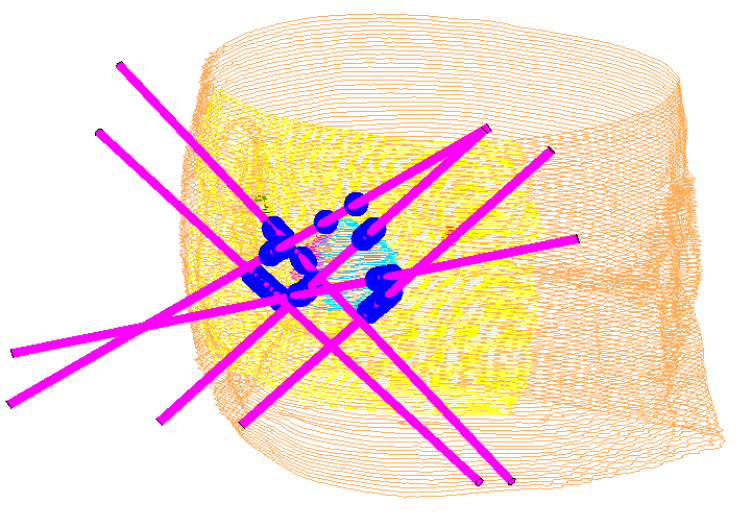} 
\captionsetup{justification=centering,font=scriptsize}
\caption{Variant-5}
\label{figure:needle:free:clustering:variant5}
\end{subfigure}
\begin{subfigure}{.32\textwidth}
\centering
\empty
\label{figure:needle:free:clustering:empty}
\end{subfigure}
\caption{Needle configurations for variants of the clustering approach}
\label{figure:needle:free:clustering:variants}
\end{figure}

\begin{figure}
\centering
\begin{subfigure}{.32\textwidth}
\centering
\includegraphics[width=1.1\linewidth]{ need5}
\label{figure:needle:free:clustering:variant4:angle1}
\end{subfigure}
\hspace*{0.5in}
\begin{subfigure}{.32\textwidth}
\centering
\includegraphics[width=1\linewidth]{ 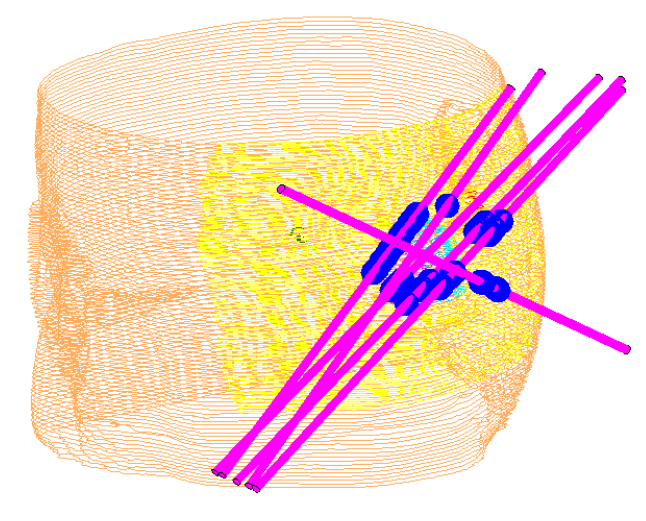} 
\label{figure:needle:free:clustering:variant4:angle2}
\end{subfigure}
\caption{Needle configuration in Variant-4 from different view angles}
\label{figure:needle:free:clustering:variant4:angles}
\end{figure}

Table~\ref{table:comparison:results:free:variants} lists the dosimetric indices for each variant.   
An analysis of this table reveals that all formulations produce similar improvements for most metrics, with only minor differences for the others. 
However, Variant-1 generates a solution much faster than the others. 
This is the reason that we decided to select $\mdla=``\times/="$, $\mdlb=``\times"$, and $\normp=2$ when evaluating the performance of the clustering approach in Section~\ref{section:results:free}.

\begin{table}[!htbp]
\begin{center}
\begin{tabular}{|c|c|c|c|c|c|c|c} 
 \hline
\backslashbox{Metric}{Method} & \multicolumn{1}{c|}{Variant-1} &  \multicolumn{1}{c|}{Variant-2}  & \multicolumn{1}{c|}{Variant-3} &  \multicolumn{1}{c|}{Variant-4} & \multicolumn{1}{c|}{Variant-5}  \\
 \hline

$LE_{V_{50}}$	&  0.00	&	0.00	&	0.00	&	0.00	&               	       0.00\\	
 \hline										
$RE_{V_{50}}$	&  0.02	&	0.00	&	0.00	&	0.00	& 	0.02\\	
 \hline										
$ST_{V_{50}}$	& 0.29	&	0.27	&	0.29	&	0.27	&	0.28	\\
 \hline										
$ST_{V_{100}}$	& 0.08	&	0.07	&	0.07	&	0.07	&	0.07	\\
 \hline										
$ST_{V_{150}}$	&0.00	&	0.00	&	0.00	&	0.00	&	0.00	\\
 \hline										
$ST_{V_{200}}$	& 0.00	&	0.00	&	0.00	&	0.00	&	0.00\\	
 \hline										
$SW_{V_{50}}$	&0.25	&	0.22	&	0.24	&	0.21	&      	0.23\\	
 \hline										
$SW_{V_{100}}$	&0.01	&	0.01	&	0.01	&	0.01	&	0.01\\	
 \hline										
$SW_{V_{150}}$	& 0.00	&	0.00	&	0.00	&	0.00	&          	                0.00                         \\	
 \hline										
$SW_{V_{200}}$	& 0.00	&	0.00	&	0.00	&	0.00	&	0.00	 \\
 \hline										
$RB_{V_{100}}$	&0.91	&	0.91	&	0.91	&	0.91	&                                     0.91  \\	
 \hline										
$RB_{V_{150}}$	&0.04	&	0.05	&	0.04	&	0.04	&	0.04	\\
 \hline										
$RB_{V_{200}}$	& 0.00	&	0.00	&	0.00	&	0.00	&	0.00	 \\
 \hline										
$RS_{V_{100}}$	& 0.98	&	0.98	&	0.98	&	0.98	&	0.98	\\
 \hline										
$RS_{V_{150}}$	& 0.00	&	0.02	&	0.00	&	0.02	&	0.00\\	
 \hline										
$RS_{V_{200}}$	&0.00	&	0.00	&	0.00	&	0.00	&	0.00	\\
 \hline										
$LB_{V_{100}}$	& 0.95	&	0.94	&	0.94	&	0.92	&	0.94\\	
 \hline										
$LB_{V_{150}}$	&0.04	&	0.05	&	0.01	&	0.03	&	0.02	\\
 \hline										
$LB_{V_{200}}$	&0.00	&	0.00	&	0.00	&	0.00	&	0.00	\\
 \hline										
$LS_{V_{100}}$	&0.97	&	0.94	&	0.96	&	0.96	&	0.97	\\
 \hline										
$LS_{V_{150}}$	& 0.04	&	0.04	&	0.04	&	0.03	&	0.06	\\
 \hline										
$LS_{V_{200}}$	&  0.00	&	0.00	&	0.00	&	0.00	& 	0.00 \\	
\hline										
  Solution time (s) & 46.20 &95.28	 &	153.33	&1238.48&	1584.37 \\		
\hline
\hline
\end{tabular}
\caption{Dosimetric indices and solution times for variants of the clustering approach}
\label{table:comparison:results:free:variants}
\end{center}
\end{table}

%\clearpage

%\paragraph{Approach 2 (A maximum coverage approach to determining needle positions)}\label{using maximum coverage-results}

\end{document}